\documentclass[smallcondensed, nospthms]{svjour3}

\usepackage{mathrsfs}
\usepackage{amsmath, amscd,amssymb, amsthm, amsfonts, thmtools, verbatim, breqn}
\usepackage[mathcal]{eucal}
\usepackage{hyperref}
\usepackage{relsize}
\usepackage[cmtip, arrow, all]{xy}
\usepackage{pb-diagram ,pb-xy}
\usepackage{tikz-cd}
\usepackage{graphicx}
\usepackage{subcaption}
\usepackage{epstopdf}
\usepackage{color}
\usepackage{ifthen}

\renewcommand{\epsilon}{\varepsilon}

\newcommand{\E}{\mscr{E}}

\newcommand{\mscr}{\mathscr}

\newcommand{\Z}{\mathbb Z}

\newcommand{\op}{\operatorname}

\newcommand{\mbb}{\mathbb}

\newcommand{\from}{\leftarrow}
\newcommand{\ip}[1][\cdot,\cdot]{\left\langle #1 \right\rangle}

\newcommand{\R}{\mbb R}

\renewcommand{\Im}{\op{Im}}

 \DeclareMathOperator{\End}{End}
 
\DeclareMathOperator{\Sym}{Sym} \DeclareMathOperator{\Hom}{Hom}

\DeclareMathOperator{\Tr}{Tr}

\def\cB{\mathcal B}
\def\cF{\mathcal F}

\def\cZ{\mathcal Z}

\def\sE{\mathscr E}\def\sF{\mathscr F}
\def\sL{\mathscr L}
\def\sM{\mathscr M}\def\sO{\mathscr O}
\def\sS{\mathscr S}
\def\sV{\mathscr V}\def\sW{\mathscr W}

\def\bL{\mathbf L}

\def\fJ(E){\mathfrak E}
\def\fJ{\mathfrak J}

\def\fT{\mathfrak T}

\def\fg{\mathfrak g}

\declaretheoremstyle[
spaceabove=7pt, spacebelow=7pt,
headfont=\normalfont\bfseries,
notefont=\mdseries, notebraces={(}{)},
bodyfont=\itshape,
postheadspace=5pt,
headpunct = .
]{thm}

\declaretheoremstyle[
spaceabove=7pt, spacebelow=7pt,
headfont=\normalfont\bfseries,
notefont=\mdseries, notebraces={(}{)},
bodyfont=\itshape,
postheadspace=10pt,
headpunct = .
]{def}

\declaretheoremstyle[
spaceabove=4pt, spacebelow=7pt,
headfont=\itshape,
postheadspace=5pt,
headpunct = :,
postheadspace = 3pt, qed = $\lozenge$
]{rem}

\declaretheorem[numbered = no, style = thm, name = Main Theorem]{thmmain}
\declaretheorem[numbered = no, style = thm, name = Proposition]{uproposition}
\declaretheorem[numbered = yes, parent = section, style = thm]{theorem}

\declaretheorem[numbered = no, style = thm, name = Corollary]{ucorollary}
\declaretheorem[numbered = no, style = thm, name = Corollary A]{corA}
\declaretheorem[numbered = no, style = thm, name = Corollary B]{corB}

\declaretheorem[sibling = theorem, style = thm, name = Theorem/Definition]{thm-def}
\declaretheorem[sibling = theorem, style = thm]{proposition}
\declaretheorem[sibling = theorem, style = thm]{lemma}

\declaretheorem[sibling = theorem, style = thm, name = Definition-Lemma]{deflem}

\declaretheorem[sibling = theorem, style = def]{definition}

\declaretheorem[sibling = theorem, style = rem]{remark}

\declaretheorem[sibling = theorem, style = rem]{example}

\newcommand{\cinfty}{C^{\infty}}

\newcommand{\Obstr}{\mathfrak {Obstr}}
\def\cot{T^*[-1]\E}

\newcommand\ind{\operatorname{ind}}
\newcommand\Str{\operatorname{Str}}
\newcommand\coker{\operatorname{coker}}

\newcommand\Obcl[1][~]{\ifthenelse{ \equal{#1}{~}} {
  \operatorname{Obs}^{cl}
}{
  \operatorname{Obs}^{cl}(#1)
}}
\newcommand\Obq[1][~]{\ifthenelse{ \equal{#1}{~}} {
  \operatorname{Obs}^{q}
}{
  \operatorname{Obs}^{q}(#1)
}}
\newcommand\Obqh{\operatorname{Obs}^q_{\hbar=1}}
\newcommand\Obqobstr{\operatorname{Obs}^q_{\hbar=1,\Obstr}}
\newcommand\Obqhtriv{\operatorname{Obs}^q_{\hbar=1,0}}
\renewcommand{\L}{\mathscr{L}}
\begin{document}


\title{A Mathematical Analysis of the Axial Anomaly}
\author{Eugene Rabinovich}
\institute{Department of Mathematics, University of California, Berkeley\\ \email{erabin@math.berkeley.edu}}
\date{}
\maketitle
\begin{abstract}
As is well known to physicists, the axial anomaly of the massless free fermion in Euclidean signature is given by the index of the corresponding Dirac operator. We use the Batalin-Vilkovisky (BV) formalism and the methods of equivariant quantization of Costello and Gwilliam to produce a new, mathematical derivation of this result. Using these methods, we formalize two conventional interpretations of the axial anomaly, the first as a violation of current conservation at the quantum level and the second as the obstruction to the existence of a well-defined fermionic partition function. Moreover, in the formalism of Costello and Gwilliam, anomalies are measured by cohomology classes in a certain obstruction-deformation complex. Our main result shows that---in the case of the axial symmetry---the relevant complex is quasi-isomorphic to the complex of de Rham forms of the spacetime manifold and that the anomaly corresponds to a top-degree cohomology class which is trivial if and only if the index of the corresponding Dirac operator is zero.  
\end{abstract}

\begin{acknowledgements}
The author would like to thank Owen Gwilliam and Brian Williams for many incredibly helpful discussions and immeasurable guidance in this project. He has Owen Gwilliam to thank for pointing him in the direction of this work. He would also like to thank his advisor Peter Teichner for inviting him to the Max Planck Institut f\"ur Mathematik in Bonn, for encouraging the author to speak in his (Peter's) student seminars, and for his comments on an earlier draft of this paper. Furthermore, he would like to thank Matthias Ludewig for many helpful comments on an earlier draft. Finally, thanks are due to the Max Planck Institut f\"ur Mathematik in Bonn for its hospitality while much of the work on this project was completed. 

This research was supported in part by Perimeter Institute for Theoretical Physics. Research at Perimeter Institute is supported by the Government of Canada thrgough the Department of Innovation, Science and Economic Development and by the Province of Ontario through
the Ministry of Research and Innovation. Some revisions of the present work resulted from the author's stay at the Perimeter Institute, and the author is thankful to Ben Albert and Kevin Costello for helpful comments. 

The author would like to thank an anonymous referee for numerous and helpful suggestions.

This material is based upon work supported by the National Science Foundation Graduate
Research Fellowship Program under Grant No. DGE 1752814. Any opinions,
findings, and conclusions or recommendations expressed in this material are those of the
author and do not necessarily reflect the views of the National Science Foundation.
\end{acknowledgements}

\tableofcontents

\section{Introduction}
\label{sec: intro}
\subsection{Background}
The axial anomaly is the failure of a certain classical symmetry of the massless free fermion to persist after quantization (see \cite{BellJackiw} for an original reference on the topic). It is well known to physicists that the axial anomaly is measured precisely by the index of the Dirac operator (see, e.g., Chapter 22.2 of \cite{wein} or \cite{NRS}); the aim of this paper is to prove this fact in a mathematically rigorous context for perturbative quantum field theory (QFT). Namely, we use the framework for the study of anomalies in Euclidean signature developed by Costello and Gwilliam in \cite{CG2}.

Given a generalized Dirac operator $D$ on a $\Z/2$-graded vector bundle $V\to M$ over a Riemannian manifold $M$, the \textit{massless free fermion} is described by the equation of motion 
\[
D\phi = 0,
\]
for $\phi$ a section of $V$. Consider the operator $\Gamma$ which is the identity on even sections of $V$ and minus the identity on odd sections of $V$, i.e. $\Gamma$ is the chirality involution. Then, $\Gamma$ anti-commutes with $D$, since $D$ is odd for the $\Z/2$-grading. It follows that if $D\phi = 0$, then \[D\left( \Gamma \phi\right)=0.\] In other words, the operator $\Gamma$ preserves the equations of motion, so generates a $U(1)$ symmetry of the classical theory. This symmetry is known as the \textit{axial symmetry}. One can ask whether this symmetry persists after the massless free fermion is quantized. This is in general not the case; the \textit{axial anomaly} measures the obstruction to the promotion of this classical symmetry to a quantum one. More generally, we say that a classical symmetry is anomalous if it does not persist after quantization. Anomalies in fermionic field theories have been an object of renewed recent interest (\cite{fermpathintegral} and \cite{freed1}), in part because of their relevance to topological phases of matter.

The question of what it means mathematically to quantize a field theory is still an open one. However, the formalism of Costello and Gwilliam (\cite{CG1}, \cite{CG2}, \cite{cost}) is one approach to the perturbative quantization of field theories which has been able to reproduce many properties of quantum field theory---especially those pertaining to the observables of quantum field theories---that physicists have long studied. This formalism naturally includes a framework for studying symmetries, and it is within this framework that we prove our results.

From another (heuristic) vantage point, a fermionic anomaly is an obstruction to the existence of a well-defined fermionic partition function, in the following sense. To compute the partition function of a quantum field theory involving fermionic fields, one performs a path integral over the space of all quantum fields. To do this, one first fixes a number of pieces of geometric and topological data, such as a manifold $M$, a metric on $M$, a spin structure on $M$, and so forth. Let $\cB'$ denote the ``space'' of all relevant pieces of geometric and topological data. The scare quotes are there because we might need a more abstract notion of space (such as a stack) to appropriately describe the situation. Given a point $x\in \cB'$, we can form the space of fields, which is a product $\cB_x\times \cF_x$ of the fermionic fields $\cF_x$ and the non-fermionic fields $\cB_x$ at $x$. The non-fermionic fields $\cB_x$ may include a gauge field or a scalar field with which the fermionic fields interact. As $x$ varies, these spaces of fields fit together into something like a product $\cB\times \cF$ of fibrations over $\cB'$. In a given fiber $\cB_x\times \cF_x$, one can choose to perform the path integral by first integrating over $\cF_x$, choosing a fixed $b\in \cB_x$. Because the fermionic fields are coupled to the fields in $\cB_x$, the fermionic integration produces something which varies over $\cB$, i.e. depends both on the fixed non-fermionic field $b\in\cB_x$ and the point $x\in \cB'$. However, the fermionic integration does not canonically produce a number in general; it produces an element of a (super-)line depending on the background data. In other words, there is a line bundle 
\[
L \to \cB.
\]
This is called the \textit{determinant line bundle,} because its fiber over a point $y\in \cB$ is the natural home for the determinant of the Dirac operator corresponding to the massless free fermion theory with background data encoded in $y$. The fermionic path integral produces a section $\sigma$ of this bundle, which deserves to be called ``the partition function'' only when $L$ is endowed with a trivialization, for in that case, one can compare $\sigma$ to this trivialization to get an honest function.

There is a well-developed mathematical literature addressing many of the relevant issues in the case where one chooses $\cB=\cB'$ and restricts attention to a subspace of $\cB'$ which is an actual space (see \cite{bf1}, \cite{bf2}, \cite{freed2}). These ideas make maneuvers of physicists (e.g. zeta-function regularization) precise in the special context of families of massless free fermionic field theories. In the present article, we allow $\cB\to \cB'$ to be non-trivial, fix $x\in \cB$ (and therefore a manifold $M$ on which the theory lives), and study the dependence of $L$ only over the fiber $\cB_x$. In other words, we study the pullback
\[
\begin{tikzcd}
L'\arrow[r]\arrow[d] \arrow[rd, phantom, "\lrcorner", at start] & L\arrow[d]\\
\cB_x \arrow[r,hook] & \cB
\end{tikzcd}.
\] 
We will let $\cB_x$ be a space of perturbative gauge fields, i.e. $\cB_x$ will be the space of background gauge fields for the axial symmetry.

We note that the present work differs from the mathematical treatments in the vein of \cite{freed2} in two ways. The first is that the ``space'' $\cB_x$ is best understood as an object of formal derived deformation theory (\cite{DAGX}). In \cite{CG2}, the authors explain how perturbative field theory can be described with the language of formal derived deformation theory; the essence of this relationship is that, for the purposes of quantum field theory, formal derived deformation theory provides a nice unification of perturbation theory and gauge theory. In the context of symmetries of quantum field theory, the approach of \cite{CG2} amounts essentially to turning on a background gauge field for the symmetry in question and then studying the gauge-invariance of the partition function as a function of this background gauge field. In the present work, we will mainly use formal derived deformation theory for the geometric intuition it provides. In particular, we will prove homological algebraic results about certain modules over a commutative differential-graded algebra, and intuitions from formal derived deformation theory will allow us to interpret these results as related to the heuristic discussion above.

The second way in which our work differs from existing treatments of anomalies is that we focus on a general approach to field theories and their quantization known as the Batalin-Vilkovisky (BV) formalism. In the work of Costello and Gwilliam, this formalism is used to perform the quantization of the observables of field theories and their symmetries. Our work shows that, simply by ``turning the crank on the machine'' of BV quantization, one naturally recovers the axial anomaly. In other words, we study the axial anomaly using a framework that was not hand-made for the case of the massless free fermion.

We note also that we work entirely in Euclidean signature; for a mathematical discussion of the axial anomaly in Lorentzian signature, see \cite{BS}, which also provides a historical survey of the topic.

\subsection{Presentation of Main Results}
\label{subsec: mainresults}
This paper has two main results. First, we compute the equivariant quantum observables of a massless free fermion with a symmetry. Second, we compute the axial anomaly of the massless free fermion and show that it is given by the index of the associated Dirac operator. 

Costello and Gwilliam provide definitions of the notions of free quantum field theory and of symmetries of such theories. In their formalism, a symmetry is effected by the action of a differential-graded Lie algebra (dgla) $\sL$ on the space of fields, and one can study the \textit{obstruction} to the quantization of this symmetry, as well as the cochain complex of \textit{equivariant quantum observables}, which forms a differential-graded module for the commutative differential-graded algebra $C^\bullet(\sL)$ (the Chevalley-Eilenberg cochains of $\sL$). Our first main result is the following proposition:

\begin{uproposition}[Cf. Proposition \ref{prop: line}, Lemma 
\ref{lem: obstrtrivbund}]
Let $\sL$ be a dgla acting on the massless free fermion by symmetries, $C^\bullet(\sL)$ be the corresponding Chevalley-Eilenberg cochain complex, and $\sO(\sL[1])$ the underlying graded vector space of $C^\bullet(\sL)$. The equivariant quantum observables of the massless free fermion with action of $\sL$ are quasi-isomorphic to a $C^\bullet(\sL)$-module $P$ whose underlying graded vector space is $\sO(\sL[1])$. This module is isomorphic to the trivial such module if and only if the obstruction is trivial.
\end{uproposition}

In Section \ref{sec: line}, we will comment on the interpretation of this proposition as an instantiation of the heuristic discussion of the line bundle $L'\to\cB$ in the previous subsection.

In the framework of Costello and Gwilliam, there is an \textit{obstruction-deformation complex} $C^\bullet_{red, loc}(\sL)$ and the obstruction is a cohomology class $[\Obstr]$ in this complex (see Definition-Lemma \ref{def: tobstr} and Lemma \ref{lem: obstr}). The obstruction is what physicists would call an ``anomaly'', and we will use the two terms more or less interchangeably.

Our aim is to compute the anomaly, at least for certain actions of a dgla $\sL$ on the massless free fermion. To accomplish this, one needs to understand the structure of $C^\bullet_{loc, red}(\sL)$. For the case of the axial symmetry, which corresponds to the  dgla $\sL=\Omega^\bullet_M$, we show the following two results:

\begin{uproposition}[Cf. Proposition \ref{prop: qism}]
\label{prop: introqism}
If $M$ is oriented, there is a canonical quasi-isomorphism 
\[
\Phi: \Omega^\bullet_M[n-1]\hookrightarrow C^\bullet_{loc, red}(\Omega^\bullet_M)
\]
of complexes of sheaves on $M$. Here $\Omega^\bullet_M$ is the abelian elliptic dgla encoding the axial symmetry of the massless free fermion theory.
\end{uproposition}

The assumption of orientability is not necessary for a version of the above proposition to be true; see Remark \ref{rem: twisteddR}. 

\begin{thmmain}[Cf. Theorem \ref{thm: str}]
\label{thm: introstr}
The cohomology class of $\Obstr$ is equal to the cohomology class of 
\[
\Phi\left((-1)^{n+1}2\frac{\ind{D}}{\text{vol}(M)}dVol_g\right),
\]
where $\Phi$ is the quasi-isomorphism of the Proposition and $dVol_g$ is the Riemannian volume form on $(M,g)$. 
\end{thmmain}

In fact, we will have very similar theorems in case a Lie algebra $\fg$ acts on the fermions in a way that commutes with the corresponding Dirac operator. We prove the corresponding generalizations of the above results in Section \ref{sec: eq}.

We note that the obstruction-theoretic approach to anomalies is known also to physicists: see, e.g. Chapter 22.6 of \cite{wein}. In the present work, we give a mathematically precise version of this approach; in particular, we use the renormalization techniques of \cite{cost}. 

Readers with a background in physics are likely to be more familiar with the perspective on anomalies as violations of current conservation at the quantum level. This perspective is essentially dual to the one we highlight here; in fact, the formalism of Costello and Gwilliam in Chapter 12 of \cite{CG2} incorporates both perspectives---albeit at a very abstract level---and the authors freely move between the two. We comment on the current-conservation perspective in Section \ref{subsec: fujikawa}. 

Let us summarize. In this paper, we focus on two perspectives on anomalies in fermionic theories. In the first, anomalies are obstructions to the construction of a well-defined fermionic partition function, which perspective will be justified by Lemma \ref{lem: obstrtrivbund}. In the second perspective, anomalies are the obstruction to the persistence of a symmetry after quantization. This perspective is justified by the arguments of Section \ref{subsec: fujikawa}, using the obstruction theoretic result of the Main Theorem.

\subsection{Future Directions}
In this subsection, we outline a few natural extensions of the present work.

One of the novel realizations of Costello and Gwilliam is that the observables of a quantum field theory fit together into a local-to-global (on the spacetime manifold), cosheaf-like object known as a factorization algebra. In this paper, we use only the global observables on $M$. However, the factorization-algebraic structure of observables often reproduces familiar algebraic gadgets which mix geometry and algebra: \textit{locally constant} factorization algebras, as shown by Lurie (\cite{higheralgebra}), are equivalent to $E_n$-algebras, and \textit{holomorphic} factorization algebras produce vertex algebras (See chapter 5 of \cite{CG1} for the general construction, and \cite{ggw} for an example of the construction). It would be interesting to understand the analogous sort of structure in the case at hand.

Another interesting local-to-global aspect of field theories in the formalism of Costello and Gwilliam is that, given a theory on a manifold $M$, the formalism naturally produces a sheaf of theories on $M$. Moreover, the construction of the massless free fermion on $M$ depends only a metric, a $\Z/2$-graded vector bundle on $M$, and a generalized Dirac operator $D$. All three objects are local in nature, so that it is to be expected that the massless free fermion can be extended to a sheaf of theories on an appropriate site of smooth manifolds (i.e., manfiolds equipped with a metric, $\Z/2$-graded vector bundle, and Dirac operator). We would like to give sheaf-theoretic extensions of the main theorems of this paper.

In yet another direction, we would like to use the BV formalism to make analogous constructions to those in \cite{freed2}. Namely, in the case where we have a family of Dirac operators parametrized by a manifold $B$, we expect the BV formalism to produce a line bundle over $B$ just as we saw in the first subsection. We hope to give a BV-inspired construction of a metric and compatible connection, which can be used to probe the line bundle.

Finally, in the case that the anomaly does vanish, we would like to compute the partition function.

\subsection{Plan of the Paper}
The plan for the rest of the paper is as follows. In the next section, we provide a lightning review of basic concepts from the theory of generalized Dirac operators sufficient to define the theory of massless free fermions in the BV formalism. In Section \ref{sec: ferm}, we introduce the main example of the massless free fermion and its axial symmetry. This section relies heavily on the techniques of Costello and Gwilliam, the relevant background on which we review in the appendix. In Section \ref{sec: lemmas}, we discuss a few lemmas concerning BV quantization which appear implicitly in some form or other in the literature, but whose statements and proofs we make explicit in the present context. Next, in Section \ref{sec: line}, we prove the first proposition from Subsection \ref{subsec: mainresults}. In Section \ref{sec: main} (Subsections \ref{subsec: main} and \ref{subsec: comps}) , we show that the index of the Dirac operator completely characterizes the anomaly by proving more precise versions of the second Proposition above and the Main Theorem. Also in Section \ref{sec: main} (Subsection \ref{subsec: fujikawa}), we bring our discussion closer to the physics literature by describing the relationship between the axial anomaly and the existence of conserved currents in the massless free fermion quantum field theory. Finally, in Section \ref{sec: eq}, we prove equivariant generalizations of the preceding results.

\subsection{Notation and Conventions} 
\label{subsec: not}
We assemble here many conventions and notations that we use throughout the remainder of the text. Some may be explained where they first appear, but we believe it to be useful to the reader to collect them in one place.

\begin{itemize}
\item The references we cite are rarely the original or definitive treatments of the subject; they are merely those which most directly inspired and informed the work at hand. We have made every attempt to indicate which results are taken or modified from another source.
\item The free fermion theory will be formulated on a Riemannian manifold $M$. We will assume $M$ to be connected for simplicity, though our results are easily extended to non-connected $M$. 
\item If a Latin letter in standard formatting is used to denote a vector bundle, e.g. $E$, then that same letter in script formatting, e.g. $\sE$, is used to denote its sheaf of sections. \textbf{Letters in script formatting, except $\sO$, will always denote sheaves of smooth sections of vector bundles.} 
\item We use the notation $\cong$ for isomorphisms and $\simeq$ for quasi-isomorphisms (weak equivalences). 
\item We use the notation $\cinfty_M$ and $\Omega^k_M$ for the sheaves of smooth functions and smooth $k$-forms, respectively, on $M$. We use the notation $\underline{\R}$, $\Lambda^k T^*M$ for the corresponding vector bundles.
\item If $V$ is a finite-dimensional vector space, $V^\vee$ is the linear dual to $V$; if $V$ is a topological vector space, then $V^\vee$ is the continuous linear dual. If $V\to M$ is a finite-rank vector bundle, then $V^\vee \to M$ is the fiberwise dual to $V$. 
\item If $V_1\to M_1$ and $V_2\to M_2$ are vector bundles and $M_3$ is a manifold, then we denote by $V_1\boxtimes V_2$ the bundle $p_1^*(V_1)\otimes p_1^*(V_2)\to M_1\times M_2\times M_3$, where $p_1: M_1\times M_2\times M_3 \to M_1$ and $p_2: M_1\times M_2\times M_3\to M_2$ are the canonical projections. If $M_3$ is not explicitly mentioned, it is assumed to be $pt$.
\item We will use $\otimes$ to denote the completed projective tensor product of topological vector spaces as well as the algebraic tensor product of finite-dimensional vector spaces. The completed projective tensor product has the nice characterization that if $\sV(M)$ and $\sW(N)$ are vector spaces of global sections of vector bundles $V\to M$ and $W\to N$, then $\sV(M)\otimes \sW(N)$ is the space of global sections of the bundle $V\boxtimes W$ over $M\times N$.
\item If $V$ is a $\Z\times \Z/2$-graded vector space and $v$ is a homogeneous element of $v$, then $|v|$ and $\pi_v$ refer to the $\Z$ and $\Z/2$-degrees of $v$, respectively. We will also call $|v|$ the \textit{ghost number} of $v$ and $\pi_v$ the \textit{statistics} of $v$. This gives a convenient terminology for distinguishing between the various types of grading our objects will have; it is a terminology that conforms loosely to that of physics. We will sometimes also refer to ghost number as cohomological degree.
\item If $V\to M$ is a vector bundle, then $V^!$ is the bundle $V^\vee\otimes Dens_M$, where $Dens_M$ is the density bundle of $M$. $\sV^!$ will denote the sheaf of sections of $V^!$.
\item If $V$ is a $\Z$-graded vector space, then $V[k]$ is the $\Z$-graded vector space whose $i$-th homogeneous space is $V_{i+k}$. In other words, $V[k]$ is $V$ shifted \textit{down} by $k$ ``slots''.
\item If $M$ is compact and Riemannian with metric $g$, then for a one-form $\alpha$, 
\[
|\alpha|^2 : = g^{-1}(\alpha, \alpha)\in \cinfty(M).
\]
Here, $g^{-1}$ is the metric on $T^*M$ induced from $g$. There is a similar definition if $\alpha$ is a section of a vector bundle $V\to M$ with metric $(\cdot, \cdot)$.
\item If $V$ is a $\Z\times \Z/2$-graded vector space, then $Bilin(V)$ is the $\Z\times \Z/2$-graded vector space of bilinear maps $V\otimes V\to V$.
\item When dealing with $\Z\times \Z/2$-graded objects, we will always say ``graded (anti)-symmetric'' if we wish Koszul signs to be taken into account. If we simply say (anti)-symmetric, we mean that this notion is to be interpreted without taking Koszul signs into account.
\item If $v\in  \Sym(V)$ for some vector space $V$, then $v^{(r)}$ denotes the component of $v$ in $\Sym^r(V)$.
\item Let $V$ be a (non-graded) vector space and $\phi: V\to V$ a linear map. We denote by $\phi_{0\to1}$ the cohomological degree 1 operator 
\[
V \overset{\phi}{\longrightarrow} V[-1],
\]
and similarly for $\phi_{1\to 0}$. We denote by $\phi_{0\to 0}$ the cohomological degree 0 operator on $V\oplus V[-1]$ which acts by $\phi$ on $V$ and by 0 on $V[-1]$; similarly, $\phi_{1\to 1}$ is the operator on $V\oplus V[-1]$ which acts by 0 on $V$ and by $\phi$ on $V[-1]$.
\item If $\sL$ is an elliptic dgla, then we will use $d_{\sL}$ to denote the differential on the complex of Chevalley-Eilenberg cochains of $\sL$.
\end{itemize}


\section{Generalized Laplacians, Heat Kernels, and Dirac Operators}
\label{sec: review}
We present here a list of definitions and results relevant to our work, taking most of these from \cite{bgv}. Throughout, $(M,g)$ is a closed Riemannian manifold of dimension $n$ with Riemannian volume form $dVol_g$. We let $V\to M$ be a $\Z/2$-graded vector bundle with $V^+,V^-$ the even and odd bundles, respectively. We let $\mathscr V$ be the sheaf of smooth sections of $V$, and similarly $\sV^\pm$ the sheaf of smooth sections of $V^\pm$. We always use normal-font letters for vector bundles and script letters for the sheaves of sections of the corresponding vector bundles.


\begin{definition}
 A \textbf{$\Z/2$-graded metric bundle} is a metric bundle $V\to M$ with an orthogonal decomposition $V=V^+\oplus V^-$.
\end{definition}


\begin{definition}[Proposition 2.3 of \cite{bgv}]
	Let $V$ be a $\Z/2$-graded metric bundle. A \textbf{generalized Laplacian} on $V$ is a differential operator
	\[
	H:\mathscr V \to \mathscr V
	\]
	such that 
	\[
	[[H,f],f] = -2|df|^2,
	\]
	where we are thinking of smooth functions as operators which act on sections of $V$ by pointwise scalar multiplication .
	
	We say that $H$ is \textbf{formally self-adjoint} if for all $s,r\in \mathscr V$, 
	\[
	\int_M (s,Hr)dVol_g = \int_M(Hs,r)dVol_g.
	\]
	
\end{definition}

\begin{remark}
The equation defining a generalized Laplacian is a coordinate-free way of saying that $H$ is a second-order differential operator on $V$ whose principal symbol is just the metric, i.e., that in local coordinates, $H$ looks like $g^{ij} \partial_i \partial_j+\text{lower-order derivatives}$.
\end{remark}
\begin{remark}
We did not strictly need $V$ to be $\Z/2$-graded to define a generalized Laplacian, but since we only ever work in this case, we do not use this slightly different definition.
\end{remark}


\begin{definition}[Definition 3.36 of \cite{bgv}]
	A \textbf{Dirac operator} on $V$ is an odd differential operator
	\[
	D: \mathscr V^\pm \to \mathscr V^\mp
	\]
	such that $D^2$ is a generalized Laplacian. The definition of formal self-adjointness makes sense also for Dirac operators.
\end{definition}

\begin{example}
\label{ex: euler}
Let $V = \Lambda^\bullet T^*M$, with the $\Z/2$-grading given by the form degree mod 2. Then the operator $d+d^*$ is a generalized Dirac operator since its square $(d+d^*)^2$ is the Laplace-Beltrami operator. Here, $d^*$ is the formal adjoint of $d$, characterized by
\[
\int_M \left( \alpha, d\beta\right) dVol_g = \int_M \left( d^*\alpha, \beta\right) dVol_g.
\]
It is clear from this description that $d+d^*$ is formally self-adjoint, since it is the sum of an operator with its formal adjoint. Via Hodge theory, the kernel of $d+d^*$ is canonically identified with the de Rham cohomology $H^\bullet_{dR}(M)$.
\end{example}

There are numerous other examples of interest---for example, the traditional Dirac operator on an even-dimensional spin manifold and $\sqrt{2}\left(\bar\partial +\bar\partial^*\right)$ on a K\"ahler manifold---and our formalism will work for these as well. We refer the reader to Section 3.6 of \cite{bgv} for the details of the construction of these operators.


\begin{deflem}[Proposition 3.38 of \cite{bgv}]
\label{deflem: cliffaction}
If $V$ is a $\Z/2$ graded bundle over $(M,g)$ a Riemannian manifold and $D$ a Dirac operator on $V$, then \textbf{the Clifford action of 1-forms on $V$} 
\[
c: \Omega^1(M) \to \End^{odd}(V)
\]
is defined on exact forms by 
\begin{equation}
\label{eq: cliffaction}
c(df) = [D,f],
\end{equation}
and extended to all 1-forms by $\cinfty$-linearity. This action is well-defined.
\end{deflem}


\begin{theorem}[The Heat Kernel; Section 8.2.3 of \cite{CG2}]
\label{thm: heatkernel}
	Let $V$ be a $\Z/2$-graded metric bundle with metric $(\cdot, \cdot)$ and Dirac operator $D$. Write $H:=D^2$ for the generalized Laplacian corresponding to $D$. Then there is a unique \textbf{heat kernel} $k \in \Gamma(M\times M,V\boxtimes V)\otimes \cinfty(\R_{>0})$ satisfying: 
	\begin{enumerate}
		\item 
		\[
		\frac{d}{dt}k_t + (H\otimes 1)k_t = 0,
		\]
		where $k_t\in \Gamma(M\times M,V\boxtimes V)$ is the image of $k$ under the map \[id\otimes ev_t:\Gamma(M\times M,V\boxtimes V)\otimes \cinfty(\R_{>0})\to \Gamma(M\times M,V\boxtimes V)\] induced from the evaluation of functions at $t\in \R_{>0}$.
		\item For $s\in \sV$,
		\[
		\lim_{t\to 0}\int_{y\in M} id\otimes (\cdot,\cdot)(k_t(x,y)\otimes s(y)) dVol_g(y) = s(x),
		\]
where the limit is uniform over $M$ and is taken with respect to some norm on $V$.
	\end{enumerate}
	The scale $t$ heat kernel $k_t$ is the integral kernel of the operator $e^{-tH}$ in the sense that 
	\[
	\int_{y\in M}id\otimes (\cdot, \cdot)(k_t(x,y)\otimes s(y))dVol_g(y)= (e^{-tH}s)(x).
	\]
\end{theorem}


\begin{definition}[Proposition 3.48 of \cite{bgv}]
	Let $D^\pm$ denote $D\mid_{\sV^\pm}$ for $D$ a formally self-adjoint Dirac operator $D$. The \textbf{index} $\ind(D)$ of $D$ is $\dim(\ker(D^+))-\dim(\coker(D^+))$. By the formal self-adjointness of $D$, we have also the alternative formula \[\ind(D) = \dim(\ker(D^+))-\dim(\ker(D^-))=:\text{sdim} (\ker D)\]
for the index of $D$.
\end{definition}

\begin{example}
The index of the Dirac operator from Example \ref{ex: euler} is 
\[
\sum_{i=0}^n (-1)^i \dim (\ker (d+d^*)\mid_{\Omega^i_{dR}})= \sum_{i=0}^n (-1)^i \dim H^i_{dR}(M)=\chi(M),
\]
where $\chi(M)$ is the Euler characteristic of $M$.
\end{example}


The following definition can be found in the discussion preceding Proposition 1.31 of \cite{bgv}
\begin{definition}
	If $\phi: V\to V$ is a grading-preserving endomorphism of the super-vector space $V$, then the \textbf{supertrace of $\phi$} is
	\[
	\Str(\phi) := \Tr(\phi\mid_{V^+})-\Tr(\phi\mid_{V^-}).
	\]
Let $\Gamma: V\to V$ denote the operator which is the identity on $V^+$ and minus the identity on $V^-$. Then $\Str(\phi) = \Tr(\phi \Gamma)$.
\end{definition}


With these definitions in place, we can finally state the following 
\begin{theorem}[Theorem 3.50 of \cite{bgv}]
\label{thm: mcs}
Let $V$ be a Hermitian, $\Z/2$-graded vector bundle on a closed Riemannian manifold $M$, with $dVol_g$ the Riemannian volume form on $M$. Let $D$ be a formally self-adjoint Dirac operator on $V$, with $k_t$ the heat kernel of $D^2$. Then
	\begin{equation}
		\label{eq: mcs}
		\ind(D)=\int_M \Str(k_t(x,x)) dVol_g(x)=\Str(e^{-tD^2}),
	\end{equation}
	where the last trace is an $L^2$ trace of the bounded operator $e^{-tD^2}$.
\end{theorem} 
 This is the \textbf{McKean-Singer theorem}.

We will also need a few technical results that will be useful especially in Section \ref{sec: eq}. The first is 

\begin{lemma}[Proposition 3.48 of \cite{bgv}]
\label{lem: hodge}
If $D$ is a formally self-adjoint Dirac operator over a compact manifold, then 
\[
\sV^{\pm} = \ker(D^\pm)\oplus \Im(D^\mp),
\]
where the decomposition is orthogonal with respect to the metric on $\sV$ induced from $(\cdot, \cdot)$.
\end{lemma}

\begin{remark}
This is a special case of a general Hodge decomposition that applies for any elliptic complex with a metric on the corresponding bundle of sections.
\end{remark}

The second result is Proposition 2.37 of \cite{bgv}, as used in the proof of the McKean-Singer Theorem.
\begin{lemma}
\label{lem: infraredkernel}
If $D$ is a formally self-adjoint Dirac operator, and if $P_{\pm}$ is the orthogonal projection onto  $\ker(D^\pm)$, then for $t$ large,
\[
\left |\Tr(\left. e^{-tD^2}\right|_{\sV^\pm}-P_\pm)\right| \leq C\cdot e^{-t\lambda_1/2},
\]
where $\lambda_1$ is the smallest non-zero eigenvalue of $D^2$ and $C$ is a constant (i.e. independent of $t$).
\end{lemma}

Finally, we present useful analytic facts about the spectrum of a Laplacian on a compact manifold.
\begin{lemma}[Cf. Proposition 2.36 of \cite{bgv}]
\label{lem: spectrum}
If $H$ is a formally self-adjoint generalized Laplacian on a compact manifold, then $H$ has a unique self-adjoint extension $\bar H$ with domain a subspace of $L^2$ sections of $V$; $\bar H$ has discrete spectrum, bounded below, and each eigenspace is finite dimensional, with smooth eigenvectors. Letting $\phi_j$ be an eigenvector for $\bar H$ with eigenvalue $\lambda_j$, we have
\[
k_t=\sum_{j} e^{-t\lambda_j} \phi_j\otimes \phi_j.
\]
Moreover, $e^{-t\bar H}$ is compact (and therefore bounded and trace-class).
\end{lemma}

\section{The Massless Free Fermion and its Axial Symmetry}
\label{sec: ferm}
In this section, we introduce the main examples which will concern us in the rest of the paper, namely the example of the massless free fermion and its axial symmetry. We use the formalism of (equivariant) Batalin-Vilkovisky (BV) quantization and renormalization, as developed in \cite{cost,CG2,CG1}. The appendix consists of a self-contained summary of all the relevant background material; we encourage the reader to peruse it as necessary. To the reader already familiar with this material, we note that in our work, we allow the space of fields of a free field theory and all objects built from this space of fields to possess an extra $\Z/2$-grading which tracks particle statistics---fermionic or bosonic, corresponding to odd or even $\Z/2$-grading, respectively. In this convention for the usage of the term ``statistics'', ghosts, anti-fields, and anti-ghosts for fermionic fields have fermionic statistics. Of course, fermionic ghosts and anti-fields are mutually commuting, so our usage of the term ``statistics'' differs from the conventional physics usage.

In the formalism of \cite{cost,CG1}, a free BV theory is specified by a triple $(\sF,Q,\ip_{loc})$, where $(\sF,Q)$ is an elliptic complex on a manifold $N$ and $\ip$ is a degree --1 pairing $F\otimes F\to \text{Dens}_N$. The triple is required to satisfy some conditions elaborated in the appendix. It is also possible to define what is meant by a gauge-fixing of a free BV theory and an action of an elliptic differential graded Lie algebra (dgla) $\sL$ on a free BV theory $\sF$. Given such an action, one can define cochain complexes $\Obcl$ and $\Obq{}[t]$ of equivariant classical observables and equivariant scale-$t$ quantum observables for the action. Moreover, there is a scale-$t$ obstruction $\Obstr[t]$ and its $t\to 0$ limit exists and is denoted $\Obstr$. The aim of this section is to present an example of a free BV theory with an action of an elliptic dgla, which we call the axial symmetry of the massless free fermion. The aim of the remainder of the paper will be to better understand the structure of $\Obq{}[t]$ and $\Obstr$.

We will show that a formally self-adjoint Dirac operator defines a free BV theory; first, though, we need to establish some notation. Let $(M,g)$ be, as in Section \ref{sec: review}, a Riemannian manifold with Riemannian density $dVol_g$, $V$ a $\Z/2$-graded metric bundle on $M$ with metric $(\cdot, \cdot)$, and $D$ a formally self-adjoint Dirac operator. Define $\Gamma: \sV\to \sV$ to be $-id_{\sV^-}\oplus id_{\sV^+}$. Let \[S:= V\oplus V[-1],\] with purely odd $\Z/2$ grading (and let $\sS$ denote the sheaf of sections of $S$), and
\[
\ip[f_1, f_2]_{loc} =(f_1,f_2)dVol_g
\]
when either $|f_1|=0$ and $|f_2|=1$ or vice versa. As described in the Notation and Conventions subsection of Section \ref{sec: intro}, we let $D_{0\to 1}$ denote the operator of degree 1 on $\sS$ which is simply given by $D$, and similarly for $D_{0\from1}$, $\Gamma_{0\to0}$, etc. With all this notation in place, we can state and prove the following 

\begin{lemma}
$(\sS,D_{0\to1}\Gamma_{0\to0}, \ip_{loc})$ is a free BV theory, which we call the \textbf{massless free fermion},  and $\Gamma_{0\to 0} D_{1\to 0}$ is a gauge-fixing for this theory. 
\end{lemma}

\begin{proof}
Let us first note the following facts: since $D$ reverses the $\Z/2$-grading of $\sV$, $D\Gamma = -\Gamma D$ (hence $D_{0\to1}\Gamma_{0\to0}=-\Gamma_{1\to1}D_{0\to1})$; moreover, $\Gamma$ is self-adjoint for the pairing $(\cdot, \cdot)$.

$D_{0\to1}\Gamma_{0\to0}$ and $\Gamma_{0\to 0} D_{1\to 0}$ can be checked directly to have the required symmetry properties with respect to $\ip$ as a consequence of the above facts and the formal self-adjointness of $D$.

$\ip$ is symmetric (in the brute, non-graded, sense) because of the symmetry of $(\cdot, \cdot)$; this is consistent with Equation \ref{eq: ipsym} since all fields have fermionic statistics, i.e. $\ip$ is graded skew-symmetric.

It remains only to check that $[Q,Q^{GF}]$ is a generalized Laplacian. Using the first fact noted at the beginning of the proof, and the fact that $\Gamma^2=1$, it is a direct computation to show that $[Q,Q^{GF}]=D^2_{0\to0}+D^2_{1\to1}$, so is in fact a generalized Laplacian. 
\end{proof}

\begin{remark}
The metric on $V$ can be used to identify $V$ with $V^\vee$, and the Riemannian density can be used to trivialize the bundle of densities on $M$, so that $V^!\cong V$. As a result, it can be shown that the the free fermion theory is isomorphic, under the obvious notion of isomorphism of free theories, to the cotangent theory to the elliptic complex
\[
\sV^+ \overset{D^+}{\longrightarrow} \sV^-.
\]
We will use this characterization in the sequel.
\end{remark}

In the appendix, we define what it means for an elliptic dgla to act on a free BV theory. The next lemma provides an action of the elliptic dgla $\Omega^\bullet_M$ on the massless free fermion.


\begin{lemma}
\label{lem: actionofforms}
Let $\sL_{\R} = \Omega^\bullet_{dR,M}$, equipped with the de Rham differential and trivial Lie bracket. If $f\in \cinfty(M)$, we let 
\begin{equation}
\label{eq: 0formaction}
[f,\cdot]:= f \Gamma_{0\to0}- f\Gamma_{1\to1}.
\end{equation}
If $\alpha\in \Omega^1(M)$, let 
\begin{equation}
\label{eq: 1formaction}
[\alpha, \cdot] = c(\alpha)_{0\to1},
\end{equation}
where $c(\alpha)$ is the Clifford action of forms on sections of $V$. For $\beta$ a form of degree greater than 1, let $[\beta,\cdot]=0$. Then,  these equations define an action of $\sL_\R$ on the massless free fermion. This action is known as the \textbf{axial symmetry} in physics.
\end{lemma}
\begin{proof}
Let us check the derivation property of a dgla: 
\[
Q[f,\varphi] = [df,\varphi]+[f,Q\varphi].
\]
Assuming that $|\varphi|=0$ (otherwise both sides of the above equation are 0), we find that the right hand side of the above equation is 
\[
[df, \varphi] - f \Gamma_{1\to1} D_{0\to1}\Gamma_{0\to0} \varphi= D_{0\to1}f\varphi,
\]
which agrees with $Q[f,\varphi]$.

The remaining checks---namely of the Jacobi equation and the invariance of $\ip$ under this action---are straightforward and direct computations along similar lines.
\end{proof}
\begin{remark}
Note that the constant functions act by a multiple of $\Gamma$ on the degree 0 fields in $\sS$. This is the sense in which the symmetry is axial: it acts by opposite factors on $\sV^+$ and $\sV^-$.
\end{remark}

\section{A Few Lemmas on Batalin-Vilkovisky Quantization}
\label{sec: lemmas}

In this section, we prove three lemmas which will be useful in the sequel, and which are implicit in the literature, but which have not been directly proved there. For a more complete discussion of the BV formalism, see the appendix, which assembles results from the literature relevant to the discussion here.

In the Appendix, we define the strong (equivariant) quantum master equation (sQME) for the action of an elliptic dgla $\sL$ on a BV theory $\sF$. For the reader familiar with the formalism of Costello and Gwilliam, a solution of the sQME is what those authors would call an inner action of $\sL$ on the quantum theory $\sF$. The first of the three remaining lemmas in this section makes precise the idea that the sQME measures the failure of an $\sL$-equivariant quantization to be trivial over $C^\bullet(\sL)$. Before we state it, we make the following 
\begin{definition}
Assume given an action of an elliptic dgla $\sL$ on a free theory $\sF$; let $I_{\hbar=1}[t]$ denote $I_{tr}[t]+I_{wh}[t]$. Then, for any $t\in \R_{>0}$, the \textbf{scale $t$ fundamentally quantum equivariant observables}, denoted $\Obqh[t]$, are
\[
\left(\widehat{Sym}\left((\sL[1](N)\oplus \sF(N))^\vee\right), Q+d_{\sL}+\{I_{\hbar=1}[t],\cdot\}_t+\Delta_t\right);
\]
in other words, $\Obqh[t]$ is $\Obq{}[t]$ with $\hbar$ set to 1. Similarly, let $\Obqhtriv$ denote the cochain complex whose underlying graded vector space agrees with that of $\Obqh[t]$ but with differential $d_{\sL}+ Q+ \Delta_t$; this is the analogous construction for the \textit{trivial} action of $\sL$ on $\sF$.
\end{definition}

\begin{remark}
It is worth noting here that the scale $t$ plays the same role here as the heat kernel parameter (which is also denoted by the symbol $t$) in Theorem \ref{thm: heatkernel}, since the BV Laplacian $\Delta_t$ is defined in terms of this heat kernel. Physically, $t$ carries units of inverse length squared.
\end{remark}

\begin{lemma}
\label{lem: sQME}
If $I[t]$ satisfies the scale $t$ sQME, then multiplication by $e^{I_{\hbar=1}[t]}$ gives a cochain isomorphism $\Obqh[t]\to\Obqhtriv[t]$. More generally, let $\Obqobstr[t]$ denote the cochain complex 
\[
\left(\sO(\sL[1]\oplus\sF),Q+d_{\sL}+ \Delta_t - \Obstr[t] \cdot\right),
\]
where $\Obstr[t]\cdot$ denotes multiplication by the obstruction in the symmetric algebra $
\sO(\sL[1]\oplus \sF)$. Then, multiplication by $e^{I_{\hbar=1}[t]}$ gives an isomorphism $\Obqh[t]\to \Obqobstr[t]$.
\end{lemma}
\begin{proof}
By direct computation,
\begin{equation}
\label{eq: volumeform}
(Q+d_{\sL}+ \Delta_t)e^{I_{\hbar=1}[t]}  = \Obstr[t] e^{I_{\hbar=1}[t]}.
\end{equation}
Let $J$ be an element of $\Obq_{\hbar^{-1}}[t]$.
Now, consider
\[
\left(Q+d_{\sL}+ \Delta_t- \Obstr[t]\right) (Je^{I_{\hbar=1}[t]}).
\]
$Q$ and $d_{\sL}$ are both graded derivations for the commutative product; $\Delta_t$, however, is not, and the Poisson bracket is precisely the failure of this to be the case. In fact, direct computation reveals that
\[
\Delta_t\left( Je^{I_{\hbar=1}[t]}\right) = \Delta_t(J)e^{I_{\hbar=1}[t]}+(-1)^{|J|}J \Delta_t e^{I_{\hbar=1}[t]} +e^{I_{\hbar=1}[t]}\{I_{\hbar=1}[t],J\}_t.
\]
It follows that 
\begin{align*}
\left(Q+d_{\sL}\right.&\left.+ \Delta_t-\Obstr[t]\right) (Je^{I_{\hbar=1}[t]})\\
& = (QJ+d_{\sL}J + \Delta_t J+\{I_{\hbar=1}[t],J\}_t)e^{I_{\hbar=1}[t]}-\Obstr[t] Je^{I_{\hbar=1}[t]} +(-1)^{|J|}J\Obstr[t]e^{I_{\hbar=1}[t]}\\
&=(QJ+d_{\sL}J+ \Delta_tJ+\{I_{\hbar=1}[t],J\}_t)e^{I_{\hbar=1}[t]}.
\end{align*}
The above equation states precisely that multiplication by $e^{I_{\hbar=1}[t]}$ intertwines the differentials of $\Obqh[t]$ and $\Obqobstr[t]$. 
\end{proof}

\begin{lemma}
\label{lem: obstrind}
For a free theory with an action of the elliptic dgla $\sL$, the obstruction $\Obstr[t]$ is independent of $t$. Using Lemma \ref{lem: obstr} of the Appendix, this implies that $\Obstr[t]$ is local.
\end{lemma}

\begin{proof}
Let us rewrite Lemma \ref{lem: obstrtangent} as the following equation:
\[
\Obstr[t'] = \frac{d}{d\epsilon}\left(W(P(t', t), I[t]+\epsilon \Obstr[t])\mod \hbar\right)
\]
The right-hand side of the above equation is represented diagrammatically by a sum over connected Feynman trees where one of the vertices uses $\Obstr[t]$ and all the others use $I_{tr}[t]$. At any scale, the Feynman diagram expansion of $\Obstr[t]$ for a free theory has only external edges corresponding to $\sL$. The elements of $\sL$ are non-propagating fields in the sense that the propagators appearing in the Feynman diagrams of RG flow connect only $\sF$ edges and not $\sL$ edges. This means, in particular, that the only connected Feynman diagrams having a vertex labelled by a term from $\Obstr[t]$ are those with only that vertex. In other words, the right hand side of the equation above is just $\Obstr[t]$, which completes the proof. 
\end{proof}

\begin{remark}
This result is particular to the case that we are dealing with families of free theories.
\end{remark}

Let us now examine what happens if $\Obstr[t]$ is exact in $C^\bullet_{red}(\sL)$:

\begin{lemma}
If $\Obstr[t]= d_{\sL}J$ for some $J\in C^\bullet_{red}(\sL)$, then $I[t]-\hbar J$ solves the scale $t$ sQME. Conversely, if there exists $J$ such that $I[t]-\hbar J$ solves the scale $t$ sQME, then $\Obstr[t]=d_{\sL}J$.
\end{lemma}
\begin{proof}
Let's start with the first assertion:
\begin{align*}
(d_{\sL}+Q)(I[t]&- \hbar J) +\frac{1}{2}\{I[t]-\hbar J, I[t]-\hbar J\}_t +\Delta_t (I[t]-\hbar J)\\
&=\hbar\Obstr[t]-\hbar d_{\sL}J=0,
\end{align*}
where we are using that $J$ doesn't depend on $\sF$, so that $QJ = \Delta_t J = \{J,\cdot\}_t=0$. Thus, $I[t]-\hbar J$ solves the sQME. Conversely, the above equation shows that if there exists a $J\in C^\bullet_{red}(\sL)$ such that $I[t]-\hbar J$ satisfies the scale $t$ sQME, then $\Obstr[t]$ is exact. 
\end{proof}

\begin{remark}
Combining the preceding lemma with Lemma \ref{lem: rgqme}, we see that the obstruction is exact at one scale if and only if it is exact at all scales, since \[W(P(t',t),I[t]-\hbar J)= I[t'] -\hbar J,\] by an analysis of the Feynman diagram expansion of \[W(P(t',t),I[t]-\hbar J).\] If one prefers not to use Feynman diagrams, one can use the explicit definition of the operator $W(P(t',t),\cdot)$ and the fact that $\partial_P J=0$.
\end{remark}

\section{The Equivariant Observables of the Massless Free Fermion}
\label{sec: line}
In this section, we use the notations of Section \ref{sec: ferm} concerning the massless free fermion. We will assume that an elliptic dgla $\sL$ acts on the free theory $\sS$. This could be the complex $\sL_\R$ from Lemma \ref{lem: actionofforms}, but our results will hold in the more general case. 

The main goal of this section is to prove the following

\begin{proposition}
\label{prop: line}
If an elliptic dgla $\sL$ acts on the massless free fermion theory $\sS$, then there is a deformation retraction of $C^\bullet(\sL)$-modules
\[
\begin{tikzcd}
\left( \sO(\sL[1]), d_{\sL}+\delta\right )\arrow[r,shift left = .5 ex,"\tilde \iota"] &  \Obqh[\infty]\arrow[l, shift left = .5 ex,"\tilde \pi"]\arrow[loop right, distance = 4em,start anchor = {[yshift = 1ex]east},end anchor = {[yshift=-1ex]east}]{}{\tilde \eta}
\end{tikzcd}.
\]
More precisely, $\tilde \eta$ is a degree --1 map and the three maps satisfy 

\begin{align*}
\tilde \iota \circ \tilde \pi -id &= [d_{\Obqh},\tilde \eta]\\
\tilde \pi\circ \tilde \iota-id &= 0.
\end{align*}
Here, $\delta$ is a degree one endomorphism such that $(d_\sL+\delta)^2=0$. We recall that $\sO(\sL[1])$ is the underlying graded vector space of $C^\bullet(\sL)$, and it is given the obvious $C^\bullet(\sL)$-module structure, namely the one where $C^\bullet(\sL)$ acts on $\sO(\sL[1])$ by multiplication.
\end{proposition}

We will see that $\delta$ is given by an explicit, if complicated, formula.

\begin{remark}
By a small modification of the arguments in \cite{CG2}, (infinitesimal) homotopical RG flow defines cochain isomorphisms $\Obqh[t]\cong \Obqh[t']$ for all $t,t'$. It follows that $\Obqh[t]$ also deformation retracts onto a rank-one $C^\bullet(\sL)$ module for all $t$.
\end{remark}

We note that, because $D$ is formally self-adjoint, we can use the characterization of the heat kernel from Lemma \ref{lem: spectrum}. This implies, in particular, that we can define the scale infinity BV heat kernel $K_\infty$ and propagator $P(0,\infty)$, and therefore the cochain complex of scale-infinity equivariant quantum observables $\Obq{}[\infty]$ in the proposition is well-defined.

Recall that the $\sL$ symmetry has an obstruction $\Obstr$ to quantization. As a corollary of Proposition \ref{prop: line}, we will find the following:

\begin{corA}
Choose the same hypothesis as in Proposition \ref{prop: line}, and recall the notation $\Obqobstr$ from Lemma \ref{lem: sQME}. There is a deformation retraction of $C^\bullet(\sL)$-modules
\[
\begin{tikzcd}
\left( \sO(\sL[1]) , d_{\sL}+ \Obstr[\infty]\cdot\right )\arrow[r,shift left = .5 ex,"\tilde \iota"] &  \Obqobstr[\infty]\arrow[l, shift left = .5 ex,"\tilde \pi"]\arrow[loop right, distance = 4em,start anchor = {[yshift = 1ex]east},end anchor = {[yshift=-1ex]east}]{}{\tilde \eta}
\end{tikzcd}.
\]
Here, $\Obstr[\infty]\cdot$ denotes multiplication by $\Obstr[\infty]$ in $\sO(\sL[1])$.
\end{corA}

Finally, in Lemma \ref{lem: obstrtrivbund}, we show that---as a consequence of the corollary---the rank-one module of the proposition is isomorphic to the trivial $C^\bullet(\sL)$-module if and only if $\Obstr$ is cohomologically trivial.

\begin{remark}
The $C^\bullet(\sL)$-module $\sO(\sL[1])$ with differential $d_\sL+\delta$ can be thought of as the space of sections of a line bundle over the formal moduli problem $B\sL(M)$. Lemma \ref{lem: obstrtrivbund} then implies that this line bundle is trivial if and only if the obstruction is cohomologically trivial. In other words, Lemma \ref{lem: obstrtrivbund} is an interpretation the statement ``the anomaly is the obstruction to the existence of a well-defined and gauge-invariant partition function'' discussed in the introduction. 
\end{remark}

\subsection{The Proof of Proposition \ref{prop: line}}

The proof of Proposition \ref{prop: line} is a minor modification of arguments in Chapters 2.5-2.6 of \cite{othesis}, and it uses the Hodge decomposition of Lemma \ref{lem: hodge} and properties of deformation retractions, including the homological perturbation lemma. Before we begin the proof itself, we spell out some details of the relevant background information.

The following is the main theorem discussed in \cite{crainic}; it is a generalization of Theorem 2.5.3 of \cite{othesis}.

\begin{lemma}[Homological Perturbation Lemma]
Suppose that 

\[
\begin{tikzcd}
(W,d_W)\arrow[r,shift left = .5 ex," \iota"] & (V,d_V)\arrow[l, shift left = .5 ex," \pi"]\arrow[loop right,distance = 4em,start anchor = {[yshift = 1ex]east},end anchor = {[yshift=-1ex]east}]{}{\eta}
\end{tikzcd}
\]
is a deformation retraction, and suppose that $\delta_V$ is a degree 1 operator such that $d_V+\delta_V$ is a differential and $(1-\delta_V\eta)$ is invertible. Then 

\[
\begin{tikzcd}
(W,d_W+\delta_W)\arrow[r,shift left = .5 ex," \iota'"] & (V,d_V+\delta_V)\arrow[l, shift left = .5 ex," \pi'"]\arrow[loop right, distance = 4em,start anchor = {[yshift = 1ex]east},end anchor = {[yshift=-1ex]east}]{}{\eta'}
\end{tikzcd}
\]
is a deformation retraction, where

\begin{align*}
\delta_W &= \pi (1-\delta_V \eta)^{-1}\delta_V \iota\\
\iota' & = \iota + \eta (1-\delta_V\eta)^{-1}\delta_V\iota\\
\pi' & = \pi + \pi (1-\delta_V\eta)^{-1}\delta_V\eta\\
\eta' & = \eta+ \eta(1-\delta_V\eta)^{-1}\delta_V\eta
\end{align*}
are the perturbed data of the deformation retraction.
\end{lemma}

\begin{lemma}
\label{lem: comphe}
Given two deformation retractions
\[
\begin{tikzcd}
(V_1,d_1)\arrow[r,shift left = .5 ex," \iota_1"] & (V_2,d_2)\arrow[l, shift left = .5 ex," \pi_1"]\arrow[loop right, distance = 4em,start anchor = {[yshift = 1ex]east},end anchor = {[yshift=-1ex]east}]{}{\eta_1}
\end{tikzcd}
\]

and

\[
\begin{tikzcd}
(V_2,d_2)\arrow[r,shift left = .5 ex," \iota_2"] & (V_3,d_3)\arrow[l, shift left = .5 ex," \pi_2"]\arrow[loop right, distance = 4em,start anchor = {[yshift = 1ex]east},end anchor = {[yshift=-1ex]east}]{}{\eta_2}
\end{tikzcd},
\]

the composite 

\[
\begin{tikzcd}
(V_1,d_1)\arrow[r,shift left = .5 ex," \iota_2 \iota_1"] & (V_3,d_3)\arrow[l, shift left = .5 ex," \pi_1 \pi_2"]\arrow[loop right, distance = 4em,start anchor = {[yshift = 1ex]east},end anchor = {[yshift=-1ex]east}]{}{\eta_2+\iota_2\eta_1\pi_2}
\end{tikzcd}
\]

is also a deformation retraction.
\end{lemma}  
\begin{proof}
Direct computation. 
\end{proof}

Armed with the preceding results on deformation retractions, we can proceed to the proof of Proposition \ref{prop: line}. The main idea of the proof is repeated application of the homological perturbation lemma. We will break the proof up into small pieces.

Let us first establish some notation that we will use throughout the remainder of this section. Let $\sO$ denote the underlying graded vector space of $\Obqh$. Let us also abbreviate $\widehat{\Sym}((H^\bullet \sS)^\vee)$ to $W$, letting $W^{(j)}$ denote the $\Sym^j$ space of $W$. Similarly, let $W_\perp$ denote $\widehat{\Sym}((H^\bullet \sS)^{\perp, \vee})$, with $W_\perp^{(j)}$ the $\Sym^j$ space of $W_\perp$.

The Hodge decomposition of Lemma \ref{lem: hodge} tells us that the cokernel of $D$ is isomorphic to its kernel. This allows us to decompose $\sS$ as follows:

\begin{equation}
\label{eq: hodge}
\begin{tikzcd}
\ker D \oplus \Im D \arrow[r,"D\Gamma"] &\ker D \oplus \Im D.
\end{tikzcd}
\end{equation}
It follows that $H^\bullet \sS = \ker D \oplus (\ker D[-1])$, and we denote by $H^\bullet \sS^\perp$ the cochain complex $\Im D \overset{D\Gamma}{\longrightarrow}\Im D$; this notation is appropriate because $\Im D=\ker D^\perp$.

It follows that $W=\widehat\Sym\left((\ker D)^\vee \oplus (\ker D)^\vee[1]\right )$ and similarly \[W_\perp=\widehat\Sym\left((\Im D)^\vee \oplus (\Im D)^\vee[1]\right )\].

Using Equation \ref{eq: hodge},
\[
\sO = \prod_{i,j,k=0}^\infty \left(\Sym^i(\sL^\vee) \otimes W^{(j)}\otimes W_\perp^{(k)}\right).
\]
Since $\sO = \sO(\sL[1])\otimes W\otimes W_\perp$, we can endow $\sO$ with a $\sO(\sL[1])$-module structure, which gives a dg-$C^\bullet(\sL)$-module structure when we endow $\sO$ with the differentials from either $\Obqh[\infty]$ or $\Obqhtriv[\infty]$. Throughout the remainder of this section, any cochain complex whose underlying graded vector space is of the form $\sO(\sL[1])\otimes U$, where $U$ is a graded vector space, is given the dg-$C^\bullet(\sL)$-module structure in an analogous way. In all cases, it will be manifest that the $\sO(\sL[1])$ action on $\sO(\sL[1])\otimes U$ interacts as necessary with the differential on $\sO(\sL[1])\otimes U$.

Recall that the differential on $\Obqh[\infty]$ can be written as a sum 
\[Q+d_{\sL}+\{I_{\hbar=1}[\infty],\cdot\}_\infty + \Delta_\infty.\]
We will treat the terms $\Delta_\infty$ and $\{I_{\hbar=1}[\infty],\cdot\}_\infty$ as successive deformations of the differential $Q+d_{\sL}$ on $\sO$, starting from a homotopy equivalence of $(\sO,Q+d_{\sL})$ with a smaller subspace. The cochain complex $(\sO,Q+d_{\sL})$ describes the \textit{classical} observables of the massless free fermion with \textit{trivial} $\sL$ action.

\begin{proposition}[Cf. Proposition 2.5.5 and Theorem 2.6.2 in \cite{othesis}] 
\label{prop: he1}
There is a deformation retraction of $C^\bullet(\sL)$-modules
\[
\begin{tikzcd}
C^\bullet(\sL)\otimes W\arrow[r,shift left = .5 ex," \iota"] & (\sO,Q+d_{\sL})\arrow[l, shift left = .5 ex," \pi"]\arrow[loop right, distance = 4em,start anchor = {[yshift = 1ex]east},end anchor = {[yshift=-1ex]east}]{}{\eta}
\end{tikzcd}.
\]
\end{proposition}

\begin{proof}
As described in Section 2.6 of \cite{othesis}, the Hodge decomposition leads to the desired homotopy equivalence, with maps $\pi,\iota,$ and $\eta$ defined as follows. $\pi$ is just projection onto the subspace of $\sO$ for which $k=0$, and $\iota$ is the inclusion of this subspace into $\sO$. On the other hand, $\eta$ is a bit more subtle. We let $P$ denote the degree --1 operator on $\sS$ whose kernel is $P(0,\infty)$ (see Definition \ref{def: propagator}). On $\sV[-1]\subset \sS$, $P$ acts by 0 on $\ker D$ and by $D^{-1}$ on $\Im D$. $P$ also induces an operator on $\sS^\vee$, and we can also extend it $C^\bullet(\sL)$-linearly to $\sO$ as a derivation for the multiplication in $W_\perp$; we will abuse notation and call all three operators $P$. We define $\eta$ to be 0 on the $k=0$ component of $\sO$ and $P/k$ on the $k\neq 0$ components of $\sO$. Note that $\eta$ acts only on the $W_\perp$ part of $\sO$. The fact that this is a deformation retraction is verified in \cite{othesis}. 
\end{proof}

Now, we turn on the deformation $ \Delta_\infty$ to compute the \textit{quantum} observables of the theory with trivial $\sL$ action:
\begin{proposition}
\label{prop: he2}
There is a deformation retraction of $C^\bullet(\sL)$-modules
\[
\begin{tikzcd}
\left(\sO(\sL[1])\otimes W,d_{\sL}+\Delta_\infty\right)\arrow[r,shift left = .5 ex," \iota'"] & \Obqhtriv[\infty]\arrow[l, shift left = .5 ex," \pi'"]\arrow[loop right, distance = 4em,start anchor = {[yshift = 1ex]east},end anchor = {[yshift=-1ex]east}]{}{\eta'}
\end{tikzcd}.
\]
In other words, the deformation $\Delta_\infty$ of the differential on $\sO$ induces the same deformation on $\sO(\sL[1]){\otimes} W$.
\end{proposition}
\begin{proof}
We use the homological perturbation lemma, though we need to check that the hypotheses of that lemma apply. First, note that the operator $Q+d_{\sL}+ \Delta_\infty$ is a differential, since it is the differential on $\Obqhtriv$, i.e. it is the differential on $\sO$ for the trivial action of $\sL$ on $\sS$.  Moreover, $(1-\Delta_\infty \eta)$ is invertible, as we now show. The operator $\Delta_\infty$ is the operator of contraction with $K_\infty$; based on the description of $K_t$ in the proof of Lemma \ref{lem: diagramcomputation} in the next section, and using the characterization of the heat kernel immediately preceding Proposition 2.37 in \cite{bgv}, it follows that 
\[
K_\infty = \sum_{i=1}^{\dim\ker D} -\phi_i \otimes \xi_i+\xi_i \otimes \phi_i.
\]
Here,  $\{\phi_i\}$ is an orthonormal basis for $\ker D$ sitting in degree 0 and $\xi_i$ is $\phi_i$ as an element of $\sV$, but living in degree 1 in $\sS$. Let $\{\phi^*_i\}$ and $\{\xi^*_i\}$ denote the dual bases for $\ker D^\vee$.

It follows that 
\[
\Delta_\infty = 2\sum_{i=1}^{\dim\ker D}\frac{\partial^2}{\partial \xi^*_i \partial \phi^*_i}.
\]
So $\Delta_\infty$ lowers $\Sym$-degree  in $W$ by 2. On the other hand, $\eta$ doesn't change $\Sym$-degree in $W$. More important, no element of $\sO$ can have a dependence of degree greater than $\dim \ker D$ on the $\phi_i$, since the $\phi_i$ are of ghost number 0 and have fermionic statistics and so anti-commute with each other. It follows that $(\Delta_\infty \eta)^{\dim \ker D+1}=0$ so that
\[
(1- \Delta_\infty \eta)^{-1} = \sum_{l=0}^{\dim\ker D} ( \Delta_\infty \eta)^l,
\]
and the homological perturbation lemma applies to our situation. Moreover, since $\iota$, $\pi$ and $\eta$ were constructed to be $C^\bullet(\sL)$-linear, the explicit formulas of the homological perturbation lemma show that the perturbed data are also $C^\bullet(\sL)$-linear.

The perturbation $\delta_{\Delta_\infty}$ of the differential on $C^\bullet(\sL){\otimes}W$ is given by the formula
\[
\delta_{\Delta_\infty} = \pi (1- \Delta_\infty \eta)^{-1}\Delta_\infty\iota= \pi\sum_{l=0}^{\dim\ker D} (\Delta_\infty \eta)^l\Delta_\infty \iota;
\]
however, since $\eta$ acts trivially on elements of $C^\bullet(\sL){\otimes}W$ and $\Delta_\infty$ preserves $C^\bullet(\sL){\otimes}W\subset \sO$, the formula for $\delta_{\Delta_\infty}$ simplifies:
\[
\delta_{\Delta_\infty} = \pi\Delta_\infty \iota.
\]
The proposition follows. 
\end{proof}

Now, we'd like to understand the cochain complex $(\sO(\sL[1]){\otimes} W, d_{\sL}+ \Delta_\infty)$ a bit better. We continue to use the notation $\xi^*_i, \phi^*_i$ from the proof of the previous proposition.
\begin{proposition}
\label{prop: he3}
There is a deformation retraction of $C^\bullet(\sL)$-modules
\[
\begin{tikzcd}
\left(\sO(\sL[1]),d_{\sL}\right)\arrow[r,shift left = .5 ex," \iota''"] & (\sO(\sL[1]){\otimes} W,d_{\sL}+ \Delta_\infty)\arrow[l, shift left = .5 ex," \pi''"]\arrow[loop right,distance=4em, start anchor = {[yshift=1ex]east}, end anchor = {[yshift=-1ex]east}]{}{\eta''}
\end{tikzcd}.
\]
\end{proposition}

The main consequence of this proposition, together with Proposition \ref{prop: he2}, is 
\begin{corB}
$\Obqhtriv[\infty]$ deformation retracts onto a free rank-one $C^\bullet(\sL)$-module.
\end{corB}

\begin{proof}[Proof of Proposition \ref{prop: he3}]
$\iota''$ and $\pi''$ will be, just like $\iota$ and $\pi$, inclusion and projection operators; we simply need to identify the relevant one-dimensional subspace of $W$ onto which to project. To this end, let $\{\tilde\phi^*_i\}_{i=1}^{\dim\ker D^+}$ be an orthonormal basis for $(\ker D^+)^\vee$ and $\{\tilde {\tilde \phi}^*_i\}_{j=1}^{\dim\ker D^-}$ be an orthonormal basis for $(\ker D^-)^\vee$. The one-dimensional subspace of $W$ we seek is spanned by 
\[
\tilde\phi^*_1\cdots \tilde \phi^*_{\dim \ker D^+}\tilde {\tilde \phi}^*_1 \cdots \tilde {\tilde \phi}^*_{\dim \ker D^-},
\]
i.e. it is the subspace $\Lambda^{top} (\ker D^+)^\vee\otimes\Lambda^{top} (\ker D^-)^\vee$. Now, $\iota''$ can be verified directly to be a cochain map, while to verify that $\pi''$ is a cochain map, note first that $\pi''\Delta_\infty=0$ since $\Delta_\infty$ necessarily lowers the degree of dependence on the $\phi_i$. Thus, 
\[\pi''(d_{\sL}\alpha+\Delta_\infty \alpha) = \pi''d_{\sL}\alpha = d_{\sL}\pi''(\alpha),\]
since $\pi''$ manifestly intertwines the $d_{\sL}$ differential. Moreover, $\pi''\iota''=id$, by construction.

It remains to construct $\eta''$ and verify that $\eta''$ is a chain homotopy between $\iota''\pi''$ and $id$. We can write 
\[
W = \R[\![ \phi_i^*, \xi_i^*]\!],
\]
where the $\phi^*_i$ have ghost number 0 and have fermionic statistics, while the $\xi^*_i$ have ghost number --1 and have fermionic statistics as well. In other words, the $\xi^*_i$ commute with each other, while the $\phi^*_i$ anti-commute among themselves and with the $\xi^*_i$. Now, we let $I=(i_1,\cdots, i_{\dim\ker D})$ be a multi-index, $J\subset (1,\cdots, \dim\ker D)$, and $\alpha \in \sO(\sL[1])$; we consider $J$ to be a multi-index with all indices zero or 1. Moreover, define
\[
\sigma_{I,J}=\#\left(\{i_j\mid j\in J, i_j\neq 0\}\right),
\]
and denote by $|I|$ the total degree of $I$, that is
\[
|I| = \sum_{k=1}^{\dim\ker D} i_k.
\]

Let 
\[
\eta''(\alpha\otimes\xi^{*}_I\phi^*_J) =
(-1)^{|\alpha|+1}\alpha\otimes \frac{(-1)^{|I|}}{2 (|J^c|+\sigma_{I,J})}\left(\displaystyle\sum_{j\notin J} \frac{1}{i_j+1}\xi^*_{j}\xi^*_I \phi^*_j \phi^*_J \right).
\]
Here, $|J^c|$ is the cardinality of the complement of $J$ in $\{1,\cdots, \dim\ker D\}$. We let $1/(|J^c|+\sigma_{I,J})=0$ when both summands in the denominator are zero. Note that \[|(J\backslash\{k\})^{c}|+\sigma_{I\backslash\{k\},J\backslash \{k\}}=|J^c|+\sigma_{I,J}\] when $i_k\neq 0$, where $I\backslash \{k\}$ is $I$ with $i_k$ set to $0$. We will use this fact tacitly below. Let us verify, using the notation $\delta_{jk}$ for the Kronecker delta:

\begin{dmath*}
(d_{\sL}+\Delta_\infty) \eta'' (\alpha\otimes\xi^*_I \phi^*_J) =-\eta''(d_{\sL}\alpha \otimes \xi^*_I\phi^*_J)
\\-\alpha\otimes\frac{(-1)^{|I|}(-1)^{|I|+1}}{|J^c|+\sigma_{I,J}}\left(\sum_{\substack{j\notin J\\ k\in J\cup \{j\}}} \frac{1}{i_j+1}\left( \delta_{jk} \xi_I^* + \xi^*_j \frac{\partial\xi_I^*}{\partial\xi^*_k}\right)\left( \delta_{jk} \phi_J^* - \phi^*_j \frac{\partial\phi_J^*}{\partial\phi^*_k}\right) \right)
= -\eta''(d_{\sL}\alpha \otimes \xi^*_I\phi^*_J)
\\+\alpha\otimes\frac{1}{|J^c|+\sigma_{I,J}}\left(\displaystyle \sum_{j\notin J} \frac{1}{i_j+1}\left( \xi^*_I \phi_J^*+i_j\xi^*_I\phi_J^*\right) -\displaystyle \sum_{\substack{j\notin J\\ k\in J}}\frac{1}{i_j+1}\xi^*_j \frac{\partial\xi^*_I}{\partial\xi^*_k}\phi_j^* \frac{\partial\phi^*_J}{\partial\phi^*_k}\right)
\end{dmath*}
\begin{dmath*}
\eta ''(d_{\sL}+ \Delta_\infty)(\alpha\otimes \xi^*_I\phi^*_J)=\eta''(d_{\sL}\alpha\otimes \xi^*_I\phi^*_J)-\alpha \otimes\frac{(-1)^{|I|-1}(-1)^{|I|}}{|J^c|+\sigma_{I,J}}\\ \left(\sum_{\substack{j\in J\\k\in J^c}} \frac{1}{i_k+1}\xi^*_k \frac{\partial\xi^*_I}{\partial\xi^*_j} \phi^*_k \frac{\partial \phi^*_J}{\partial\phi^*_j}+\sum_{\substack{j\in J\\i_j\neq 0}} \frac{1}{i_j} \xi^*_j \frac{\partial\xi^*_I}{\partial\xi^*_j}\phi^*_j \frac{\partial \phi^*_J}{\partial\phi^*_j}\right)
\end{dmath*}

It follows that, unless $|J^c|=\sigma_{I,J}=0$, i.e. unless $\xi_I^*\phi^*_J=\phi^*_1\cdots \phi^*_{\dim\ker D}$ (a spanning element of $\det D^+\subset W$), \[[\eta'',d_{\sL}+ \Delta_\infty](\xi_I^*\phi_J^*)=-\xi_I^*\phi_J^*=(-id+\iota''\pi'')(\xi_I^*\phi_J^*).\] And, if $|J^c|=\sigma_{I,J}=0$, then \[[\eta'', d_{\sL}+ \Delta_\infty]\phi_J^*=0 = (-id+\iota''\pi'')(\phi_J^*),\] whence the proposition follows, once we note that $\eta''$ is manifestly $C^\bullet(\sL)$-linear (once one takes into account the fact that $\eta''$ has odd cohomological degree and so acquires an extra sign when commuting past elements of $C^\bullet(\sL)$). 
\end{proof}

So far, we have shown that $\Obqhtriv[\infty]$---i.e. the cochain complex of equivariant observables for the \textit{trivial} $\sL$-action on $\sS$---deformation retracts onto a rank-one $C^\bullet(\sL)$-module. To prove Proposition \ref{prop: line}, we need to prove the analogous statement for the cochain complex of equivariant observables with non-trivial $\sL$ action. The differential on this latter chain complex is a perturbation of the differential on $\Obqhtriv[\infty]$ by the term $\{I_{\hbar=1}[\infty],\cdot\}_\infty$, so to complete the proof of Proposition \ref{prop: line},we need to check that this perturbation satisfies the hypotheses of the homological perturbation lemma.

\begin{proof}[Proof of Proposition \ref{prop: line}]
The three preceding propositions, together with Proposition \ref{lem: comphe},  can be combined to yield a deformation retraction
\[
\begin{tikzcd}
C^\bullet(\sL)\arrow[r,shift left = .5 ex," \iota'''"] & (\sO,d_{\sL}+Q+ \Delta_\infty)\arrow[l, shift left = .5 ex," \pi'''"]\arrow[loop right,distance=4em, start anchor = {[yshift=1ex]east}, end anchor = {[yshift=-1ex]east}]{}{\eta'''}
\end{tikzcd}.
\]
We need only to verify that the perturbation $\{I_{\hbar=1}[\infty], \cdot\}_{\infty}$ satisfies the hypotheses of the homological perturbation lemma. Then, the degree-one endomorphism $\delta$ which appears in the statement of the proposition can be computed by applying the homological perturbation lemma with $(V,d_V)=\Obqh$, $W=\sO(\sL[1])$, $\delta_V=\{I_{\hbar=1}[\infty],\cdot\}_\infty$:
\[
\delta=\pi'''\left(1-\{I_{\hbar=1}[\infty],\cdot\}_\infty\eta'''\right)^{-1}\{I_{\hbar=1}[\infty],\cdot\}_\infty\iota'''.
\]

We have seen that the operator \[Q+d_{\sL}+\Delta_{\infty}+\{I_{\hbar=1}[\infty],\cdot\}_\infty\] is a differential (satisfies the wQME), so it remains to show that $(1-\{I_{\hbar=1}[\infty],\cdot\}_\infty \circ \eta''')$ is invertible. To do this, we must understand $\eta'''$ a bit better. In the notations of Propositions \ref{prop: he2} and \ref{prop: he3}, 
\[
\eta''' = \eta'+\iota' \eta'' \pi'.
\]
Let us explicitly compute $\iota'$, $\eta'$, and $\pi'$. Letting $\iota, \eta, $ and $\pi$ be as in Proposition \ref{prop: he1}. Then 

\begin{align*}
\eta' &= \eta\sum_{l=0}^\infty ( \Delta_\infty \eta)^l\\
\iota' &= \iota +\eta\sum_{l=0}^\infty (\Delta_\infty \eta)^l  \Delta_\infty \iota = \iota\\
\pi' &= \pi + \pi \sum_{l=1}^\infty ( \Delta_\infty \eta)^l = \pi.
\end{align*}
$\iota'=\iota$ because $\Delta_\infty$ preserves $\sO(\sL[1]){\otimes} W\subset \sO$, but $\eta$ acts by 0 on this subspace. To show that $\pi' = \pi$, note that $ \Delta_\infty \eta$ commutes with $\iota\pi$ because both $\Delta_\infty$ and $\eta$ preserve $W$. Moreover, $\eta\iota \pi=0$, whence
\[
\pi' = \pi + \pi\iota\pi \sum_{l=1}^\infty ( \Delta_\infty \eta)^l =\pi+\pi \sum_{l=1}^\infty ( \Delta_\infty \eta)^l\iota\pi=\pi
\]
Recall also that the sum in $\eta'$ is finite, since $\eta$ and $\Delta_\infty$ commute and $\Delta_\infty^{\dim\ker D^++1}=0$.

We will need to understand the effect that various operators have on the $\Sym$-grading in $\sO$. The results are depicted in Figure \ref{tab: degrees}. The table should be interpreted as follows: a ``0'' means that the operator preserves the corresponding degree, a ``$>$'' means the operator has only terms which increase the corresponding degree, a ``$\geq$'' means the operator has terms which increase the corresponding degree but none that decrease it; ``$<'$' and ``$\leq$'' have analogous meanings. In all cases except the ``$>$'', the operators only have a finite number of terms increasing or decreasing the corresponding degree.

\begin{figure}
\[
\begin{array}{c|ccc}
& C^\bullet(\sL) & W & W_\perp \\
\hline 
\{I_{\hbar=1}[\infty],\cdot\}_\infty & > & \leq & \geq \\
\eta' & 0 & \leq& 0  \\
\iota' \eta'' \pi' & 0 & \geq & 0
\end{array}
\]
\caption{The effect of various operators on $\Sym$-degree and $\hbar$-power in $\sO$.}
\label{tab: degrees}
\end{figure}

To show that $(1- \{I_{\hbar=1}[\infty],\cdot\}_\infty \circ \eta''')$ is invertible, we will show that the sum
\[
\sum_{l=0}^\infty \left(\{I_{\hbar=1}[\infty],\cdot\}_\infty \circ \eta'''\right)^l
\]
is well-defined. Let $\theta^{(i,j,k)}$ denote the component of $\theta$ in 
\[\widehat{\Sym}^i(\sL[1]^\vee)\otimes W^{(j)}\otimes W_\perp^{(k)},\] 
and let $\zeta$ denote the putative value of $(1-\{I_{\hbar=1}[\infty],\cdot\}_\infty \circ \eta''')^{-1}\theta$, with $\zeta^{(i,j,k)}$ defined similarly. The $C^\bullet(\sL)$ column of Figure \ref{tab: degrees} tells us that $\zeta^{(i,j,k)}$ contains only sums of terms like 
\[
\left( \{I[\infty,\cdot\}_\infty \eta'''\right)^l\theta^{(i',j',k')} 
\] 
for $l+i'\leq i$, i.e. only a finite number of $i'$ contribute to $\zeta^{(i,j,k)}$. Moreover, since in all spaces but the $>$, the corresponding operator has only a finite number of terms changing degree, at a fixed $l$ and $i'$, there is a maximum amount by which $\left( \{I[\infty,\cdot\}_\infty \eta'''\right)^l$ can change $j',k'$. Thus, $\zeta^{(i,j,k)}$ only has contributions from 
\[\left( \{I[\infty,\cdot\}_\infty \eta'''\right)^l\theta^{(i',j',k')} \]
 when $l+i'\leq i$ and $|j'-j|,|k'-k|$ are small enough for the given $l$ and $i'$. It follows that $(1-\{I_{\hbar=1}[\infty],\cdot\}_\infty \eta''')^{-1}$ is well-defined, whence the proposition. 
\end{proof}

\subsection{More Results Concerning the Equivariant Observables}
In Section \ref{sec: lemmas} and the previous subsection, we saw some indication that the obstruction $\Obstr$ is related to the triviality of the equivariant quantum observables $\Obq$ as a $C^\bullet(\sL)$-module (see, e.g. Lemma \ref{lem: sQME} and the Corollary to Proposition \ref{prop: line}). In this subsection, we show that $\Obstr$ is precisely the measure of the non-triviality of the observables as a $C^\bullet(\sL)$-module. We begin by proving Corollary A:

\begin{proof}[Proof of Corollary A]

From the proof of Proposition \ref{prop: line}, we have a deformation retraction
\[
\begin{tikzcd}
\left(\sO(\sL[1]) ,d_{\sL}\right)\arrow[r,shift left = .5 ex," \iota'''"] & (\sO,d_{\sL}+Q+ \Delta_\infty)\arrow[l, shift left = .5 ex," \pi'''"]\arrow[loop right,distance=4em, start anchor = {[yshift=1ex]east}, end anchor = {[yshift=-1ex]east}]{}{\eta'''}
\end{tikzcd}.
\]
The differential on $\Obqobstr[\infty]$ is a perturbation of $d_{\sL}+Q+ \Delta_\infty$ by the operator $ \Obstr[\infty]\cdot$. For the same reasons as in the proof of Proposition \ref{prop: line}, this perturbation satisfies the hypotheses of the homological perturbation lemma. Therefore, all we need to check is that the corresponding differential on $\sO(\sL[1])$ is as specified. To see this, recall that the perturbation of the differential on $\sO(\sL[1])$ is given by the formula
\[
\pi'''(1- \Obstr[\infty]\cdot\eta''')^{-1}\left( \Obstr[\infty]\cdot\right) \iota'''.
\]
Since $\Obstr[\infty]\in C^\bullet(\sL)$, multiplication by $\Obstr[\infty]$ preserves the subspace \[\sO(\sL[1])\subset \Obqobstr[\infty].\] Moreover, by construction, $\eta'''$ is 0 on this subspace. Also, tracing through the definitions of $\iota'''$ and $\pi'''$, one finds that they are simply the usual inclusion and projection maps for the subspace $\sO(\sL[1])\subset \sO$ with respect to the Hodge decomposition of $\sO$. The Corollary follows. 
\end{proof}

Together with the isomorphism $\Obqh[\infty] \to \Obqobstr[\infty]$, Corollary A gives us an explicit characterization of the differential on $\sO(\sL[1])$ from Proposition \ref{prop: line}. Moreover, this differential is directly related to the obstruction. Finally, we can state and prove the following lemma.

\begin{lemma}
\label{lem: obstrtrivbund}
The dg $C^\bullet(\sL)$ module
\[
\left( \sO(\sL[1]), d_{\sL}+ \Obstr[\infty]\cdot\right)
\]
is isomorphic to the trivial $C^\bullet(\sL)$-module $C^\bullet(\sL) $ if and only if $\Obstr[\infty]$ is exact in $C^\bullet(\sL)$.
\end{lemma}

\begin{proof}
Suppose $\Obstr[\infty] = d_{\sL}\alpha$ for $\alpha \in C^0(\sL)$. Then multiplication by $e^{ \alpha}$ gives a cochain isomorphism
\[
\begin{tikzcd}
\left( \sO(\sL[1]), d_{\sL}+\Obstr[\infty]\cdot\right)\ar[r,"\cdot e^{\alpha}"] &C^\bullet(\sL)
\end{tikzcd}.
\]
On the other hand, suppose given a cochain isomorphism of $C^\bullet(\sL)$-modules 
\[
\begin{tikzcd}
\left( \sO(\sL[1]), d_{\sL}+ \Obstr[\infty]\cdot\right)\ar[r,"\Phi"] &C^\bullet(\sL)
\end{tikzcd}.
\]
The statement that this map intertwines differentials and multiplication by elements of $C^\bullet(\sL)$ implies, in particular, that 
\[
 \Obstr[\infty] \cdot\Phi(1) =\Phi(\Obstr[\infty]\cdot 1)= \Phi\left( (d_{\sL}+\Obstr[\infty]\cdot)1\right) = d_{\sL} \Phi(1).
\]
Since $\Phi$ is an isomorphism, $\Phi(1)$ must be invertible in the symmetric algebra $C^\bullet(\sL)$, so that $\log \Phi(1)$ is defined. It follows that 
\[
 \Obstr[\infty] = (d_{\sL}\Phi(1))\Phi(1)^{-1} = d_{\sL} \log \Phi(1),
\]
so that $\Obstr[\infty]$ is exact in $C^\bullet(\sL)$. 
\end{proof}

\section{The Obstruction and the Obstruction Complex}
\label{sec: main}
In this section, we study and characterize $\Obstr$ in its entirety for the case of the axial symmetry. In Section \ref{subsec: main}, we give an abstract characterization of the cohomology class of $\Obstr$ (Theorem \ref{thm: str}). This is the main theorem of the paper. The proof of the theorem relies on a Feynman diagram computation which we defer to subsection \ref{subsec: comps}, where we also compute the obstruction as evaluated on one-forms. Finally, in Section \ref{subsec: fujikawa}, we connect our approach to the more traditional Fujikawa approach, as discussed in \cite{wein}.

\subsection{The Main Results}
\label{subsec: main}
We continue to use the notations from Section \ref{sec: ferm} relating to the massless free fermion and its axial symmetry; in short, we continue to let $\sS$ be the space of fields for the free fermion theory, with $\sL_\R:=\Omega^\bullet_M$ acting by the axial symmetry. The following nice characterization of the obstruction complex $C^\bullet_{loc,red}(\sL_\R)$ will allow us, in Theorem \ref{thm: str}, to characterize the obstruction entirely.

\begin{proposition}[Obstruction Complex]
	\label{prop: qism}
	If $M$ is oriented, the map of complexes of sheaves
	\begin{align*}
	  \Phi: \Omega^\bullet_M[n-1] \hookrightarrow C^\bullet_{loc,red}(\Omega^\bullet_M)
	\end{align*}
	given by 
	\[
	\Phi(U)(\alpha): \beta \mapsto (-1)^{|\alpha|}\alpha \wedge \beta
	\]
	is a quasi-isomorphism.
\end{proposition} 


Having established a characterization of the obstruction complex, we can state the main theorem of this paper (cf. the Main Theorem of the introduction):
\begin{theorem}
\label{thm: str}
Let $(M,g)$ be a Riemannian manifold of dimension $n$, $V\to M$ a $\Z/2$-graded vector bundle with metric $(\cdot,\cdot)$, and $D$ a formally self-adjoint Dirac operator. Denote by $\sL_\R$ the dgla $\Omega^\bullet_M$. (These are the data of the massless free fermion and its axial symmetry, which has an obstruction $\Obstr$.)  Under the quasi-isomorphism $\Phi$ of Proposition \ref{prop: qism}, the cohomology class of $\Obstr[t]$ corresponds to \[(-1)^{n+1}\text{ }2\Str k_t(x,x)dVol_g(x).\]  More precisely,
\[
\left [ \Phi \left( (-1)^{n+1}\text{ }2\Str k_t(x,x)dVol_g(x)\right) \right ] = [\Obstr]
\]
in cohomology. This statement is also true if $M$ is not oriented.
\end{theorem}

\begin{remark}
Top-degree classes in de Rham cohomology on an oriented manifold are determined by their integral; the integral of $\Obstr$ is 
\[
(-1)^{n+1}2\int_M \Str k_t(x,x) dVol_g(x) = (-1)^{n+1}\text{ }2\ind(D).
\]
Thus, the cohomology class of $\Obstr$ coincides with the class of 
\[
(-1)^{n+1}\frac{2\ind (D)}{vol(M)}dVol_g,
\]
as stated in the Introduction.
\end{remark}

\begin{proof}[Proof of Proposition \ref{prop: qism}]
Our line of attack will be to first prove that the two complexes have the same cohomology, and then to use the well-understood structure of the cohomology of the de Rham complex to show that $\Phi$ is a quasi-isomorphism. This first step uses very similar techniques to those employed in Chapter 5.6 of \cite{cost}, Proposition 8.10 of \cite{ggw}, and elsewhere; we adapt those techniques to the case at hand and fill in details missing in the references. The proof in the second step is, to the best of our knowledge, new.

We first invoke Lemma \ref{lem: localCE}, which provides a quasi-isomorphism
\[
C^\bullet_{loc,red}(\Omega^\bullet_M )\simeq \text{Dens}_{M}\otimes_{D_M}^\bL \sO_{red}(J(\Omega^\bullet_M)),
\]	
where the sheaf of jets $J(\Omega^\bullet_M)$ is defined in Definition \ref{def: jets} and $\sO_{red}(J(\Omega^\bullet_M))$ is defined in Equation \ref{eq: jetce}. We will replace both terms in the derived tensor product above with quasi-isomorphic complexes. Let's start with the map of bundles of dglas and of $D_{M}$-modules
\[
\iota:C^\infty_M \hookrightarrow J(\Omega^\bullet_M)
\]
which in the fiber over $x\in M$ takes $c\in \R$ to the jet at $x$ of the constant zero-form $c$. Recall that $D_M$ is the sheaf of differential operators on $M$. By definition it acts on $\cinfty_M$; $J(\Omega^\bullet_M)$ is a $D_M$-module, as discussed in Section 6.2 of Chapter 5 of \cite{cost}. The Poincar\'{e} lemma can be used to show that this is a quasi-isomorphism of $D_{M}$ modules. Namely, since the differential arises from a fiberwise endomorphism of $J$, we can in each fiber use the Poincar\'e lemma to show that the fiber is quasi-isomorphic to $\R$ in degree 0. More precisely, on a sufficiently small open set $U\subset M$, 
\[
J(\Omega^\bullet_M)(U)\cong C^\infty(U)\otimes \R[\![x_1,\cdots,x_n]\!][dx_1,\cdots, dx_n],
\]
with $dx_i$ in cohomological degree 1. Here, the $x_i$ are formal coordinates on $U$, i.e. variables with respect to which the Taylor expansions of elements in the fibers of the jet bundle of $\Lambda^\bullet T^*M$ over $U$ are computed. The ``evaluation at 0'' map $\R[\![x_1,\cdots,x_n]\!][dx_1,\cdots, dx_n]\to \R$ gives a map 
\[\pi(U):J(\Omega^\bullet_M)(U)\to \cinfty_M(U)\]
which, as $U$ varies over the open subsets of $M$, gives a sheaf map. Finally, the proof of the Poincar\'e Lemma furnishes a a degree --1 map
\[
\eta: \R[\![x_1,\cdots,x_n]\!][dx_1,\cdots, dx_n]\to \R[\![x_1,\cdots,x_n]\!][dx_1,\cdots, dx_n]
\] 
which provides a chain homotopy in a deformation retraction of $\R[\![x_1,\cdots,x_n]\!][dx_1,\cdots, dx_n]$, equipped with the de Rham differential, onto $\R$; we let 
\[
\eta(U):= id\otimes \eta :J(\Omega^\bullet_M)(U)\to J(\Omega^\bullet_M)(U).
\]
$\eta(U)$ is manifestly $\cinfty_M(U)$-linear and intertwines with restriction maps. Moreover, $\iota(U)$, $\pi(U)$, and $\eta(U)$ together form a deformation retraction of complexes of sheaves; it follows that $\iota(U)$ is a quasi-isomorphism for all sufficiently small $U$, so that $\iota$ is a quasi-isomorphism of complexes of sheaves.

Moreover, $\iota$ induces a quasi-isomorphism of $D_{M}$-modules
\[
\sO_{red}(J(\Omega^\bullet_M)) =\left(\widehat{Sym}^{>0}_{\cinfty_M}(J(\Omega^\bullet_M)^\vee), d_{CE}\right)\to \left(\widehat{\Sym}^{>0}_{\cinfty_M} (\cinfty_M[-1]), d_{CE}\right)=\cinfty_M[-1].
\]

We consider next the $\text{Dens}_{M}$ factor in the local Chevalley-Eilenberg cochains. Since $M$ is Riemannian, $\text{Dens}_{M}$ is trivial as a bundle, and its (right) $D_{M}$-module structure is induced from viewing a density as a distribution. This $D_M$-module structure encodes the fact that total derivatives integrate to zero, and it has the following projective resolution:
\[
\Omega^\bullet_M[n]\otimes_{\cinfty_M}D_M\to \text{Dens}_M, 
\] 
where the differential $\delta$ is given by the formula
\begin{align*}
\delta(\alpha\otimes fX_1\cdots X_s) = d(f\alpha)\otimes X_1\cdots X_s+(-1)^{|\alpha|}f\alpha \wedge \nabla ( X_1\cdots X_s),
\end{align*}
where $\nabla$ is the natural flat connection on $D_M$ and the map $\Omega^n_{dR}\otimes D_M \to\text{Dens}_{M}$ is given by the action of $D_{M}$ on densities, which happen to also be $n$ forms because $M$ is oriented. It is a direct computation to check that on a local coordinate patch $U$, this is a Koszul resolution of $\text{Dens}_{M}(U)$ with respect to the regular sequence $\{\partial_i\}_{i=1}^n$ of $D_{M}$, and so gives a quasi-isomorphism $\left(\Omega^\bullet_M[n]\otimes_{\cinfty_M}D_M\right)\to \text{Dens}_{M}$.

As a result, we have the following equivalence
\[
C^\bullet_{loc,red}(\sL_\R) \simeq \left(\Omega^\bullet_M\otimes_{\cinfty_M}D_M\right) \otimes_{D_{M}} \cinfty_M[n-1],
\]
where here, the symbol $\simeq$ means ``has isomorphic cohomology to''. We now need to identify this last complex with $\Omega^\bullet_M$, but this simply follows from the associativity of the tensor product and variations on the fact that $A\otimes_A M\cong M$ (isomorphism of left $A$-modules) for any left $A$-bimodule $M$. More precisely, by associativity we have
\[
(\Omega^\bullet_M\otimes_{\cinfty_M} D_M)\otimes_{D_M}\cinfty_M\cong \Omega^\bullet_M\otimes_{\cinfty_M} (D_M\otimes_{D_M}\cinfty_M);
\]
$D_M\otimes_{D_M}\cinfty_M$ is isomorphic as a left $D_M$-module, and hence as a $\cinfty_M$-module, to $\cinfty_M$; it follows that
\[
C^\bullet_{loc,red}(\sL_\R) \simeq \Omega^\bullet_M [n-1],
\]
where, again, the symbol $\simeq$ means ``has isomorphic cohomology to''.

We need only to show that $\Phi$ itself is a quasi-isomorphism, which can be checked on open balls. Suppose that $U$ is homeomorphic to $\R^n$; then by the Poincar\'e Lemma, the cohomology of $\Omega^\bullet_M(U)[n-1]$ is just $\R$ concentrated in degree $-n+1$ and spanned by the constant function 1. Thus, it suffices to check that the following ghost number $-n+1$ element (which is precisely the image of $1\in \Omega^0_{dR}(U)$ under $\Phi$) of $C^\bullet_{loc}(\sL_\R)(U)$ is not exact:
\[
\Xi: \omega \mapsto (-1)^{|\omega|}\omega,
\]
where on the left hand side we are thinking of $\omega$ as a collection of jets of $n$ forms on $U$ and on the right hand side we are thinking of $\omega$ as a density on $U$. $\Xi$ is linear in the jets of its input, and since we're dealing with an abelian dgla, the differential $d_{\sL_\R}$ preserves the $\Sym$ grading. Therefore, if $\Xi$ were exact in $C^\bullet_{loc,red}(\sL_\R)(U)$, it would be $d_{\sL_\R}J$ for $J$ some degree $-n$ element of 
\[\text{Dens}_M\otimes_{D_M}\Hom_{\cinfty_M}(J(\Omega^\bullet_M)[1], \cinfty_M);\]
but the degree $-n$ elements of this latter space are zero, and $\Xi\neq 0$, so $\Xi$ is not exact. Thus, $\Phi(U)$ is a quasi-isomorphism for all open $\R^n$'s $U$, which implies the statement of the proposition. 
\end{proof}

\begin{remark}
\label{rem: twisteddR}
The assumption of orientability is not crucial: we can replace the de Rham complex in the domain of $\Phi$ with the twisted de Rham complex, i.e. the space of de Rham forms with values in the orientation line bundle. The proof is essentially identical, except that since the densities are no longer to be identified with top forms but rather with top twisted forms, we use $D_M\otimes \Omega^\bullet_{tw,dR}$ in the Koszul resolution of the sheaf of densities.
\end{remark}

\begin{proof}[Proof of Theorem \ref{thm: str}]
$\Obstr$ represents a degree 1 cohomology class in $C^\bullet_{loc}(\sL_\R)$, so corresponds under $\Phi$ to a top-degree cohomology class in $H^\bullet_{dR,tw}(M)$. Whether or not $M$ is oriented, its top degree twisted de Rham cohomology is $\R$, and the map $ev_{dR}:\Omega^{top,tw}_{dR}(M) \to \R$ given by integration is a surjection onto cohomology. In a similar vein, define a map $ev_{loc}: C^1_{loc, red}(\sL_\R)\to \R$ given by evaluation of the cochain at the constant function 1. These maps descend to cohomology, and we have the following commutative diagram:
\begin{equation}
\label{eq: diag}
\begin{tikzcd}
H^{top}_{dR,tw}(M)	\arrow[rr,"H^1\Phi"]\arrow[d,"ev_{dR}"]	&& H^1(C^\bullet_{loc,red}(\sL_\R)(M))\arrow[d,"(-1)^{n+1}ev_{loc}"]\\
\R 	\arrow[rr,"id"]	&&	\R.
\end{tikzcd}
\end{equation}

We claim that the map $H^1\Phi$ is an isomorphism. This is because $\Phi$ is a quasi-isomorphism, so induces a quasi-isomorphism on derived global sections. Both $\Omega^\bullet_{dR,tw}$ and $C^\bullet_{loc,red}(\sL_\R)$ are complexes of $c$-soft sheaves since they are sheaves of sections of smooth vector bundles, so their spaces of global sections give a model for the derived global sections. The claim follows.

All but the arrow on the right side of the square have been shown to be isomorphisms, so that arrow is also an isomorphism. It follows that $\alpha \in H^1(\sO_{loc}(\sL_\R)(M)[1])$ corresponds to $\beta \in H^{top}_{dR,tw}(M)$ under $H^1\Phi$ if and only if $ev_{dR}(\beta) = (-1)^{n+1}ev_{loc}(\alpha)$. To complete the proof, we need to show that $\Obstr(1) = -2\int_M \Str k_t(x,x)dVol_g(x)$. We will prove this claim in the next subsection using a Feynman diagram computation.
\end{proof}

\subsection{Explicit Computation of the Obstruction}
\label{subsec: comps}

Our first main goal here is to prove the following

\begin{lemma}
\label{lem: diagramcomputation}
Let $(M,g)$ be a Riemannian manifold, $V\to M$ a $\Z/2$-graded metric vector bundle with metric $(\cdot,\cdot)$, and $D$ a formally self-adjoint Dirac operator. Denote by $\sL_\R$ the dgla $\Omega^\bullet_M$. (These data define the massless fre fermion theory $\sS$, together with an action of $\sL_\R$ on the theory, which has an obstruction $\Obstr$.) Then, $\Obstr[t]$, when evaluated on a constant function $\lambda$, satisfies,
\[
\Obstr[t](\lambda) = -2\lambda \int_M\Str(k_t(x,x))dVol_g(x),
\]
where $k_t$ is the heat kernel for the generalized Laplacian $D^2$.
\end{lemma}

Once this lemma is proved, Theorem \ref{thm: str} follows immediately.

\begin{remark}
Compare this formula to Equations 22.2.10 and 22.2.11 in \cite{wein}, which provide an unregulated formula for the anomaly, i.e. the author computes the $t\to 0$ limit of the above expression by computing the supertrace of the $t\to 0$ limit of $k_t$. This is ill-defined: since the $t\to0$ limit of $k_t$ is the kernel of the identity operator, this limit is singular on the diagonal in $M\times M$. In \cite{wein}, the author regulates those expressions by introducing a mass scale and taking that scale to infinity.  On the other hand, in the present work, we have used the heat kernel techniques of \cite{cost} to make sense of the ill-defined expressions.
\end{remark}

The proof of Lemma \ref{lem: diagramcomputation} is not tremendously difficult; it simply requires a careful accounting of signs. One major component of this accounting consists of determining the precise relationship between the BV heat kernel $K_t$ and the heat kernel $k_t$ as defined in Theorem \ref{thm: heatkernel}, and the next lemma is devoted to specifying this relationship.

\begin{lemma}
\label{lem: BVvsregular}
Let $k_{t,1}$ denote the heat kernel for the generalized Laplacian $D^2$, viewed as an element of $\sV\otimes (\sV[-1])\subset \sS\otimes \sS$. Similarly, let $k_{t,2}$ denote the heat kernel viewed as an element of $(\sV[-1])\otimes \sV$. Then, the BV heat kernel satisfies the following equation
\[
K_t = -k_{t,1}+k_{t,2}.
\]
In particular, $K_t$ is anti-symmetric under interchange of its two factors.
\end{lemma}

\begin{proof}
Recall that, by Definition \ref{def: bvlaplacian}, $K_t\in \sS\otimes \sS$ satisfies 
\begin{equation}
\label{eq: ev}
-id\otimes \langle \cdot, \cdot \rangle(K_t\otimes f) = \exp\left(-t[Q,Q^{GF}]\right)f = \exp\left(-t\left(D^2_{0\to0}+D^2_{1\to1}\right)\right)f,
\end{equation}
where $f\in \sS$. We need to make precise meaning of the expression on the left hand side of Equation \ref{eq: ev}, keeping track of all the signs that arise from the Koszul sign rule. This is a slightly subtle issue, so we will first consider a simple example to illustrate why this subtlety arises. 

Let us consider the symmetric monoidal category of $\Z\times \Z/2$-graded vector spaces where the braiding is given by 
\[
\begin{array}{rcl}
\tau_{W_1,W_2}: W_1\otimes W_2 & \to & W_2\otimes W_1\\
w_1\otimes w_2 & \mapsto & (-1)^{|w_1||w_2|+\pi_{w_1}\pi_{2_2}}w_2\otimes w_1.
\end{array}
\]
Given $W$ in this category, we ask what the natural way to pair an element $\omega_1\otimes \omega_2 \in W^\vee\otimes W^\vee$ with an element $w_1\otimes w_2\in W\otimes W$ is, assuming $\omega_1$, $\omega_2$, $w_1$, and $w_2$ are homogeneous. To do this, we compute the image of $\omega_1\otimes \omega_2\otimes w_1 \otimes w_2$ under the following composition of maps:
\begin{align*}
W^\vee\otimes W^\vee \otimes W\otimes W\overset{id\otimes \tau_{W^\vee,W}\otimes id}{\longrightarrow} W^\vee\otimes W\otimes W^\vee\otimes W \overset{ev\otimes ev}{\longrightarrow} \R\otimes \R\overset{\mu}{\to} \R,
\end{align*}
where $ev: W^\vee\otimes W\to \R$ is the natural pairing of a vector space with its dual and $\mu$ is multiplication of real numbers. In short, by $(\omega_1\otimes \omega_2)(w_1\otimes w_2)$, we mean \[(-1)^{|w_1||\omega_2|+\pi_{w_1}\pi_{\omega_2}}\omega_1(w_1)\omega_2(w_2).\]

Analogously, we understand the term on the left hand side of Equation \ref{eq: ev} as the image of $-id\otimes \langle \cdot, \cdot \rangle \otimes K_t\otimes s$ under the following chain of compositions:

\[
\begin{tikzcd}
\Hom(\sS,\sS)\otimes \text{Bilin}(\sS)\otimes (\sS)^{\otimes 3}\arrow[d,"id\otimes \tau_{Bilin(\sS), \sS}\otimes id\otimes id"]&\\
\Hom(\sS,\sS)\otimes \sS\otimes \text{Bilin}(\sS) \otimes \sS\otimes \sS \arrow[r]&\sS,
\end{tikzcd}
\]

\noindent where the first map uses the Koszul braiding on our symmetric monoidal category to move the first factor of $\sS$ past $\text{Bilin}( \sS)$ and the second map is a tensor product of evaluation maps 
\[
\Hom(\sS,\sS) \otimes \sS \to \sS
\] 
and 
\[
\text{Bilin}(\sS)\otimes \sS \otimes \sS\to \R.
\]

It follows that $K_t = -k_{t,1}+k_{t,2}$, with the relative minus sign arising because the first factors of $k_{t,1}$ and $k_{t,2}$ have ghost numbers of opposite parity. 
\end{proof} 

\begin{proof}[Proof of Lemma \ref{lem: diagramcomputation}]
Recall (Lemma \ref{lem: formofobstr}) that \[\Obstr[t]= d_{\sL_\R}I_{wh}[t]+\Delta_t I_{tr}[t].\] We note that the term $d_{\sL_\R}I_{wh}[t](\lambda)$ is zero because it evaluates $I_{wh}$ with $d\lambda=0$ in one of the slots, so the only term we have to consider is $\Delta_{t}I_{tr}[t](\lambda)$.

Let us study the Feynman diagrams appearing in $I_{tr}[t]$ with a $\lambda$ input into each $\sL_\R$ leg, i.e. onto each ``wavy'' external edge. See Figure \ref{fig: trees} for a diagrammatic interpretation of $I_{tr}[t]$. Let $I_{tr}^{(r)}[t]$ denote the component of $I_{tr}[t]$ in $\Sym$-degree $r$ with respect to $\sL_\R$ inputs ($I_{tr}^{(r)}[t]$ is represented by a tree diagram with $r$ squiggly edges). Because of the shift by one in the definition of Chevalley-Eilenberg cochains, $I^{(r)}$ must be anti-symmetric with respect to zero-form inputs; in particular, $I^{(r)}(\lambda,\cdots,\lambda)=0$ if $r>1$. It follows that $\Obstr[t](\lambda) =\Delta_t I_{tr}^{(1)}[t](\lambda)$, which has a diagrammatic interpretation as the tadpole in Figure \ref{fig: tadpole}. 

\begin{figure}[h]
\centering
\includegraphics[scale = 0.20]{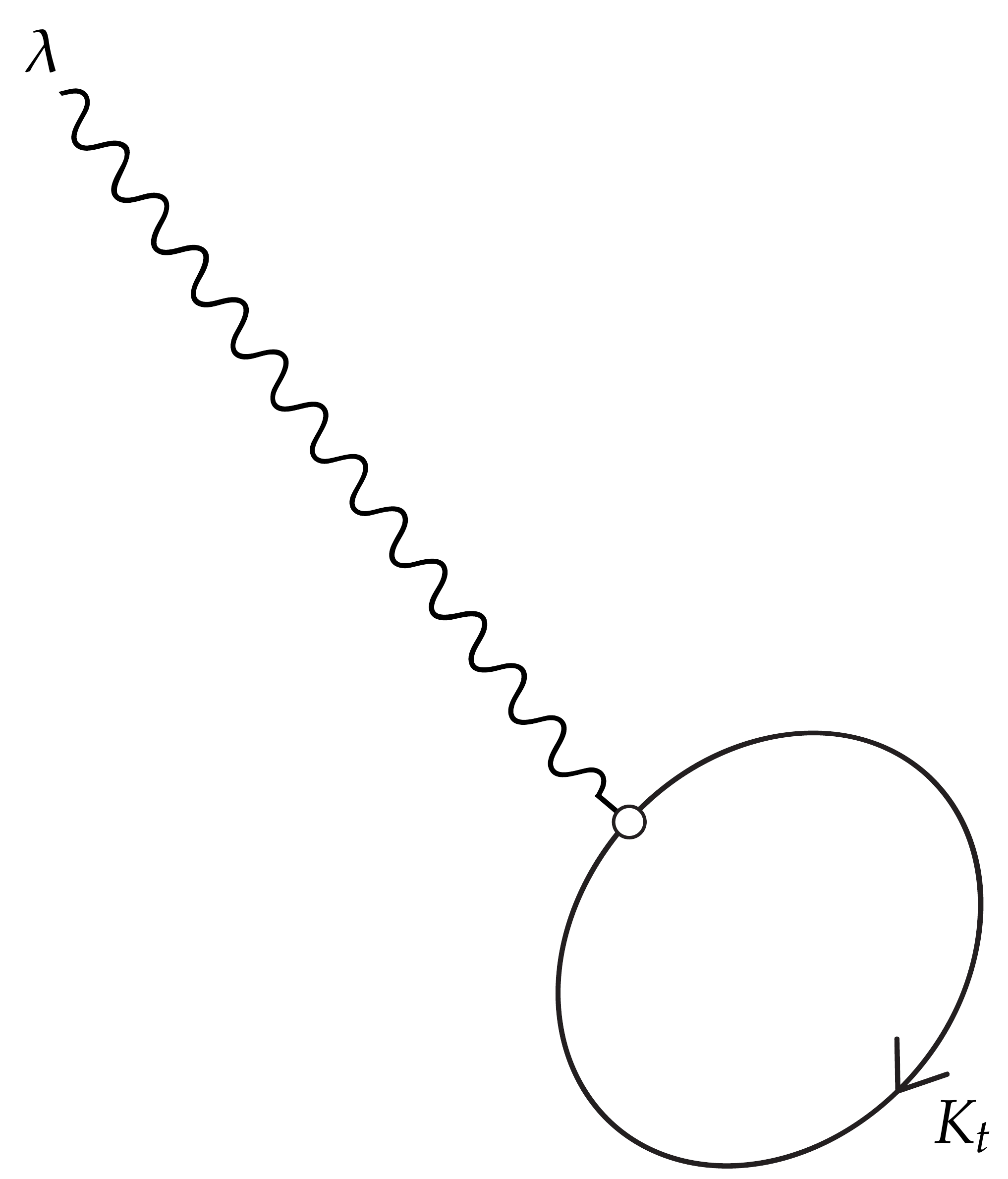}
\caption{The single diagram appearing in the Feynman diagram expansion for $\Obstr[t](\lambda)$.}
\label{fig: tadpole}
\end{figure}

The tadpole diagram, evaluated on $\lambda$, gives 
\begin{equation}
\label{eq: obstructionheatkernel}
\Obstr[t](\lambda)=-\partial_{K_t}I^{(1)}_{tr}[t](\lambda) = -\frac{1}{2}(\langle\cdot,[\lambda, \cdot] \rangle)(K_t+\tau K_t)= -(\langle\cdot,  [\lambda,\cdot]\rangle )(K_t),
\end{equation}
where $\tau$ is the Koszul braiding on the symmetric monoidal category of $\Z\times \Z/2$-graded (topological) vector spaces. Relative to the corresponding notion without Koszul signs, $\tau K_t$ has 
\begin{enumerate}
\item no sign arising from ghost number considerations, since $K_t$ is of degree 1 and therefore belongs to $\sV\otimes (\sV[-1])\oplus (\sV[-1])\otimes \sV$, and
\item a sign arising from the fact that all elements of $\sS$ have fermionic statistics.
\end{enumerate}
The last equality in Equation \ref{eq: obstructionheatkernel} then follows because by construction $\langle \cdot, [\lambda, \cdot]\rangle$ is anti-symmetric in its two arguments. 

We need to show that 
\[
-\langle \cdot, [\lambda,\cdot] \rangle (K_t)=-2\lambda\int_M \Str(k_t(x,x))dVol_g(x),
\]
which will complete the proof in light of Equation \ref{eq: obstructionheatkernel}. Heuristically, the term \[(\langle\cdot,  [\lambda,\cdot]\rangle )(K_t)\] should compute the trace of $K_t$, since $\lambda$ simply acts by multiplication by $\pm \lambda$ and $\ip (K_t)$ pairs the two tensor factors of $K_t$ via $(\cdot,\cdot)$.

To compute $-\langle \cdot, [\lambda, \cdot]\rangle (K_t)$, we take $K_t(x,x)$ (which is an element of the fiber of $V\otimes V$ over $x$), let $\lambda$ act on the second factor, pair the two factors using $(\cdot,\cdot)$, and integrate the resulting function over $M$ against the Riemannian density. We also use the description of $K_t$ from Lemma \ref{lem: BVvsregular}. Taking into account all of these comments, we have
\begin{align*}
-\langle \cdot, [\lambda,\cdot] \rangle (K_t)& = \lambda\int_M \Tr\left((id\otimes \Gamma_{0\to0} -id\otimes \Gamma_{1\to1})(k_{t,1}(x,x)-k_{t,2}(x,x))\right)dVol_g(x)
\\&=-2\lambda\int_M \Str(k_t(x,x))dVol_g(x),
\end{align*}
which is precisely what we claimed; this completes the proof. 
\end{proof}

\begin{remark}
The proof of the McKean-Singer formula in \cite{bgv} has two parts: one first shows that the supertrace of the heat kernel is independent of $t$, then one shows that the $t\to \infty$ limit of the supertrace is $\ind(D)$. Our work has given a physical, Feynman-diagrammatic interpretation of the first part of that proof. Namely, Lemmas \ref{lem: obstrind} and Lemma \ref{lem: diagramcomputation} together imply that 
\[
\int_M \Str(k_t(x,x))dVol_g(x)
\]
is independent of $t$. One can take the $t\to\infty$ limit to get $\ind (D)$.
\end{remark}

We have so far only understood explicitly the obstruction $\Obstr$ evaluated on closed zero forms $\lambda$. What happens when we evaluate it on other closed forms? The following two results show that Lemma \ref{lem: diagramcomputation} already contains all of the difficult computations we need to do to answer this question. This subsection lies somewhat outside the main line of development of this paper, but we include it here for completeness.

\begin{lemma}
\label{lem: degreecounting}
Let $\Obstr^{(r)}[t]$ be the component of $\Obstr[t]$ living in $\Sym$-degree $r$ in $C^\bullet_{red}(\sL_\R)$. Then, for $\alpha_1,\cdots, \alpha_r\in \Omega^\bullet_M(M)$, $\Obstr^{(r)}[t](\alpha_1,\cdots, \alpha_r)$ is zero unless one of the $\alpha_i$ is a function and the rest are one-forms.
\end{lemma}
\begin{proof}
Since $\Obstr[t]$ is a degree 1 Chevalley-Eilenberg cocycle, $\Obstr^{(r)}[t]$ is non-zero only when evaluated against arguments whose total ghost number in $\sL_\R[1]$ is--1. However, if any of these forms has degree higher than 1 in $\sL_\R$, then $\Obstr$ will evaluate to zero because $\Obstr$ only uses forms through their action on $\sS$ and these higher forms act trivially on $\sS$. Thus, $\Obstr^{(r)}[t]$ is non-zero only when evaluated on a collection of one-forms and functions whose total ghost number in $\sL_\R[1]$ is --1. This can only be accomplished if all but one of the forms is a one-form. 
\end{proof}

Thus, to understand $\Obstr[t]$ explicitly, we need only to evaluate it on collections of de Rham forms of the sort described in the preceding lemma. In fact, because a symmetric function is determined by its value when all of its inputs are identical, it suffices to consider $\Obstr[t](\lambda+\alpha)$, where $\alpha$ is a 1-form and $\lambda$ is a function, and we are most interested in the case that $\alpha$ and $\lambda$ are closed. The following corollary then follows from Theorem \ref{thm: str} and suffices to characterize $\Obstr[t]$ evaluated on any collection of closed forms. 

\begin{ucorollary}
Let $\alpha$ be a closed 1-form and $\lambda$ a closed 0-form. Then 
\[\Obstr[t](\lambda+\alpha) = -2\lambda\int_M \Str(k_t(x,x))dVol_g(x).\]
Together with Lemma \ref{lem: degreecounting}, this characterizes the value of $\Obstr[t]$ on all closed forms.
\end{ucorollary}

\begin{proof}
First note that both $\Obstr[t]$ and $\Delta_t I_{tr}^{(1)}$ are closed elements of $C^\bullet_{loc,red}(\sL_\R)$. For $\Obstr[t]$, this is contained in the statements of Definition-Lemma \ref{def: tobstr} and Lemma \ref{lem: obstrind}. On the other hand, $\Delta_t I_{tr}^{(1)}$ is local because by a small modification in the proof of Lemma \ref{lem: diagramcomputation}, it satisfies
\[
\Delta_t I_{tr}^{(1)}(f)=  \int_M f(x)  \Str(k_t(x,x))dVol_g(x)
\]
for $f\in \cinfty(M)$. This is manifestly local. Furthermore, $\Delta_t I_{tr}^{(1)}$ is closed because it is an element of $\Sym^1(\sL_\R^\vee[-1])$ of top ghost number and $d_{\sL_\R}$ preserves $\Sym$-degree. Thus, $d_{\sL_\R} \Delta_t I_{tr}^{(1)}$ is of Sym-degree 1 and of ghost number one greater than $\Delta_t I_{tr}^{(1)}$, hence is 0. 

The crucial observation is that $\Obstr[t](\lambda)=\Delta_t I_{tr}^{(1)}(\lambda)$, as we saw in the proof of Lemma \ref{lem: diagramcomputation}.  Then, by the proof of Theorem \ref{thm: str}, it follows that $\Obstr[t]$ and $\Delta_t I_{tr}^{(1)}$ represent the same cohomology class in $H^\bullet(C^\bullet_{loc,red}(\sL_\R)(M))$. Since evaluation on cocycles in $H^\bullet_{dR}$ is independent of the choice of representative of a class in $H^\bullet(C^\bullet_{loc,red}(\sL_\R)(M))$, the two cocycles take the same value when evaluated on closed de Rham forms, so that 
\[
\Obstr[t](\lambda+\alpha) = \Delta_t I_{tr}^{(1)}(\lambda+\alpha)=\Delta_t I_{tr}^{(1)}(\lambda)+\Delta_t I_{tr}^{(1)}(\alpha),
\]
where the last equality follows from the linearity of $I_{tr}^{(1)}$ in its $\sL_\R$ arguments. But, by Lemma \ref{lem: degreecounting}, $\Delta_{t} I_{tr}^{(1)}(\alpha)=0$, so that 
\[
\Obstr[t](\lambda+\alpha)=\Delta_t I_{tr}^{(1)}(\lambda)=-2\lambda \int_M \Str(k_t(x,x))dVol_g(x),
\]
as desired. 
\end{proof}

\subsection{The Anomaly as a Violation of Current Conservation}
\label{subsec: fujikawa}
In this subsection, we make contact with the traditional (and familiar to physicists) notion that an anomaly is the failure of a current which is conserved classically to be conserved at the quantum level. The discussion will be relatively informal, but we will secretly be using a very simplified version of the Noether formalism of Chapter 12 of \cite{CG2}. 

Before we begin, we should establish some notation. So far, we have worked only with the equivariant observables. In this subsection, we will work also with the non-equivariant observables; that is, we will work with the observables of the theory defined by $\sS$, an object which is well-defined without reference to any particular action of a dgla on $\sS$. We let $\Obcl[\sS]$ denote the classical observables of $\sS$, i.e. the Chevalley-Eilenberg cochains of the abelian dgla $\sS[-1]$, and $\Obq{}[t](\sS)$ denote the quantum observables of $\sS$, i.e. the graded vector space $\sO(\sS)[\![\hbar]\!]$ together with the differential $d_{\sS[-1]}+\hbar \Delta_t$. 

We discuss the classical master equation (Equation \ref{eq: cme}) in the appendix.  In the case at hand, the classical master equation is equivalent to the statement that the map 
 
\begin{align*}
I: \sL_\R[1] &\to \Obcl{}(\sS)
\end{align*}
given by 
\[
I(X)(\phi, \psi) = \ip[\phi,{[X,\psi]}]
\]
is a cochain map.

\begin{remark}
Recall that we defined $I$ to be an element of $\sL_\R[1]^\vee\otimes  \Obcl{}(\sS)$. We abuse notation and refer to the corresponding map $\sL_\R[1]\to \Obcl(\sS)$ also as $I$. We will often use the notation $I_X$ for the observable $I(X)$.
\end{remark}

\begin{remark}
We should think of the above map as defining a conserved vector current because it (in particular) takes in a one-form and gives a degree zero observable. More precisely, let us ``define'' a vector current 
\[
j\in \Obcl(\sS) \otimes \text{Vect}(M)
\]
by the equation 
\[
\int_M \alpha(j) dVol_g = I_\alpha.
\]
for all one-forms $\alpha$. $j$ is conserved in the following sense. The above Lemma tells us that if $\alpha = df$, then \[I_\alpha = \int_M df(j) dVol_g\] is an exact observable. Unpacking the definition of exactness in $\Obcl{}(\sS)$, we discover that $I_\alpha$ is zero when evaluated on fields $\phi, \psi$ which satisfy the equation of motion $Q\phi = Q\psi =0$. In other words, the following observable 
\[
\int_M df(j) dVol_g = -\int_M f \text{div}(j) dVol_g
\]
is zero when evaluated on fields satisfying the equations of motion. This implies that $\text{div}(j)=0$ when evaluated on the solutions of the equations of motion, which is the usual statement of current conservation.
\end{remark}

Now, we would like to see whether we can define an analogous map $\sL_\R[1]\to \Obq{}[t](\sS)$. We will see that this is obstructed precisely by the supertrace of the heat kernel. In particular, suppose we had a map 
\[
I' : \sL_\R[1] \to \Obq{}[t](\sS)
\]
whose $\hbar^0$ component is $I$. Let us write
\[
I' = I+\sum_{i=1}^\infty \hbar^i I^{(i)}.
\]
To say that $I'$ is a cochain map, we need to show that 
\[
I'_{dX} = (d_{\sS[-1]}+\hbar \Delta_t)I'_X.
\]

The $\hbar^0$ part of this equation is satisfied by the previous Lemma. However, the $\hbar^1$ coefficient is:
\[
d_{\sS[-1]}I^{(1)}_X + \Delta_t I_X=I^{(1)}_{dX}.
\]

It follows after a bit of reflection that $I^{(1)}$ satisfies this equation if and only if 
\[\Delta_t I_X = J(dX)\] 
for some linear map $J: \Omega^\bullet_M[1]\to \R$. But this last equality is equivalent to the cohomological triviality of $\Obstr$, since $[\Obstr]=[\Delta_t I^{(1)}_{tr}[t]]$, as we have seen. Moreover, we can simply take $I^{(j)} =0$ for $j>1$. We have therefore shown the following Lemma:

\begin{lemma}
There exists a lift in the diagram
\[
\begin{tikzcd}
& \Obq{}[t](\sS)\ar[d,"\mod \hbar"]\\
\sL_\R[1] \ar[r, "I"]\ar[ur," I'", dashrightarrow]& \Obcl[\sS]
\end{tikzcd}
\]
if and only if $\Obstr$ is cohomologically trivial in $C^\bullet_{red, loc}(\sL_\R)$.
\end{lemma}

\begin{remark}
In the proof of the previous lemma, we found that the lift was obstructed by the one-leg (i.e. $\Sym^1$) term of $\Obstr$ (the one represented by the tadpole diagram of Figure \ref{fig: tadpole}). However, $\Obstr$ has higher-leg contributions which made no appearance in the proof. Although these higher-leg terms are exact as a corollary of Theorem \ref{thm: str}, it is still worth wondering why those terms did not appear in the proof of the lemma. The answer lies in the fact that we are only considering whether or not $I$ can be lifted to a quantum observable, whereas our computation is a precise way of addressing the question of whether or not $e^{I/\hbar}$ is a quantum observable.
\end{remark}

The preceding Lemma still seems to make little point of contact with the notion of conserved quantum vector currents. The following remark closes the gap.  

\begin{remark}
We have not proved or discussed this here, but the exactness of an observable in $\Obq$ implies that that observable has expectation value zero, if expectation values exist at all (see \cite{CG1} for details on this point). Thus, the lift $ I'$ (when it exists) should be viewed as giving a current whose expectation value is conserved. More precisely, just as we should think of $I_{dA}$ as the divergence of the conserved vector $j$, $ I'_{dA}$ is a quantization of the divergence of $j$. The fact that $I'_{dA}$ is exact says that the quantum expectation value of $div(j)$ vanishes.
\end{remark}

\section{Equivariant Generalizations}
\label{sec: eq}
In this section, we continue to study the massless free fermion theory $\sS$ associated to a $\Z/2$-graded vector bundle $V$ equipped with a Dirac operator; however, we would like to replace the Lie algebra $\sL_{\R}$ from Lemma \ref{lem: actionofforms} with a slightly more general algebra, in the following way. Suppose we have an ordinary Lie algebra $\fg$ acting on $V$ and that this action commutes both with the $\cinfty_M$ action and with $D$. In other words, we have a map of $\Z/2$-graded Lie algebras
\[
\rho: \fg \to \cinfty(M;\End(V))
\]
such that $\rho(\gamma)(D\phi)=D\rho(\gamma)(\phi)$ for all $\phi\in \sV$, $f\in \cinfty(M)$, and $\gamma\in \fg$. What we have done so far is the special case that $\fg=\R$ acting by scalar multiplication, and in fact the general case proceeds very similarly, as we will see.

\begin{remark}
The physical interpretation of this $\fg$-action is as symmetries of the massless free fermion: since elements of $\fg$ commute with $D$, they infinitesimally preserve the equation of motion $D\phi = 0$. 
\end{remark}

We will usually just denote the action of an element $\gamma\in \fg$ on a section $\phi\in \sV$ by $\gamma\phi$. Since $\gamma$ commutes with $D$, it preserves $\ker D^+$ and $\coker D^+$, so we are free to make the following

\begin{definition}
	Let $\fg$ act on $V$ as in the first paragraph of this section. Then, the \textbf{equivariant index} of $D$ is the following linear function on $\fg$:
	\[
	\ind(\gamma, D)= \Tr(\gamma\mid_{\ker D^+})-\Tr(\gamma\mid_{\coker D^+}).
	\]
\end{definition}

\begin{remark}
Note that $\ind([\gamma,\gamma'],D)=0$, since traces of commutators are always zero. In other words, the index defines an element of the Lie algebra cohomology of $\fg$. This will be important in the sequel, since the obstruction-deformation complex in this context will be related to the Lie algebra cohomology of $\fg$ in a way made precise by Proposition \ref{prop: eqqism}.
\end{remark}


With this definition in hand, we can now state the \textbf{equivariant McKean-Singer theorem}:
\begin{theorem}
Let $(M,g)$ be a Riemannian manifold, $V\to M$ a $\Z/2$-graded metric vector bundle, and $D$ a formally self-adjoint Dirac operator on $V$. Suppose given an action $\rho$ of the Lie algebra $\fg$ on $V$ which commutes with the Dirac operator. Then,
\begin{equation}
	\label{eq: eqmcs}
	\ind(\gamma, D) = \int_M \Str(\rho(\gamma) k_t(x,x))dVol_g(x).
\end{equation}
Here, $dVol_g$ is the Riemannian density.
\end{theorem}

The version of the equivariant McKean-Singer theorem presented here is the infinitesimal version of Proposition 6.3 of \cite{bgv}. Alternatively, one can prove Equation \ref{eq: eqmcs} using a slight modification of either of the proofs of that proposition. 

Note that in the special case where $\fg = \R$ and the action is given by scalar multiplication, this theorem just reproduces the regular McKean-Singer formula.

The relevant elliptic dgla in this context is $\sL_{\fg}:=\fg\otimes \Omega^\bullet_M$. We will see that, assuming that $\fg$ acts as above, $\sL_{\fg}$ acts on the free fermion theory. Let us first make the following notational choice. Recall that $\sS=\sV\oplus \sV[-1]$. If $\gamma \in \fg$, we denote by $\gamma^\ddagger$ the operator which acts by $\gamma$ on $\sV^+\oplus\sV^-[-1]$ and $-\gamma^T$ on $\sV^-\oplus\sV^+[-1]$. In other words, $\gamma^\ddagger$ acts on $\sV^+\oplus\sV^-[-1]$ by the given representation $\rho$ of $\fg$ and on $\sV^-\oplus\sV^+[-1]$ by the dual representation. Now, we can construct the action of $\sL_{\fg}$ on the theory $\sS$.

\begin{lemma}
\label{lem: eqaction}
Let 
\[
\sigma: \sL_{\fg} \otimes \sS \to \sS
\]
be given by 
\[
\sigma(\gamma\otimes f, \zeta) = f\gamma^\ddagger\zeta
\]
for $f\in \cinfty(M)$,
\[
\sigma(\gamma\otimes \alpha, \zeta) = c(\alpha)_{0\to1}\gamma^\ddagger\Gamma_{0\to0} \zeta
\]
for $\alpha\in \Omega^1_{dR}(M)$, and 
\[
\sigma(\gamma\otimes \omega, \zeta) =0
\]
for $\omega\in \Omega^{>1}_{dR}(M)$. $\sigma$ is an action of $\sL_{\fg}$ on the massless free fermion theory. Here, $c(\alpha)$ denotes the Clifford action of Definition-Lemma \ref{deflem: cliffaction}.
\end{lemma}

\begin{proof}
As above, we will often use $[\cdot, \cdot]$ to denote the bracket on $\sS[-1]\oplus \sL_{\fg}$ provided by $\sigma$. We need to check that
\begin{enumerate}
\item $Q[\alpha, \zeta] = [d_{dR}\alpha,\zeta]+(-1)^{|\alpha|}[\alpha, Q\zeta]$ for all forms $\alpha$ and all $\zeta \in \sS$, i.e. $Q+d_{dR}$ is a derivation for $[\cdot, \cdot]$,
\item $[\alpha, [\beta, \zeta]]=[[\alpha, \beta],\zeta]+(-1)^{|\alpha||\beta|}[\beta,[\alpha,\zeta]]$ for all forms $\alpha, \beta$ and all $\zeta \in \sS$, i.e. the Jacobi identity is satisfied, and 
\item  $\ip[{[\alpha,\zeta]},\eta]+(-1)^{|\alpha||\zeta|}\ip[\zeta, {[\alpha,\eta]}]=0$ for all forms $\alpha$, $\zeta,\eta \in  \sS$, i.e. $[\alpha, \cdot]$ is a derivation for the pairing $\ip$.
\end{enumerate}
The terms in the first identity are non-zero only if $\alpha$ has degree zero in $\sL_{\fg}$ and $\zeta$ has ghost number zero in $\sS$; in that case, the identity reads
\begin{equation}
\label{eq: derivcheck}
D_{0\to1}\Gamma_{0\to0}f \gamma^\ddagger\zeta \overset{?}{=} c(df)_{0\to1} \gamma^\ddagger\Gamma_{0\to0}\zeta - f\gamma^\ddagger D_{0\to 1}\Gamma_{0\to0} \zeta.
\end{equation}
Here, the question mark over the equals sign means that the equality does not necessarily hold, and is rather an equality whose truth we want to check.

We note the following identities
\begin{enumerate}
\item $\gamma^\ddagger \Gamma_{i\to i}= \Gamma_{i\to i}\gamma^\ddagger$
\item $D_{0\to 1} \gamma^\ddagger = \gamma^\ddagger D_{0\to 1}$
\item $\gamma^\ddagger c(df)_{0\to 1} = c(df)_{0\to 1} \gamma^\ddagger$.
\end{enumerate}
The first holds because $\gamma$ and $\gamma^T$ preserve  $\sV^+$ and $\sV^-$; the second identity holds because by assumption $\gamma$ commutes with $D$ and so too does $\gamma^T$ by the formal self-adjointness of $D$; the third identity holds because $c(df)$ is defined in terms of only $D$ and $f$, both of which commute with $\gamma$ and $\gamma^T$ by assumption. Using these identities, Equation \ref{eq: derivcheck} becomes
\[
D_{0\to 1} f \gamma^\ddagger \Gamma_{0\to0}\zeta \overset{?}{=} c(df)_{0\to 1} \gamma^\ddagger\Gamma_{0\to0}\zeta +fD_{0\to 1}\gamma^\ddagger \Gamma_{0\to0}\zeta,
\]
which is true by the definition of $c(df)$. 

Now we turn to the verification of the Jacobi identity. If $\alpha$ and $\beta$ are both of non-zero cohomological degree, then both sides of the Jacobi identity are automatically zero.  Hence, suppose $\alpha = \gamma_1 \otimes f$ and $\beta= \gamma_2\otimes g$ where $f,g\in \cinfty$ and $\gamma_1,\gamma_2 \in \fg$. We have
\begin{align*}
[\gamma_1\otimes f,[\gamma_2\otimes g,\zeta]]&=fg\gamma_1^\ddagger \gamma_2^\ddagger\zeta\\
[[\gamma_1\otimes f, \gamma_2 \otimes g],\zeta]+[\gamma_2\otimes g, [\gamma_1\otimes f, \zeta]]&=fg[\gamma_1,\gamma_2]^\ddagger\zeta + fg\gamma_2^\ddagger\gamma_1^\ddagger\zeta,
\end{align*}
and we see that the Jacobi identity is satisfied because the actions of $\gamma$ and $-\gamma^T$ are Lie algebra actions. In case $f,\gamma_1,\gamma_2$ are as above and $\alpha$ is a one-form, we have
\begin{align*}
[\gamma_1\otimes f,[\gamma_2\otimes \alpha,\zeta]]&=fg\gamma_1^\ddagger\gamma_2^\ddagger(c(\alpha_{0\to1})\zeta)\\
[[\gamma_1\otimes f, \gamma_2 \otimes \alpha],\zeta]+[\gamma_2\otimes \alpha, [\gamma_1\otimes f, \zeta]]&=f[\gamma_1,\gamma_2]^\ddagger(c(\alpha)_{0\to1}\zeta) + f\gamma_2^\dagger(c(\alpha)_{0\to1}\gamma_1^\ddagger\zeta).
\end{align*}
Because $c(\alpha)$ is defined using only the actions of $D$ and $\cinfty_M$ on $\sV$, $c(\alpha)$ commutes with $\gamma_1^\ddagger$, and so the last term in the above equation is $\gamma_2\gamma_1(c(\alpha)\zeta)$. Just as for the first case, then, the Jacobi identity is satisfied.

Finally, we must verify that the pairing is invariant under this action. To this end, let $\zeta, \eta\in \sS$ have ghost numbers 0 and 1 respectively and both be elements of $\sV^+$, $f\in \cinfty_M$, $\gamma \in \fg$, and $\alpha \in \Omega^1_{dR}$. Then 
\begin{align*}
\ip[{[\gamma\otimes f, \zeta]},\eta]& =\int_M f(\gamma^\ddagger \zeta, \eta) dVol_g=\int_M f(-\gamma^T \zeta, \eta) dVol_g\\&=\int_M f(\zeta, -\gamma\eta) dVol_g=-\ip[\zeta, {[\gamma\otimes f, \eta]}];
\end{align*}
the result holds analogously if $\zeta, \eta \in \sV^-$; if $\zeta, \eta$ have the same ghost number or live in opposite $\Z/2$ components of $\sV$, all terms in the above chain of equalities are zero. For $\zeta, \eta$ now both of ghost number 0 and in $\sV^-$, we have 

\begin{align*}
\ip[{[\gamma\otimes \alpha,\zeta]},\eta] &= \ip[c(\alpha)_{0\to1}\gamma^\ddagger\Gamma_{0\to0}\zeta,  \eta]=\int_M (c(\alpha)\gamma \Gamma \zeta, \eta)dVol_g =\\&-\int_M (\zeta, \Gamma \gamma^T c(\alpha)\eta)dVol_g=\int_M(\zeta, c(\alpha)\gamma^T\Gamma\eta)dVol_g = -\ip[\zeta,{[\gamma\otimes \alpha,\eta]}];
\end{align*}
just as for the case of zero-forms, the proof for $\zeta, \eta \in \sV^+$ is entirely analogous and if $\zeta, \eta$ have differing ghost number or live in pieces of $\sV$ of opposite $\Z/2$-grading, all terms in the preceding chain of equalities are zero. 
\end{proof}

We also have results analogous to those of Section \ref{sec: main}.
\begin{proposition}
\label{prop: eqqism}
Let $M$ be a manifold, $\fg$ a Lie algebra, and $\sL_\fg:= \fg\otimes \Omega^\bullet_M$. If $M$ is orientable, there is a map
\begin{align*}
\Phi: C^\bullet_{red}(\mathfrak g) \otimes\Omega^\bullet_M[n]\to C^\bullet_{loc,red}(\sL_{\fg})
\end{align*}
given by
\[
\Phi(\omega\otimes \alpha)\left( \gamma_1\otimes \alpha_1,\cdots, \gamma_r\otimes \alpha_r\right) = (-1)^{|\alpha|}\omega(\gamma_1,\cdots, \gamma_r)\alpha \wedge \alpha_1\wedge \cdots \wedge \alpha_r,
\]
and $\Phi$ is a quasi-isomorphism. Here, $C^\bullet_{loc, red}(\sL_\fg)$ is the cochain complex of reduced, local Chevalley-Eilenberg cochains. There is an analogous statement if $M$ is not orientable  using twisted de Rham forms.
\end{proposition}

\begin{proof}
The proof is similar to that of Proposition \ref{prop: qism}. Letting $J_{\fg}$ refer to the jet bundle of $\Omega^\bullet_M\otimes \fg$, we have a quasi-isomorphism of sheaves of dglas
\[
C^\infty_M\otimes \mathfrak g \hookrightarrow J_{\fg},
\]
and then $C^\bullet_{red}(J_\fg)$ is quasi-isomorphic to $C^\infty_M \otimes C^\bullet_{red}(\fg)$. Then, just as in Proposition \ref{prop: qism}, we can show that both the source and target of $\Phi$ have the same cohomology groups.

It remains to show that the map specified above is a quasi-isomorphism. We show this by a spectral sequence argument. Both source and target of $\Phi$ are filtered by $\Sym$ degree and $\Phi$ preserves the filtrations. Therefore, the cohomology of both can be computed by a spectral sequence and there is a map of spectral sequences between the two. The first page of the spectral sequence just computes the cohomology of both sides with respect to the differential induced by the de Rham differential; on both sides, this is $\Sym^{>0}(\fg[1]^\vee)\otimes H^\bullet_{dR} $. This is obvious on the source side, while on the target side, this is a minor modification of the argument of Proposition \ref{prop: qism}. The induced map on $E_1$ pages is induced from $\Phi$, where we think of both source and target as having only the truncated, Sym-degree-preserving differential. A small modification of the argument from Proposition \ref{prop: qism} can be used to show that the map on $E_1$ pages is indeed an isomorphism. Namely, on a neighborhood $U$ of $M$ homeomorphic to $\R^n$, the cohomology of $\Sym^{> 0}(\fg[1]^\vee)\otimes \Omega^\bullet_M(U)$ is $\Sym^{> 0}(\fg[1]^\vee)$. To prove that the induced map on $E_1$ pages is an isomorphism, we need to show that the images of elements of $\Sym^{> 0}(\fg[1]^\vee)$ under $\Phi$ are not exact in the cochain complex
\[
\left(\sO_{loc, red}(\fg\otimes \Omega^\bullet_M)(U)[1], d_{dR}\right).
\]
(The differential on this complex is just the differential on the $E_1$ page in the spectral sequence for $C^\bullet_{loc,red}(\sL_\fg)$, since the differential induced from the Lie bracket changes $\Sym$-filtration.) However, if $\Phi(\omega)$ were exact in this truncated CE complex for some $\omega$, then
\[\int_U \omega(\gamma_1,\cdots, \gamma_r)\alpha_1\wedge \cdots \wedge \alpha_r
\]
would be zero for any $\gamma_1,\cdots, \gamma_r$ and any closed $\alpha_1,\cdots, \alpha_r$ of total degree $n$, with at least 1 of the $\alpha_i$ compactly-supported. Take $\gamma_1,\cdots,\gamma_r$ such that $\omega(\gamma_1,\cdots, \gamma_r)\neq0$, set $\alpha_1=\cdots= \alpha_{r-1}=1$, and choose a closed but non-exact compactly-supported top form $\alpha_r\in \Omega^n_{dR}$. Such a top form exists because the compactly-supported cohomology of $\R^n$ is $\R$ in degree $n$. Then,
\[
\int_U \omega(\gamma_1,\cdots, \gamma_r)\alpha_1\wedge \cdots \wedge \alpha_r=\omega(\gamma_1,\cdots,\gamma_r) \int_U \alpha_r \neq 0.
\]
This shows that the map induced on $E_1$-pages is locally an isomorphism, so it is an isomorphism of sheaves. It follows that the induced map on $E_2$-pages is an isomorphism. But the $E_2$ page is also the $E_\infty$ page, so we see that $\Phi$ indeed induces an isomorphism on cohomology. 
\end{proof}

We also have an analogue to Theorem \ref{thm: str}; however, this proof requires a slightly more sophisticated touch, since the appropriate generalization of the diagram in Equation \ref{eq: diag} is not immediately obvious.
\begin{theorem}
\label{thm: eqobstr}
Let $(M,g)$ be a Riemannian manifold, $V\to M$ a $\Z/2$-graded vector bundle with metric $(\cdot,\cdot)$, and $D$ a formally self-adjoint Dirac operator. Let $\fg$ act on $V$ $D$-equivariantly and denote by $\sL_\fg$ the dgla $\fg\otimes \Omega^\bullet_M$. (These are the data defining a free classical fermion theory $\sS$, together with an action of $\sL_\fg$ on the theory, which has an obstruction $\Obstr$.) Under the quasi-isomorphism $\Phi$ of Proposition \ref{prop: eqqism}, the cohomology class of $\Obstr$ as an element of \[C^\bullet_{loc,red}(\fg\otimes \Omega^\bullet_M(M))\] coincides with the cohomology class of \[\mathfrak{T}:\gamma \mapsto (-1)^{n+1}2\Str(\rho(\gamma) k_t)dVol_g\] in $C^\bullet_{red}(\fg)\otimes \Omega^\bullet_M(M)[n]$.
\end{theorem}

\begin{proof}
Consider the following diagram:
\[
\begin{tikzcd}
C^\bullet_{red}(\fg) \otimes \Omega^\bullet_M(M)[n]\arrow[r,"\Phi"]\arrow[rd,"id\otimes PD" below left]& C^\bullet_{loc,red}(\fg\otimes \Omega^\bullet_M(M))\arrow[r,hook,"i"]&C^\bullet_{red}(\fg\otimes \Omega^\bullet_M(M))\arrow[ld,"res"]\\
&C^\bullet_{red}(\fg)\otimes \Hom_\R(H^{-\bullet}_{dR}(M),\R)& \\\\
\end{tikzcd}.
\]
Only the downward pointing arrows have not yet been discussed. $PD$, whose letters stand for ``Poincar\'{e} Duality,'' is a composition given by first taking the projection of $\Omega^\bullet_M$ onto $H^\bullet_{dR}$ given by Hodge theory and then taking the Poincar\'e duality map. The map $res$ is a bit more involved. Let $p$ be an element of $\Sym^k((\fg\otimes \Omega^\bullet_M(M))^\vee)$. Then $res$ is defined by 
\[
res(p)(\beta,\gamma_1,\cdots, \gamma_k) =p(\gamma_1\otimes \beta,  \gamma_2\otimes 1,\cdots, \gamma_k\otimes 1),
\]
where on the right hand side we have used the harmonic representative for the cohomology class $\beta$ to understand $\beta$ as a de Rham form. We are also thinking of the right-hand side of the above expression as an element of the coinvariants with respect to the $S_k$ action of permutation of the $\gamma_i$'s. It is a quick verification that the diagonal map is a chain map and that the diagram commutes.

Now, since the two maps out of $C^\bullet_{red}(\fg)\otimes \Omega^\bullet_M(M)$ are quasi-isomorphisms, it follows that the composite map \[\eta:=res\circ i\] is a quasi-isomorphism as well. In particular, this tells us that the inclusion of local Chevalley-Eilenberg cochains into the full space of Chevalley-Eilenberg cochains is injective on cohomology, i.e. if two local cochains differ by an exact term in the full space, they also differ by an exact term in the local complex. 

We claim that \[[\eta(\Obstr[t])]=[id\otimes PD(\fT)]\] in cohomology. Here, it is crucial to us that $\eta$ factors through the full Chevalley-Eilenberg cochain complex, since this allows us to treat the $d_{\sL_\fg}I_{wh}[t]$ and $\Delta_t I_{tr}[t]$ terms separately. (Recall that we know only that the sum $d_{\sL_\fg}I_{wh}[t]+\Delta_t I_{tr}[t]$ is local.) First of all, this allows us to note that the term $d_{\sL_\fg}I_{wh}[t]$ in the scale $t$ obstruction is exact in $C^\bullet_{red}(\fg \otimes \Omega^\bullet_M)$, so that it does not contribute to the class $[\eta(\Obstr[t])]$. 

We are left with the task of showing that the cohomology class of $\Delta_t I_{tr}[t]$ coincides with that of $\Phi(\fT)$. We claim first of all that the one-leg term $\eta(\Delta_tI^{(1)}_{tr}[t])=id\otimes PD(\fT)$ To see this, note that $id\otimes PD (\fT)\in C^1_{red}(\fg)\otimes \Hom_\R(H^0_{dR}(M),\R)=C^1_{red}(\fg)$, so corresponds to a linear functional on $\fg$. In fact, we have 
\[
\eta\circ \Phi(\fT)(\gamma) = -2\int_M \Str(\rho(\gamma)k_t)dVol_g.
\]
To prove the claim about the one-leg term, it suffices to show that 
\begin{equation}
\label{eq: eqdiagram}
\Obstr[t](\gamma\otimes 1) = -2\int_M \Str(\rho(\gamma)(x) k_t(x,x))dVol_g(x).
\end{equation}
The single-leg contribution corresponds to the tadpole diagram in Figure \ref{fig: tadpole}, and that diagram gives a contribution
\begin{align*}
-\int_M \Tr_{S}&(\gamma^\ddagger(x) k_t(x,x))dVol_g(x)
\\&=-\int_M \Tr_{V^+}(\rho(\gamma) k_t(x,x))dVol_g(x)+\int_M \Tr_{V^-}(\rho(\gamma)^T k_t(x,x))dVol_g(x)-(+\leftrightarrow -)
\\&=-2\int_M \Str(\rho(\gamma)(x) k_t(x,x))dVol_g(x).
\end{align*}
In the last equality, we use the fact that $k_t=k_t^T$ as a result of the formal self-adjointness of $D$, so that we have $\Tr_{V^-}(\gamma^T k_t(x,x)) = \Tr_{V^{-}}((k_t(x,x)\gamma)^T)=\Tr_{V^-}(\gamma k_t(x,x)) $. Equation \ref{eq: eqdiagram} follows.

To complete the proof of the theorem, it now suffices to show that all higher leg contributions to $\Delta_t I_{tr}[t]$ go to something exact under $res$. To see this, note that---by Lemma \ref{lem: degreecounting}, which applies equally well here---any term with $r$ legs is non-zero only when it has one zero-form input and $r-1$ one-form inputs. This implies, in particular, that $\eta$ takes any term in $\Delta_t I_{tr}[t]$ with more than two legs to zero. This is because such diagrams are zero when evaluated on anything with more than one zero-form input, and $\eta$ evaluates a Sym-degree $r$ element of $C^\bullet_{loc,red}$ on $r-1$ zero-forms. 

We are left with just the task of showing that the two-leg term in $\Delta_t I_{tr}[t]$ goes to something exact under $\eta$. Under $\eta$, the two-leg term gives a class in $H^2(\fg) \otimes \Hom_\R(H^1_{dR}(M),\R)$; however, for the special case of a simple Lie algebra $\fg$, $H^2(\fg)=0$, so that the two-leg term is exact under $\eta$, and therefore exact in $C^\bullet_{loc,red}(\fg\otimes \Omega^\bullet_M(M))$.

We prove the general case in the $t\to\infty$ limit, after first establishing some notation. Let  \[\cZ_t:=\gamma_1^\ddagger \otimes \alpha \circ P_t\circ \gamma_2^\ddagger\otimes 1 +\gamma_2^\ddagger\otimes 1 \circ P_t \circ \gamma_1^\ddagger\otimes \alpha,\] where $P_t$ is the operator which appears in Definition \ref{def: propagator} for the propagator $P(0,t)$. We will let $P$ denote the $t\to\infty$ limit of $P_t$; $P$ is characterized by the fact that it is $0$ on $\ker D[-1]\subset \sS$ and $D^{-1}$ on $\Im D[-1]\subset \sS$. Both $P_t$ and $P$ are distributional in nature, acting as operators on smooth sections of $\sS$. We also allow $t$ to take the value $\infty$ in $\cZ_t$. Finally, let $H:=[Q,Q^{GF}]$ be the generalized Laplacian of the gauge-fixed massless free fermion.

With all this notation in place, we note that the two-leg contribution to $\eta(\Delta_t I_{tr}[t])$ comes from the diagram shown in Figure \ref{fig: twolegs}, and its value, on a harmonic one-form $\alpha$ and elements $\gamma_1, \gamma_2\in \fg$, is
\begin{equation}
\label{eq: tdependenttrace}
\Tr_{\sS}(\cZ_t \exp(-tH)),
\end{equation}
where $\Tr_{\sS}(\ell(x,y))=\int_M \Tr(\ell(x,x))dVol_g$ for any smooth kernel $\ell \in\cinfty(M\times M;S\boxtimes S)$. We claim that the $t\to\infty$ limit of the above quantity is given by 
\begin{equation}
\label{eq: trace}
\Tr_{H^\bullet \sS}(\cZ_\infty),
\end{equation}
Here, for any operator $Z$ on $\sS$ $\Tr_{H^\bullet \sS}(Z): = \Tr_{\sS}(Z\Pi)$, where $\Pi$ is the orthogonal projection onto cohomology afforded to us by Lemma \ref{lem: hodge}.

\begin{figure}[h]
\centering
\includegraphics[scale = 0.17]{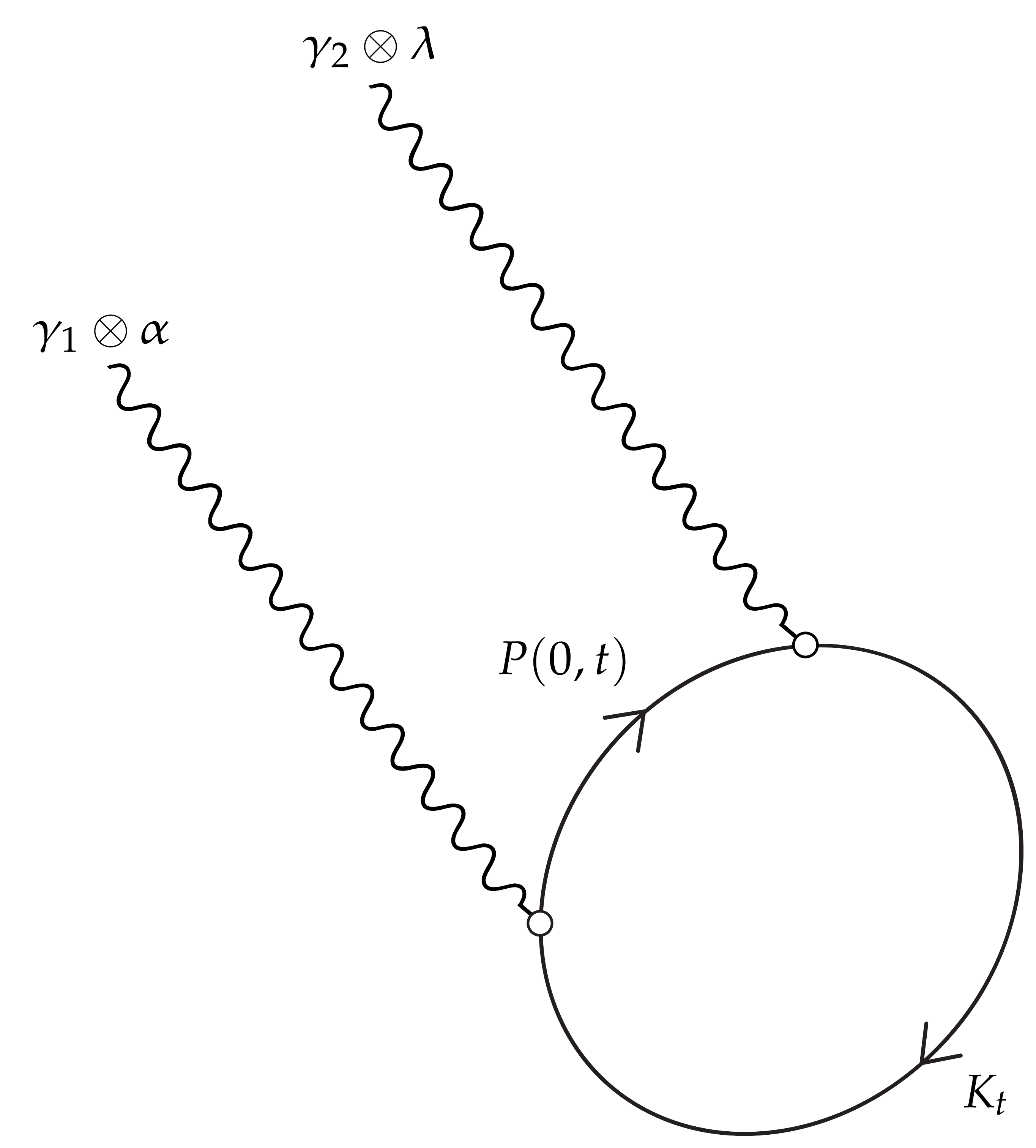}
\caption{The two-leg contribution to $\Delta_t I_{tr}[t]$.}
\label{fig: twolegs}
\end{figure}

Let us assume for the moment that the claim is proved; we wish to show that the quantity in  Expression \ref{eq: trace} is zero. To see this, suppose that $\varphi \in H^1(\sS)=\coker(D_{0\to1})=\ker(D)[-1]\subset \sV[-1]$; then, because $\gamma_1^\ddagger \otimes \alpha$ changes degree by $+1$, $\gamma_2^\ddagger\otimes 1 \circ P \circ \gamma_1^\ddagger\otimes \alpha\varphi=0$; on the other hand, $P\varphi =0$, since $\varphi \in H^\bullet(\sS)$ and by construction $P$ is zero on the cohomology of $\sS$. Moreover, $\gamma_2^\ddagger\otimes 1$ preserves the cohomology of $\sS$, so 
\[
(\gamma_1^\ddagger \otimes \alpha \circ P\circ \gamma_2^\ddagger\otimes 1) \varphi = 0.
\]
Thus, $Z_\infty \varphi=0$. Similarly, if $\varphi\in H^0(\sS)$, 
\[
\cZ_\infty \varphi = 0.
\] 
This shows that Expression \ref{eq: trace} is zero.

It remains only to show that Expression \ref{eq: tdependenttrace} has Expression \ref{eq: trace} as its $t\to \infty$ limit. The difference between the two expressions is
\[
\Tr_{\sS}\Big(\cZ_te^{-tH}-\cZ_\infty \Pi\Big)=\Tr_{\sS}\Big(\cZ_t\left(e^{-tH}-\Pi\right)+(\cZ_t-\cZ_\infty) \Pi\Big).
\]
Let us note that $\cZ_te^{-tH}=\cZ_te^{-tH/2}e^{-tH/2}$, so is the product of two operators $\cZ_te^{-tH/2}$ and $e^{-tH/2}$ with smooth kernels, which therefore admit extensions to bounded operators acting on $L^2$ sections of $S$. As operators on $L^2(M;S)$, they are Hilbert-Schmidt and the trace $Tr_\sS(\cZ_t e^{-tH})$ corresponds to the Hilbert space trace, using a modification of the arguments presented in the proof of Proposition 2.32 of \cite{bgv}. Similar arguments apply for $\cZ_t\Pi$ and $\cZ_\infty\Pi$, using $e^{-tH/2}\Pi=\Pi$. Thus, we are permitted to switch to the $L^2$ context for the computation of traces; all appearances of the expression $\Tr_\sS$ in the sequel will therefore refer to the trace in the Hilbert space sense. Now, on an eigenvector of $D$ with eigenvalue $\lambda$, $P_t$ is zero if $\lambda$ is zero and otherwise multiplies the eigenvector by $\lambda^{-1}(1-e^{-t\lambda^2})$. This is bounded as a function of $\lambda$, so that $P_t$ is a bounded operator. Similarly, $P$ acts by 0 on $\ker D$ and by $\lambda^{-1}$ on non-zero eigenvectors of $D$ with eigenvalue $\lambda$. Since the spectrum of $D$ is discrete (Lemma \ref{lem: spectrum}), it follows that $P$ is bounded in norm. It also follows from the characterizations of $P_t$ and $P$ that $P_t-P$ acts by zero on $\ker D$ and by $\lambda^{-1}e^{-t\lambda^2}$ on eigenvectors with non-zero eigenvalue. This last quantity is monotonically decreasing with $\lambda$ and so 
\[
||(\cZ_t-\cZ_\infty)e^{-tH/2}||\leq C |\lambda_1^{-1}|e^{-3t\lambda_1^2/2},
\]
where $\lambda_1$ is the smallest non-zero eigenvalue of $D$ and $C$ is some constant and $||\cdot||$ refers to the operator norm on bounded operators on $L^2$. Again, we have used the fact that the spectrum of $D^2$ is discrete. It follows that $\cZ_te^{-tH/2}\to \cZ_\infty e^{-tH/2}$ in the operator norm. Thus,
\begin{align}
\label{eq: lasteq}
\left|\Tr_{\sS}\left(\cZ_te^{-tH}-\cZ_\infty\Pi\right)\right| \leq ||\cZ_te^{-tH/2}|| \left|\Tr_{\sS}\left(e^{-tH/2}-\Pi\right)\right|+||(\cZ_t-\cZ_\infty)e^{-tH/2}|||\Tr_\sS(\Pi)|,
\end{align}
where we have used the fact that for a bounded operator $A$ and a trace-class, formally self-adjoint, non-negative operator $B$, 
\[
|\Tr(AB) |\leq ||A||Tr(B).
\]
Now, since $\cZ_t\to \cZ_\infty$ in the operator norm and $e^{-tH}\to \Pi$ in the trace norm (Lemma \ref{lem: infraredkernel}), this shows that 
\[
\Tr_\sS(\cZ_te^{-tH})\overset{t\to\infty}{\longrightarrow} \Tr_\sS(\cZ_\infty \Pi),
\]
as desired. 
\end{proof}

\begin{remark}
Just as in Section \ref{subsec: comps}, the $t$-independence of the obstruction, together with the single-leg computation in the above proof, show the $t$-invariance of the quantity appearing on the right-hand side of the equivariant McKean-Singer formula, Equation \ref{eq: eqmcs}. Moreover, one can take the $t\to \infty$ limit of this quantity (see \cite{bgv}). As is to be expected, the result is $\ind(\gamma,D)$, the global quantity appearing in the equivariant McKean-Singer formula.
\end{remark}

\appendix
\section{Background on Equivariant BV Quantization}
\label{sec: background}
\subsection{Introduction}

In this appendix, we set up the mathematical framework for studying what is in physics called an \textit{anomaly.} Roughly speaking, an anomaly is the failure of a symmetry of a classical field theory to persist after quantization. Symmetries are usually encoded in actions of a Lie algebra---or more generally a dgla, or even more generally an $L_\infty$-algebra---on the space of fields of the theory. We develop this perspective below, following \cite{CG2} and \cite{othesis}, which provide a general framework for treating actions of an $L_\infty$-algebra on a quantum field theory. We will not need the full generality of that framework below---restricting ourselves to dgla actions on \textit{free} field theories---though by illuminating the relationship of this special case to index theory, we hope to provide motivation for the study of the more general case. 

Another perspective on symmetries is that they consist of an interaction term $I$ representing a coupling of the physical fields to background gauge fields. $I$ is required to satisfy a Maurer-Cartan equation which contains the content of the fact that $\fg$ acts on the corresponding theory. This equation is called the classical master equation. This perspective is equivalent to the perspective of the previous paragraph.

To summarize, there are two equivalent perspectives on the notion of a symmetry of a classical theory. We can view a symmetry as

\begin{enumerate}
\item an action of a dgla $\fg$ on the space of fields of a field theory, or
\item a $\fg$-dependent interaction term in the action of the field theory. The elements of $\fg$ are considered to be background fields.
\end{enumerate}

There will be analogues of the classical master equation for $I$ in the quantum case, called the \textbf{weak and strong quantum master equations}. The strong quantum master equation, in particular, is central to the discussion of anomalies.

\subsection{Free Batalin-Vilkovisky Theories}
Free Batalin-Vilkovisky (BV) theories are the starting point for quantization, since in the perturbative formalism we think of interacting theories as perturbations of free ones. In this paper, we are mostly be concerned with families of free theories. However, the formalism for families of free theories requires us to consider interactions, even if the individual theories in the family are free; we will discuss the necessary formalism for interactions in Section \ref{subsec: eqquant}.

We will frequently deal with $\Z\times\Z/2$-graded vector spaces in the sequel. Recall from Section \ref{subsec: not} that if $v$ is a homogeneous element of a $\Z\times\Z/2$-graded vector space, we use the notation $|v|$ to denote its $\Z$ degree and $\pi_v$ to denote its $\Z/2$ degree.

The following is Definition 7.0.1 of Chapter 5 of \cite{cost}, modified to allow an extra $\Z/2$-grading on the space of fields.
\begin{definition}
A \textbf{free classical BV theory} on a compact manifold $N$ consists of the following data:
\begin{enumerate}
\item A $\Z\times \Z/2$-graded $\R-$vector bundle $F$ over $N$. We will call the sheaf of sections $\sF$ the \textbf{space of fields} of the theory.
\item A map of vector bundles
\[
\ip_{loc}: F\otimes F \to \text{Dens}_N
\]
of degree $(-1,0)$. We will denote by $\ip$ the following pairing on the space of global sections $\sF(N)$:
\[
\ip[f_1,f_2] = \int_N \ip[f_1,f_2]_{loc}.
\]
\item A differential operator $Q: \sF\to \sF$ of degree (1,0).
\end{enumerate}
The above data are required to satisfy the following additional conditions:
\begin{enumerate}
\item $\ip_{loc}$ is graded anti-symmetric and non-degenerate on each fiber.
\item $(\sF,Q)$ is a ($\Z/2$-graded) elliptic complex.
\item $Q$ is graded skew self-adjoint for $\ip$.
\end{enumerate}
\end{definition}
\begin{remark}
By graded anti-symmetry of $\ip_{loc}$, we mean that if $f_1,f_2$ are homogeneous elements of the fiber of $F$ over $x\in N$, then
\[
\ip[f_1,f_2]_{loc} = -(-1)^{|f_1||f_2|+\pi_{f_1}\pi_{f_2}}\ip[f_2,f_1]_{loc}.
\]
All notions of graded symmetry and anti-symmetry should be analogously understood. We also note that since $\ip$ is of degree --1, it is only non-zero on pairs $f_1$ and $f_2$ with $|f_1|$ and $|f_2|$ of opposite parity, so that we can rewrite the above equation as 
\begin{equation}
\label{eq: ipsym}
\ip[f_1,f_2]_{loc}=-(-1)^{\pi_{f_1}\pi_{f_2}}\ip[f_2,f_1]_{loc}.
\end{equation}
\end{remark}
\begin{remark}
The elliptic complex $(\sF,Q)$ encodes the information of the linear PDE
\[
Q\phi = 0.
\] 
The zeroth cohomology of $(\sF,Q)$ is the space of solutions to this PDE modulo the identification of solutions differing by an exact term; physically, the zeroth cohomology is the space of solutions of the equations of motion modulo the identification of physically indistinguishable (gauge-equivalent) configurations.
\end{remark}


When we discuss the quantization of free BV theories in Section \ref{subsec: eqquant}, we will need to equip the free BV theory with one more piece of auxiliary data:
\begin{definition}[Definition 7.4.1 of Chapter 5 of \cite{cost}]
A \textbf{gauge fixing} $Q^{GF}$ for a free BV theory is a degree $(-1,0)$ differential operator $\sF\to \sF$, such that $[Q,Q^{GF}]$ is a generalized Laplacian for some metric on $N$ and $Q^{GF}$ is graded self-adjoint for the pairing $\ip$.
\end{definition}

\begin{remark}
In physics, the (opposite of the) $\Z$-degree of an element of $\sF(N)$ is called the ``ghost number,'' while its $\Z/2$-degree is often called its ``statistics;'' we say that an element of $\sF(N)$ has bosonic or fermionic statistics depending on whether it has grading 0 or 1, respectively. The extra $\Z/2$-grading has the effect of changing signs in commutation relations.
\end{remark}


As a very simple example of a free field theory, we have
\begin{example}
Given an (ordinary, i.e. $\Z$-graded) elliptic complex $(\sF,P)$, the \textbf{cotangent theory} $\cot$ to $\sF$ has space of fields $\sF\oplus \sF^![-1]$, with all fields having odd $\Z/2$-grading. We define $\ip_{loc}$ to be the natural pairing on $\sF\oplus \sF^![-1]$, and $Q$ to be $P+P^!$. Here, $\sF^!$ is the sheaf of sections of the bundle $F^\vee\otimes \text{Dens}$.
\end{example}


As another simple example, we have the massless scalar boson theory:
\begin{example}
Let $(N,h)$ be a Riemannian manifold with Riemannian density $dVol_h$.The \textbf{massless scalar boson} is the free BV theory with space of fields $\cinfty_N\oplus \cinfty_N[-1]$, $Q$ the Laplacian on $N$, and with the pairing $\ip_{loc}$ determined uniquely by 
\[
\ip[f_1,f_2]_{loc} : =  f_1f_2 dVol_h,
\]
when $|f_1|=0$ and $|f_2|=1$. To satisfy the appropriate anti-symmetry conditions, we must take the opposite sign if $|f_1|=1$ and $|f_2|=0$. All fields have bosonic statistics.
\end{example}

\subsection{Actions of an Elliptic Differential Graded Lie Algebra on a Free Theory}


We would like to understand what it means for a dgla to act on a free theory. To begin with, we should specify the particular class of dglas suited to the analytic nature of quantum field theory, and the associated notions of Chevalley-Eilenberg cochains. 
\subsubsection{Elliptic Differential Graded Lie Algebras and their Chevalley-Eilenberg Cochains}
\begin{definition}[Definition 7.2.1 of \cite{othesis}]
An \textbf{elliptic differential graded Lie algebra} $\sL$ is a sheaf of sections of a $\Z$-graded vector bundle $L\to N$ which is also a sheaf of dglas. The differential must give $\sL$ the structure of an elliptic complex, and the bracket must be given by a bidifferential operator.
\end{definition}

\begin{example}
For a manifold $M$ and a Lie algebra $\fg$, the sheaf $\Omega^\bullet_M\otimes \fg$ is an elliptic dgla with differential given by the de Rham differential and bracket given by the wedge product on the differential form factor and Lie bracket on the $\fg$ factor.
\end{example}


An elliptic dgla has the standard Chevalley-Eilenberg (CE) cochain complex, as well as a cochain complex of functionals that are more nicely behaved and are defined using the fact that $\sL$ is a sheaf of smooth sections of a vector bundle.

\begin{definition}
Let $\sL$ be an elliptic dgla. The \textbf{Chevalley-Eilenberg cochain complex} of $\sL$ is the commutative differential graded algebra 
\[
C^\bullet(\sL) = \prod_{n=0}^\infty \Sym^n(\sL(N)^\vee[-1])
\]
equipped with the Chevalley-Eilenberg differential defined as in the finite-dimensional case. Here, we use $\vee$ to denote continuous linear dual, and the $\Sym$ is taken with respect to the projective tensor product of Nuclear spaces. Namely, 
\[
\Sym^n(\sL(N)^\vee[-1]) = Hom((\sL(N)[1])^{\otimes n} ,\R)_{S_n},
\]
where the subscript $S_n$ denotes coinvariants of the $S_n$ action which permutes the tensor factors. See Appendix 2 in \cite{cost} and Appendix B.1 in \cite{CG2} for more details. We will often use the notation $\sO(\sL[1])$ to describe the underlying graded vector space.
\end{definition}

\begin{definition}[Definition 7.2.7 of \cite{othesis}]
\label{def: locchains}
The cochain complex $C^\bullet_{loc}(\sL)$ of \textbf{local Chevalley-Eilenberg cochains} of the dgla $\sL$ is the subcomplex of $C^\bullet(\sL)$ given by sums of functionals of the form
\[
F_n(\phi_1,\cdots, \phi_n) = \int_N D_1\phi_1 \cdots D_n\phi_n d\mu,
\]
where each $D_i$ is a differential operator from $\sL$ to $\cinfty_N$ and $d\mu$ is a smooth density on $N$. We view $F_n$ as a coinvariant of the $S_n$ action. We can also study $C^{\bullet}_{loc,red}(\sL)$, which is the quotient of $C^\bullet_{loc}(\sL)$ by the subcomplex $\R$ (which lives in $\Sym$-degree 0). We also have $C^\bullet_{red}(\sL)$, defined analogously.
\end{definition}

\begin{remark}
Here, we are using that $N$ is compact to be able to integrate densities. However, to prove our main theorem, it is necessary to extend $C^\bullet_{loc}(\sL)$ to a sheaf on $N$. On open subsets of $N$, we would like to define $C^\bullet_{loc}(\sL)(U)$ to be the space of all Lagrangian densities (an idea made precise below). Densities cannot in general be integrated over non-compact $U$. Thus, in this case, we cannot think of local cochains on $U$ as living inside the space of all CE cochains.
\end{remark}


In light of the previous remark, we would like to have an alternative characterization of the local Chevalley-Eilenberg cochain complex that is naturally sheaf-like. To do this, we introduce the language of jets:
\begin{definition}[Cf. Section 5.6.2 of \cite{cost}]
\label{def: jets}
The \textbf{sheaf of jets} $J(L)$ of the elliptic dgla $\sL$ is the sheaf of sections of the vector bundle of dglas over $N$ whose fiber at a point $x\in N$ is the space of formal germs at $x$ of sections of $L$.
\end{definition}

The dgla structure on the fiber is induced from the same stucture on $L$ and the fact that the dgla operations of $\sL$ are polydifferential operators. Moreover, as discussed in Section 6.2 of Chapter 5 of \cite{cost}, $J(L)$ has a natural (left) $D_N$-module structure, where $D_N$ is the sheaf of differential operators on $N$. Furthermore, $J(L)$ is an inverse-limit of finite-dimensional $\cinfty_N$-modules; we can therefore endow it with the topology of the inverse limit. 

Now, let $J(L)^\vee$ denote the following $\cinfty_N$-module
\begin{equation}
\label{eq: jet dual}
\Hom_{\cinfty_N}(J(L), \cinfty_N),
\end{equation}
where $\Hom_{\cinfty_N}$ means \textit{continuous} homomorphisms of $C^\infty_N$ modules (i.e., those that respect the inverse-limit topology on $J(L)$). Let also 
\begin{equation}
\label{eq: jetce}
\sO_{red}(J(L)) : = \prod_{n>0} \Sym^n_{\cinfty_N}\left( J(L)^\vee\right).
\end{equation}
Note that $\sO(J(L))$ has a natural differential encoding the Chevalley-Eilenberg differential on $C^\bullet_{loc}(\sL)$. Moreover, $\sO(J(L))$ has the structure of a left $D_N$-module.


\begin{lemma}[Lemma 6.6.1 of Chapter 5 of \cite{cost}]
\label{lem: localCE} 
For $\sL$ an elliptic dgla, there is a canonical isomorphism of cochain complexes
\[
C^\bullet_{loc,red}(\sL)\cong \text{Dens}_N(N)\otimes_{D_N}\sO_{red}(J(L))(N).
\]
Here, $\text{Dens}_N$ is given the right $D_N$-module structure induced from considering densities as distributions.
\end{lemma}

\begin{remark}
The object on the right hand side of the isomorphism of Lemma \ref{lem: localCE} is manifestly the space of global sections of the sheaf $\text{Dens}_N\otimes_{D_N}\sO_{red}(J(L))$. We therefore will, in the sequel, abuse notation and use $C^\bullet_{loc,red}(\sL)$ to refer also to this sheaf. 
\end{remark}

\begin{lemma}[Lemma 6.6.2 of Chapter 5 of \cite{cost}]
 \label{lem: derivedlocalCE}
For $\sL$ an elliptic dgla, there is a quasi-isomorphism
 \[
 C^\bullet_{loc,red}(\sL)\simeq \text{Dens}_N\otimes_{D_N}^{\mathbb L}\sO_{red}(J(L)).
 \]
Putting this together with the previous lemma, we see that the actual tensor product is a model for the derived tensor product.
\end{lemma}

We use this Lemma in the proof of Proposition \ref{prop: qism}.

\subsubsection{Actions of an Elliptic dgla on a Free Theory}

In this subsubsection, we review the main results concerning the action of an elliptic dgla on a free BV theory. However, most of what we do works in a far more general context of elliptic $L_\infty$-algebras acting on a general BV theory. We will point out the few instances in which results are specific to the case we consider. 


\begin{definition}[Special case of Definition 11.1.2.1 of \cite{CG2}]
If $\sL$ is an elliptic dgla, and $\E$ is an elliptic complex, then \textbf{an action of $\sL$ on $\E$} is the structure of an elliptic dgla on $\E[-1]\oplus \sL$ such that the maps in the linear exact sequence 
\[
0\to \E[-1] \to \E[-1]\oplus \sL\to \sL \to 0,
\]
as well as the map $\sL\to \E[-1]\oplus \sL$, are openwise dgla maps (i.e. they commute with brackets). We are thinking of $\E[-1]$ as an elliptic abelian dgla.
\end{definition}

In other words, an action of $\sL$ on $\E$ is just a map 
\[
\rho: \sL\otimes \sE \to \sE
\]
satisfying certain coherence relations encoding the fact that the differential on $\E[-1]\oplus \sL$ is a derivation for $\rho$ and that the bracket on $\sE[-1]\oplus \sL$ satisfies the Jacobi identity. Equivalently, we require that $\sE$ be an $\sL$-module. We will also use $[\cdot,\cdot]$ to denote $\rho$ and $\sL\ltimes \sE[-1]$ to denote the resulting dgla.


Finally, we define what it means for an elliptic dgla to act on a free BV theory:
\begin{definition}[Definition 11.1.2.1 of \cite{CG2}]
Let $(\sF, Q,\ip)$ be a free BV theory and $\sL$ an elliptic dgla. An \textbf{action of $\sL$ on $(\sF,Q,\ip)$} is an action of $\sL$ on the elliptic complex $(\sF,Q)$ such that the map $\rho$ leaves the pairing $\ip$ invariant in the sense that $\rho(\alpha, \cdot)$ is a graded derivation with respect to $\ip$ for all $\alpha\in \sL$. This means that
\[
\ip[{[\alpha, f_1]},f_2] +(-1)^{|f_1||\alpha|}\ip[f_1,{[\alpha, f_2]}]=0
\] 
for all $f_1,f_2\in \sF$. Moreover, we require the action of $\sL$ on $\sF$ to be even with respect to the $\Z/2$ grading of $\sF$ in the sense that 
\[
\pi_{[\alpha, f_1]}=\pi_{f_1}
\]
for all $f_1\in \sF$ and $\alpha \in \sL$.
\end{definition}

\begin{remark}
The requirement that the differential on $\sF[-1]\rtimes \sL$ act as a derivation for $[\cdot, \cdot]$ implies that $H^\bullet \sL$ acts on $H^\bullet \sF[-1]$. This is the sense in which we should think of a dgla action as a symmetry; since the cohomology of $\sF[-1]$ is to be thought of as the moduli space of classical solutions, the fact that $H^\bullet \sL$ acts on $H^\bullet \sF[-1]$ is a precise way of saying that $\sL$ acts on $\sF[-1]$ in a way that (cohomologically) preserves the equations of motion.
\end{remark}

Returning now to the general context of an elliptic dgla acting on a free theory $\sF$, we note that a free BV theory provides the quadratic action functional 
\[
S(\phi) = \frac{1}{2}\ip[\phi, Q\phi]
\]
and the action of $\sL$ on $\sF$ provides the following cubic deformation of the action functional on $\sL[1]\oplus \sF$:
\begin{equation}
\label{eq: interaction}
I(X,\phi) = \frac{1}{2} \ip[\phi, \rho(X,\phi)].
\end{equation}
The complex $C^\bullet_{loc}(\sF[-1])$ has a degree +1 Poisson bracket $\{\cdot, \cdot\}$ induced from $\ip$, whose precise construction can be found in Section 5.3 of \cite{cost}. The bracket $\{\cdot, \cdot\}$ can be extended to all of $C^\bullet_{loc}(\sF[-1]\rtimes \sL)$ by $C^\bullet(\sL)$-linearity. Moreover, as discussed in Section 11.1 of \cite{CG2}, we can equivalently think of an action of $\sL$ on $\sF$ as a degree-zero element
\begin{equation}
\label{eq: mcinteraction}
I\in \sO_{loc}(\sF\oplus \sL[1])/(\sO^\bullet_{loc}(\sL[1])\oplus C^\bullet_{loc}(\sF)),
\end{equation}
satisfying the Maurer-Cartan Equation 
\begin{equation}
\label{eq: cme}
(d_{\sL}+Q) I+\frac{1}{2}\{I,I\}=0,
\end{equation}
which is also known as the \textbf{classical master equation}. In our special case, we require $I$ to have $\Sym$-degree 1 with respect to $\sL$ and $\Sym$-degree 2 with respect to $\sF$. If $I$ has the correct $\Sym$-degrees and satisfies the classical master equation, then it determines the action of $\sL$ on $\sF$ and vice versa. Furthermore, we consider $\sL$ to be the space of background fields because we have not demanded that $\sL$ be endowed with a pairing as $\sF[-1]$ is.

Notice that the interaction $I(X,\phi)$, though of cubic-and-higher order with respect to $X$ and $\phi$ together, is only quadratic in $\phi$. This is particular to our case; if we were to use the full language of $L_\infty$-algebras, we would be able to describe deformations of $\sF$ into non-free theories, in which case we would allow cubic and higher-order interactions in $\sF$ fields.

Graphically, the interaction $I$ can be depicted by a vertex with two straight half-edges and one wavy half-edge. See Figure \ref{fig: generalvertex}. The perspective on $I$ as such a vertex anticipates the formalism of Feynman diagrams. See also Remark \ref{rem: feyn}.

\begin{figure}[h]
\centering
\includegraphics[width=.25\textwidth]{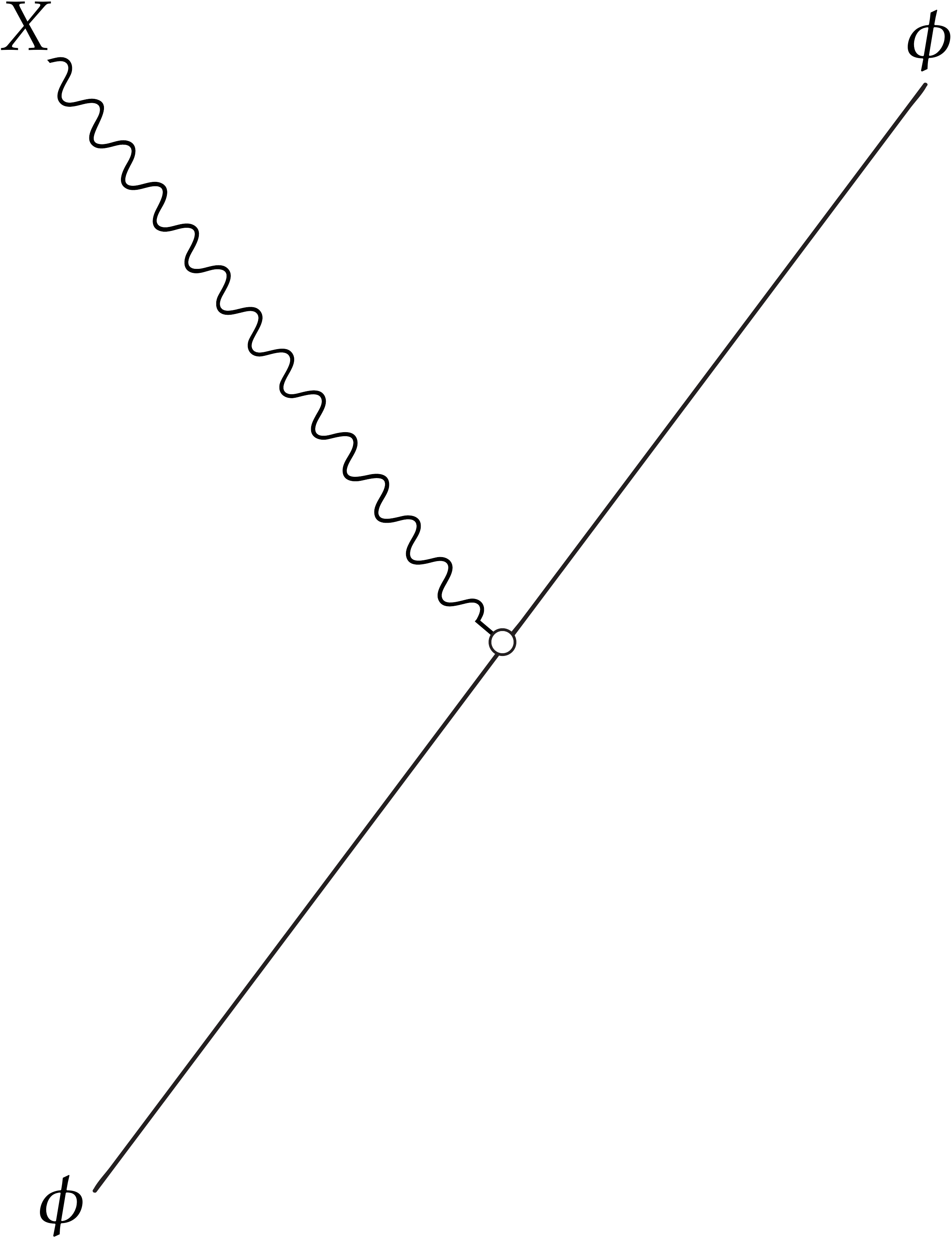}
\caption{A vertex depicting the interaction $I$.}
\label{fig: generalvertex}
\end{figure}

\begin{remark}
If $(\sF,P)$ is an elliptic complex, and $\sL$ acts on the cotangent theory to $\sF$ as an extension of an action of $\sL$ on $\sF$, then the interaction $I$ can be rewritten as 
\[
I(X,\varphi,\psi) = \ip[\varphi,{[X,\psi]}]
\]
where $\varphi\in \sF$ and $\psi \in \sF^![-1]$ are viewed as the base and fiber components of a general element of $T^*[-1]\sF$. The factor of 1/2 has disappeared because of the symmetry properties of $[X,\cdot]$. In this case, we can add arrows to the vertex to indicate that, with respect to the decomposition of $T^*[-1]\sF$ into base and fiber, $I$ is non-zero only when evaluated on a base field as one input and a fiber field as another. Figure \ref{fig: vertex} shows this modified graphical interpretation.
\end{remark}

\begin{figure}[h]
\centering
	\includegraphics[width = .25\textwidth]{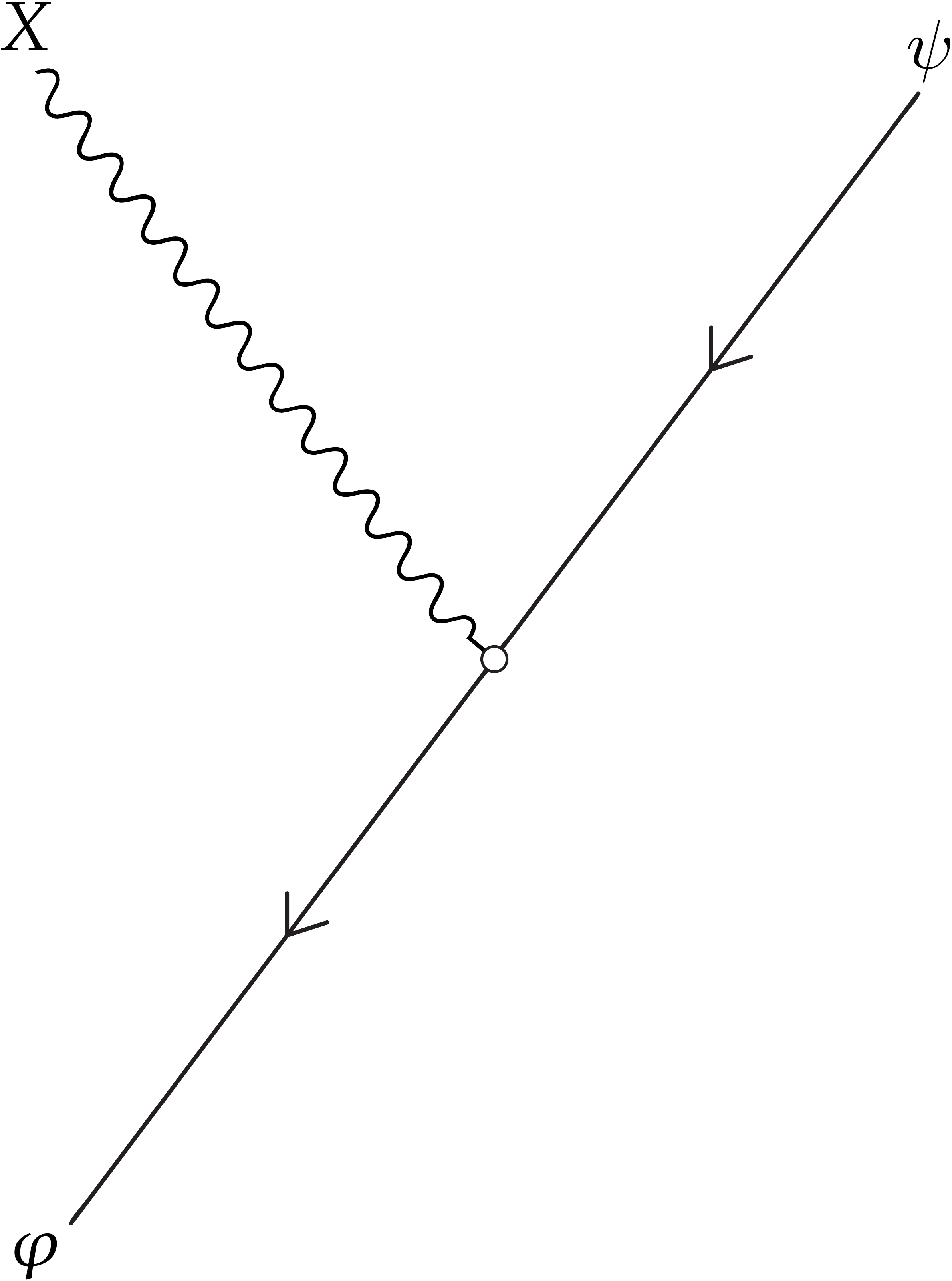}
	\caption{The vertex corresponding to the interaction $I$ for a cotangent theory.}
	\label{fig: vertex}
\end{figure}	


\begin{definition}[Proof of Proposition 11.3.0.1 in \cite{CG2}]
If $\sL$ acts on the free field theory $(\sF,Q,\ip)$, then the cochain complex of \textbf{global equivariant classical observables} $\Obcl$ is $C^\bullet(\sF[-1]\rtimes \sL)$, the (completed) differential graded commutative algebra of Chevalley-Eilenberg cochains on the dgla $(\sF[-1]\rtimes\sL)(N)$. We will use $\sO(\sL[1]\oplus \sF)$ to denote the underlying graded vector space.
\end{definition}

\begin{remark}
Recall that Chevalley-Eilenberg cochains of a dgla $\sM$ are, as a vector space, $\widehat{\Sym}( \sM^\vee[1])$; since we let $\sF$ be $\Z\times \Z/2$ graded, we view $\sL(N) \oplus \sF[-1](N)$ as a $\Z\times \Z/2$ graded object and the symmetrization is with respect to the Koszul sign rules for $\Z\times \Z/2$ graded objects.
\end{remark}

\subsection{Equivariant Quantization}
\label{subsec: eqquant}
In this subsection, we continue the notation of the previous subsection, letting $\sL$ act on the free classical theory $\sF$. Since $\ip$ gives a sort of $\sL$-equivariant symplectic structure on $\sF$, we would hope that it induces a Poisson bracket and Laplacian $\Delta$ on $\Obcl$, as can be done in the finite-dimensional case. If this were the case, we could define $\Obq$ to be $\sO^\bullet(\sL[1]\oplus \sF)[\![\hbar]\!]$ but with the differential from $\Obcl$ deformed by the term $\hbar \Delta$. This is what is considered BV quantization in the finite-dimensional case.  However, the Poisson bracket is only defined if one of the arguments is local, and the Laplacian is undefined. This is because the na\"ive way to define these operations involves pairing distributions with each other, an operation which is ill-defined in general. 

The solution to this problem is effective field theory, whose essence, as developed in \cite{cost} (especially Chapters 2 and 5), is to replace the interaction $I$ with a family $I[t]$, one for each $t\in \R_{>0}$, and to correspondingly define a Poisson bracket $\{\cdot, \cdot\}_t$ and BV Laplacian $\Delta_t$ for every $t$ (see also \cite{WILSONKOGUT} and Chapter 12.1 of \cite{PeskinSchroeder} for physical perspectives on some of these ideas.) We can then define $\Obq{}[t]$ as a graded vector space to be  $\Obcl{}[\![\hbar]\!]$, but with putative differential $Q+d_{\sL}+\{I[t],\cdot\}_t+\hbar \Delta_t$. However, the putative differential may not square to zero. We will first address the procedure of renormalization, then the question of the existence of the desired differential.

\subsubsection{Renormalization and the BV Formalism} 
We assume that $\sF$ has been equipped with a gauge-fixing $Q^{GF}$; then $H:=[Q,Q^{GF}]$ is a generalized Laplacian, and so $H$ has an integral heat kernel $k_t$. We will actually use a heat kernel more suited to the study of BV theories:


\begin{definition}[Section 8.3 of Chapter 5 of \cite{cost}]
Let $t\in \R_{>0}$. The \textbf{scale $t$ BV heat kernel} $K_t$ is the unique degree 1 element of $\sF\otimes \sF$ satisfying 
\[
-id\otimes \ip (K_t\otimes f) = e^{-tH}f
\]
for all $f\in \sF$. The non-degeneracy property of the pairing $\ip$, along with the existence of the heat kernel $k_t$, guarantees that $K_t$ exists.
\end{definition} 


\begin{definition}[Section 9.1 of Chapter 5 of \cite{cost}]
\label{def: bvlaplacian}
The \textbf{scale $t$ BV Laplacian $\Delta_t$} is the operator $-\partial_{K_t}: \sO(\sL[1]\oplus \sF)\to \sO(\sL[1]\oplus \sF)$, where $\partial_{K_t}$ is the unique order-2 differential operator on $\sO (\sL[1]\oplus \sF)$ which is 0 on $\Sym^{<2}$ and which on $\Sym^2$ is just contraction with $K_t$.
\end{definition}


\begin{definition}[Section 9.2 of Chapter 5 of \cite{cost}]
The \textbf{scale $t$ Poisson bracket $\{\cdot,\cdot\}_t$} is the map 
\[
\sO(\sL[1]\oplus \sF)\times \sO(\sL[1]\oplus \sF) \to \sO(\sL[1]\oplus \sF)
\]
given by 
\[
\{J,J'\}_t = \Delta_t(JJ')- \Delta_t(J)J' - (-1)^{|J|}J\Delta_tJ'.
\]
\end{definition}
In other words, the Poisson bracket measures the failure of the BV Laplacian to be a derivation for the commutative algebra $\sO(\sL[1]\oplus \sF)$. Figure \ref{fig: bracketlaplacian} gives a schematic depiction of the Poisson bracket and BV Laplacian. Notice that, based on the diagrammatic depiction of $\Delta_t$, $\{J,J'\}_t$ is 0 if either $J$ or $J'$ has a dependence only on $\sL$ and not on $\sF$. The same is true for $\Delta_t J$. This is because $K_t$ has components only in $\sF\otimes\sF$ as a consequence of the fact that we have not required $\sL$ to have a gauge-fixing or a pairing $\ip$. These facts about $\Delta_t$ and $\{\cdot,\cdot\}_t$ simplify our analysis dramatically.

\begin{figure}[h]
\centering
\begin{subfigure}[b]{0.6\textwidth}
	\includegraphics[width = \textwidth]{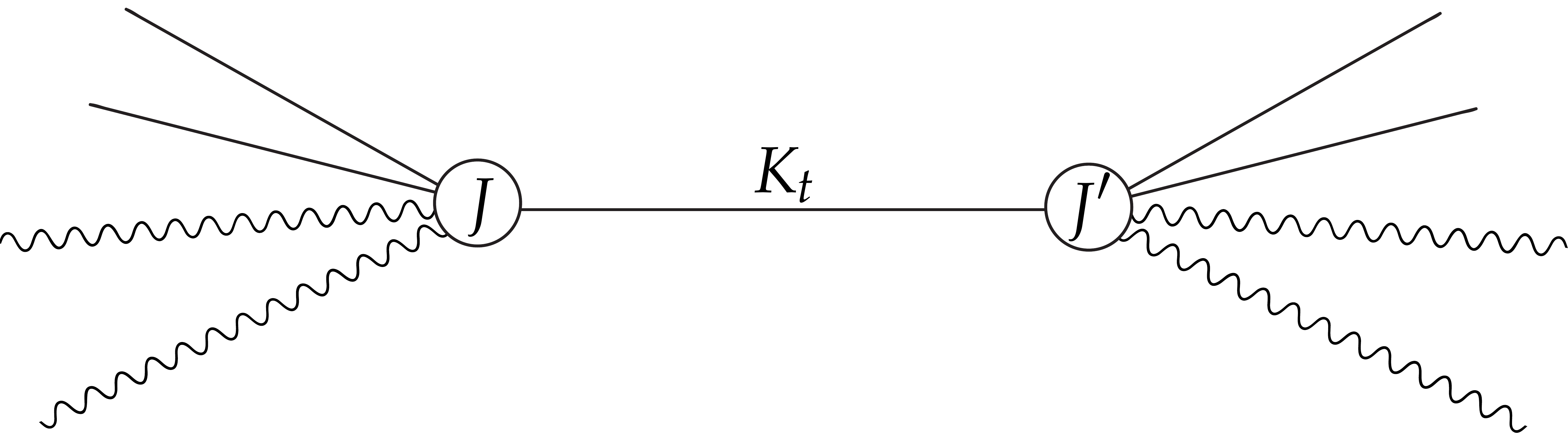}
	\caption{A diagrammatic depiction of a term in the Poisson bracket $\{J,J'\}$. The Poisson bracket is the graded sum over all ways of contracting out one edge from $J$ with one edge from $J'$ with $K_t$.}
\end{subfigure}
\\
\begin{subfigure}[b]{0.3\textwidth}
	\includegraphics[width = \textwidth]{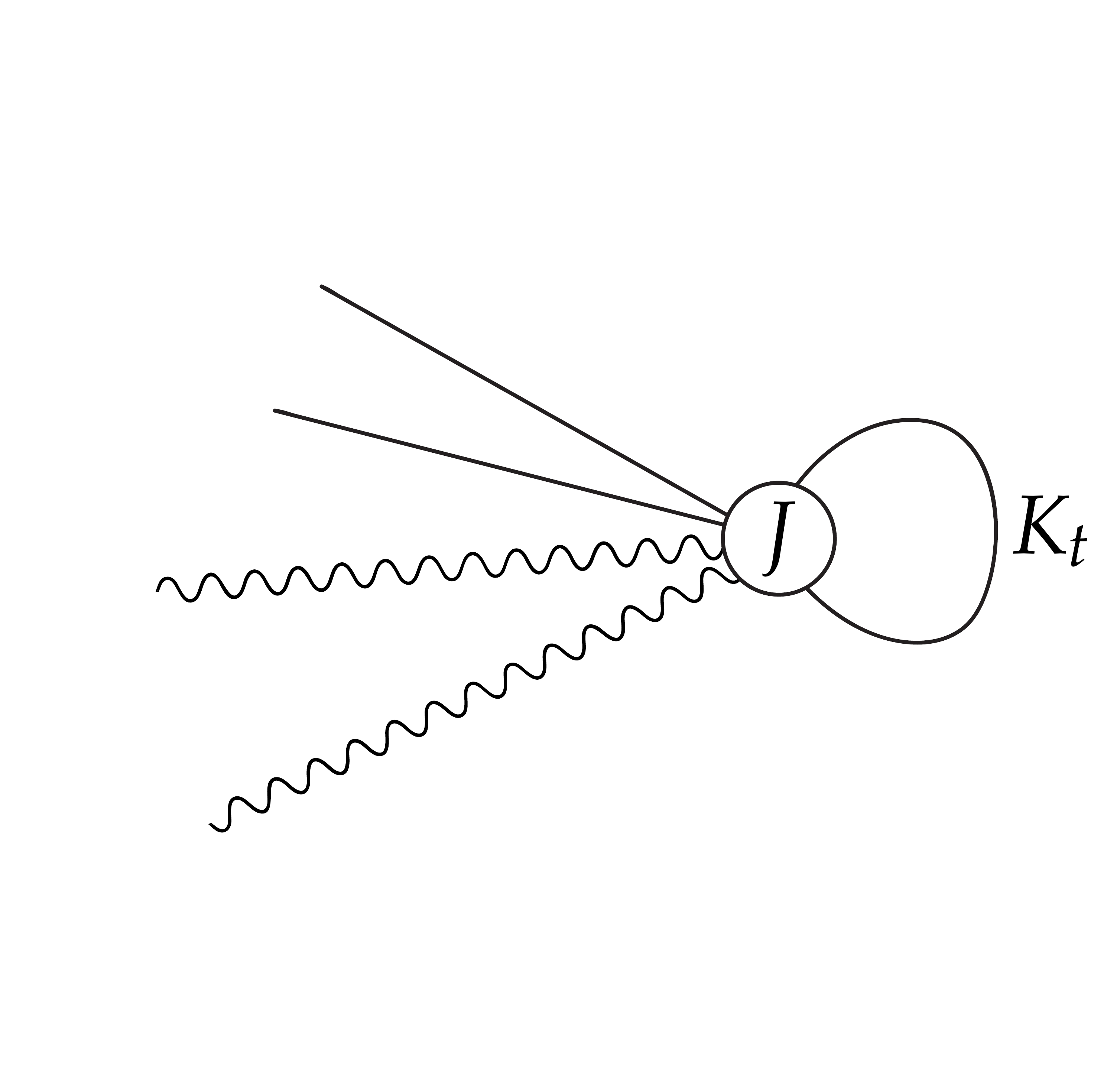}
	\caption{A diagrammatic depiction of a term in the BV Laplacian. It is a graded sum over all ways of contracting out pairs of edges in $J$ with $K_t$.}
\end{subfigure}
\caption{Schemata of the Poisson bracket and BV Laplacian.}
	\label{fig: bracketlaplacian}
\end{figure}

The following lemma is useful for computational reasons. It can be verified by explicit computation.
\begin{lemma}
\label{lem: propsofBVLaplacian}
$\Delta_t$ has the following two properties:
\begin{enumerate}
\item $\Delta_t$ squares to zero.
\item $\Delta_t$ is a derivation for the Poisson bracket $\{\cdot,\cdot\}_t$.
\end{enumerate}
\end{lemma}

We have addressed two of the three scale-$t$ objects that are needed for our purposes: we just need to explain how $I$ gives rise to the family $I[t]$. This is by far the most subtle part of renormalization, and we refer the reader to \cite{cost} for the full treatment of this subject. We will content ourselves to present the results of the renormalization procedure as applied in the case of $\sL$ acting on a free theory. To do this, we first need the following few definitions:


\begin{definition}[Section 8.3 of Chapter 5 of \cite{cost}]
\label{def: propagator}
Let $t'>t>0$. The \textbf{propagator from scale $t$ to $t'$} is denoted $P(t,t')$ and is the unique element of $\sF\otimes \sF$ satisfying
\[
id\otimes\ip (P(t,t')\otimes f) = \left(Q^{GF}\int_t^{t'} e^{-s[Q,Q^{GF}]}ds \right)f,
\]
for $f\in \sF$.
\end{definition}
This tells us that, in particular, $P(t,t')$ is smooth so long as both $t$ and $t'$ are positive. In the limit as $t\to0$, $P(t,t')$ becomes distributional. On occasion, we need to consider $P(0,t)$.


\begin{definition}[Section 13.4 of Chapter 2 of \cite{cost}]
The \textbf{RG flow operator from scale $t$ to scale $t'$} is the following operator on $\sO(\sL[1]\oplus \sF)[\![\hbar]\!]$
\[
W(P(t,t'),\cdot) : J\mapsto \hbar \log \left( \exp\left( \hbar \partial_{P}\right) \exp\left(J/\hbar\right)\right),
\]
where $P$ is the operator whose kernel is $P(t,t')$. $\partial_{P}$ is defined analogously to $\partial_{K_t}$ (see Definition \ref{def: bvlaplacian}).
\end{definition}

\begin{remark}
\label{rem: feyn}
The RG flow operator has a natural interpretation in terms of Feynman diagrams, and we refer the reader to \cite{cost} Chapter 2.3 for details. We comment only that we use the word ``leg'' where \cite{cost} uses the word ``tail'' to describe external edges of Feynman diagrams, i.e. edges that end at a univalent vertex. 
\end{remark}

As mentioned above, we would like to replace $I$ with a family of interactions $\{I[t]\}$, one for each $t>0$; however, we will do so in a way that the $I[t]$ are related by the RG flow equation $I[t'] = W(P(t,t'),I[t])$. This is the main purpose of the operator $W(P(t,t'), \cdot)$, but we will have other uses for it in the sequel, e.g. the following


\begin{definition}[Definition 7.5.1 of \cite{othesis}]
The \textbf{tree-level, scale-$t$ interaction} is the element $I_{tr}[t]$ of $C^\bullet(\sF[-1]\rtimes \sL)$ given by
\[
I_{tr}[t]: = \lim_{t'\to 0}\left( W(P(t',t),I)\mod \hbar\right),
\]
where $I$ is the interaction in Expression \ref{eq: interaction}.
\end{definition}

\begin{remark}
The limit exists, though we do not prove it here. See section 5 of Chapter 2 of \cite{cost} for details.
\end{remark}

\begin{remark}
This is called the \textit{tree-level} interaction because the Feynman diagrams which represent it are all trees.
\end{remark}

\begin{figure}[h]
\centering
\includegraphics[scale = 0.25]{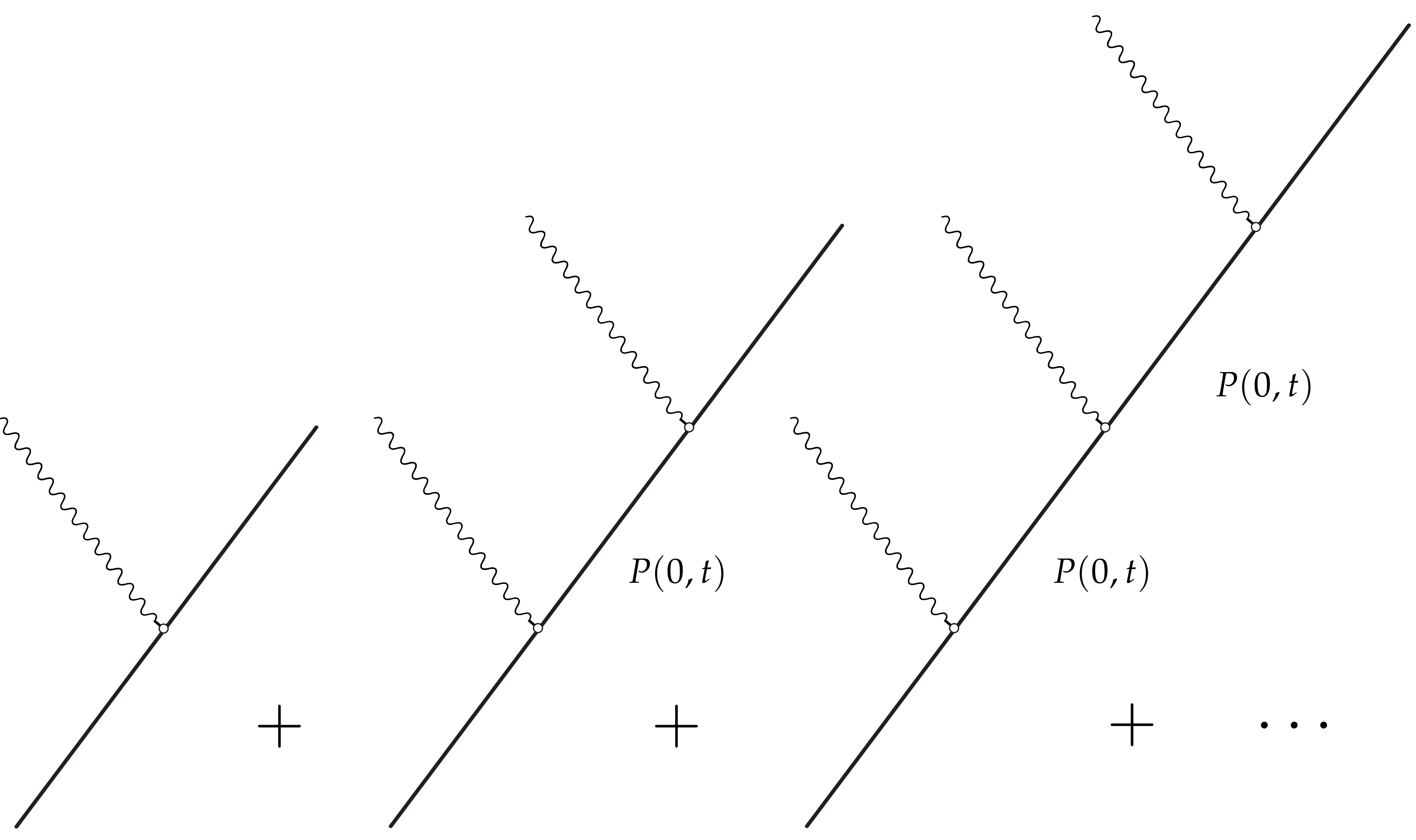}
\caption{The tree-level diagrams contributing to $I_{tr}[t]$.} 
\label{fig: trees}
\end{figure}
Since the propagator $P(t,t')$ connects only $\sF$ edges, all of the trees contributing to $I_{tr}$ have only two $\sF$ legs. Thus, $\Delta_t I_{tr}$ belongs to $\widehat{Sym}(\L[1]^\vee)$. 

If the limit as $t\to 0$ of $W(P(t,t'),I)$ (not just the $\hbar^0$ part) existed, we wouldn't need renormalization; this is, however, not the case. So, we must take the renormalized limit of the operator $W$. The precise procedure for doing this in general is spelled out in Chapter 2 of \cite{cost}, and the result of this procedure in the case at hand is discussed in \cite{othesis}. We note only that in the case of an action of $\sL$ on a free theory, this procedure produces a family $I_{wh}[t]$ of interactions schematized by Figure \ref{fig: wheel} (the subscript $_{wh}$ stands for ``wheel,'' describing the corresponding Feynman diagrams). We do not need to know much about $I_{wh}[t]$; it suffices for our purposes to know that it's an element of $C^\bullet_{red}(\sL)$. With these data in place, we can describe the promised family $I[t]$.

\begin{figure}[h]
\centering
\includegraphics[scale =.17]{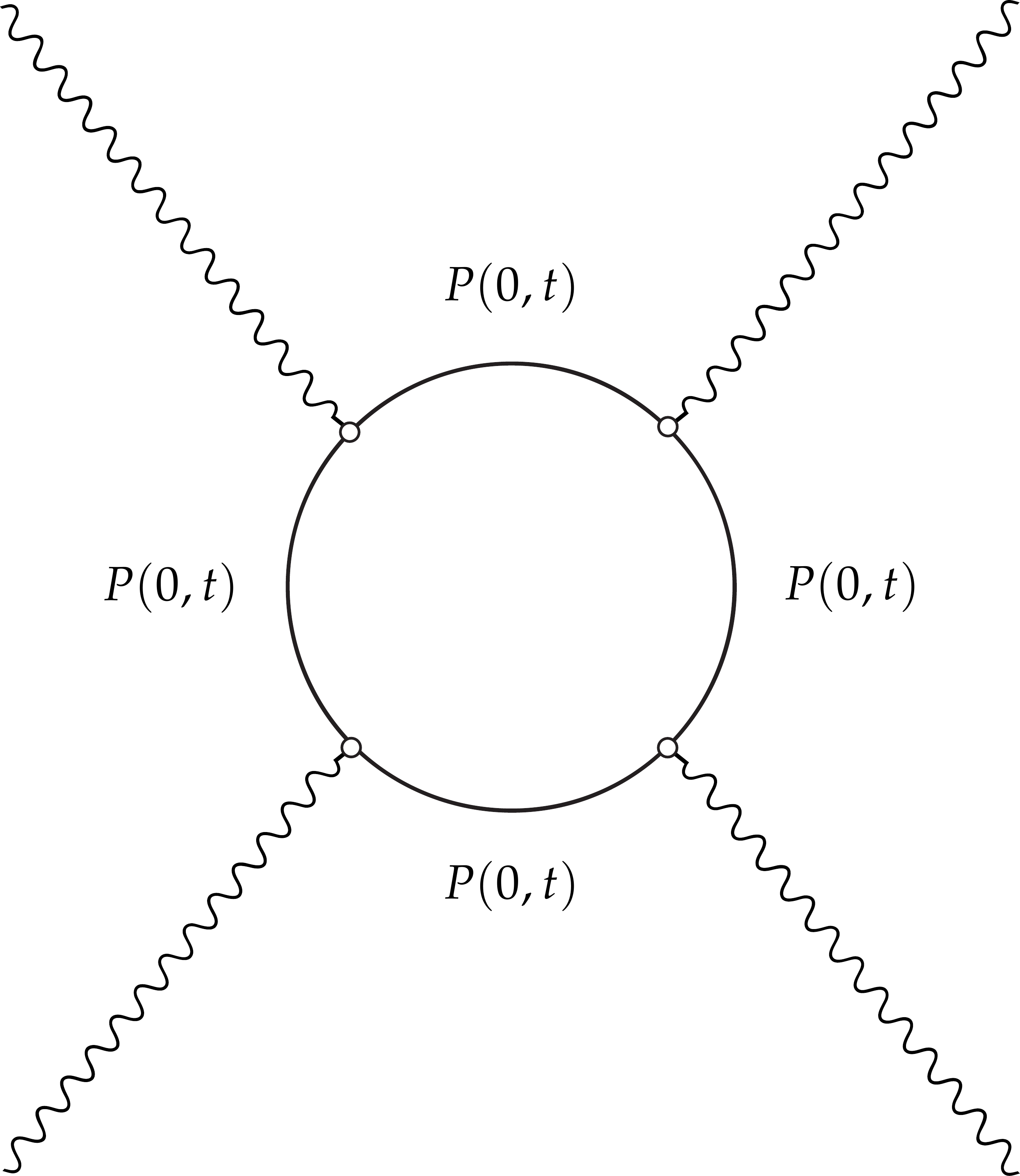}
\caption{A graphical depiction of one of the terms appearing in $I_{wh}[t]$. Not depicted is the diagram for the counter-term necessary to make this wheel diagram finite.}
\label{fig: wheel}
\end{figure}


\begin{definition}[Lemma 7.5.3 of \cite{othesis}]
The \textbf{scale-$t$ renormalized interaction for $\sL$ acting on $\sF$} is denoted $I[t]$ and given by 
\[
I[t] = I_{tr}[t] + \hbar I_{wh}[t]
\]
\end{definition}

\begin{remark}
The renormalized interaction takes this form only in the context of an elliptic dgla (or more generally an elliptic $L_\infty$-algebra) deforming a free theory into a family of free theories. If the original theory is interacting, or if the $L_\infty$-algebra deforms free theories into interacting ones, the structure of $I[t]$ becomes more complicated, and in particular includes contributions at all powers of $\hbar$.
\end{remark}

\subsubsection{The Quantum Master Equation and the Obstruction}
The essence of BV quantization is the introduction of a deformation of the Chevalley-Eilenberg differential on $\sO(\sL[1]\oplus \sF)[\![\hbar]\!]$. Indeed, we would like to define the cochain complex $\Obq{}[t]$, which is, as a vector space $\sO(\sL[1]\oplus \sF)[\![\hbar]\!]$, but which has as differential $d_{t}:=d_{\sF[-1]\rtimes \sL}+\hbar \Delta_t$. The only issue is that this differential may not square to zero. In fact,
\begin{lemma}
The square of the putative differential $d_t=d_{\sF[-1]\rtimes \sL}+\hbar \Delta_t$ on $\Obq{}[t]$ satisfies the equation
\[
(d_{\sF[-1]\rtimes \sL}+\hbar \Delta_t)^2 = \left\{ (Q+d_{\sL})I[t]+\frac{1}{2}\{I[t],I[t]\}_t+\hbar \Delta_t I[t], \cdot\right\}_t.
\]
\end{lemma} 
\begin{proof}
Direct computation, using the fact that $\Delta_t$ is a derivation for $\{\cdot,\cdot\}_t$.
\end{proof}

Therefore, $d_t^2=0$ if and only if 
\[
O[t]:=(Q+d_{\sL})I[t]+\frac{1}{2}\{I[t],I[t]\}_t+\hbar \Delta_t I[t]
\]
is in the $\{\cdot, \cdot\}_t$-center of $\Obq{}[t],$ or equivalently if $ O[t]$ lies in $C^\bullet(\sL)$.

We will see below that $d_t^2$ will always  be 0 for an action of $\sL$ on a free theory. However, we would like to demand something stronger, which we codify in the second part of the following definition.


\begin{definition}
The \textbf{weak scale $t$ quantum master equation (wQME)} is 
\[
(d_{ \sL}+Q+\{I[t],\cdot\}_t+\hbar \Delta_t)^2=0.
\]
The \textbf{strong scale $t$ quantum master equation (sQME)} is 
\[
(Q+d_{\sL})I[t]+\frac{1}{2}\{I[t],I[t]\}_t+\hbar \Delta_t I[t]=0.
\]
\end{definition}

\begin{definition}
If $I[t]$ satisfies the wQME, then the \textbf{cochain complex of global equivariant quantum observables} is 
\[
\Obq{}[t]:=\left(\sO(\sL[1]\oplus \sF),d_{ \sL}+Q+\{I[t],\cdot\}_t+\hbar \Delta_t\right).
\]
Multiplication by elements of $C^\bullet(\sL)$ in the obvious way naturally endows $\Obq{}[t]$ with the structure of a $C^\bullet(\sL)$-module.
\end{definition}

If the generalized Laplacian $[Q,Q^{GF}]$ is formally self-adjoint and positive for some metric on $F$, then we can also define scale $\infty$ observables because the operator $\exp\left( -t[Q,Q^{GF}]\right)$ is bounded and has a $t\to\infty$ limit.

\begin{remark}
In the language of \cite{CG2}, the wQME is the equation defining an action of $\sL$ on the quantum field theory of $\sF$, and the sQME is the equation defining an \textit{inner} action of $\sL$ on $\sF$.
\end{remark}

The following lemma tells us that the sQME interacts as we would like with the RG flow operator. It is Lemma 9.2.2 of Chapter 5 of \cite{cost}.

\begin{lemma}
\label{lem: rgqme}
If $I$ satisfies the scale $t$ sQME, then $W(P(t',t),I)$ satisfies the scale $t'$ sQME.
\end{lemma}

Let us examine the failure of the sQME to be satisfied in a bit more detail. Since $I$ was required to satisfy $(d_{\sL}+Q)I+\frac{1}{2}\{I,I\}=0$, the general machinery of renormalization implies that the $\hbar^0$ part of the sQME is satisfied (see, e.g. Lemma 9.4.1. of Chapter 5 of \cite{cost}). Thus, we make the following definition:


\begin{deflem}[Cf. Corollary 5.11.1.2 of \cite{cost}]
\label{def: tobstr}
The \textbf{scale $t$ obstruction to the $\sL$-equivariant quantization of $\sF$} is 
\[
\Obstr[t]:= \frac{1}{\hbar} \left( QI[t]+d_{\sL}I[t]+\frac{1}{2}\{I[t],I[t]\}_t+\hbar \Delta_t I[t]\right).
\]
The obstruction is a closed, cohomological degree 1 element of $\Obq{[t]}$.
\end{deflem}

The following is Lemma 7.5.4 of \cite{othesis}. It gives us an explicit formula for the scale $t$ obstruction.

\begin{lemma}
\label{lem: formofobstr}
For the action of an elliptic dgla on a free theory,
\[\Obstr[t] = (\Delta_t I_{tr}[t]+d_{\sL}I_{wh}[t]).\]
$\Obstr[t]$ depends only on $\sL$, i.e. is also a closed, ghost number 1 element of $C^\bullet_{red}(\sL)\subset \Obq{}[t]$.
\end{lemma}
\begin{proof}
Recall that $I[t] = I_{tr}[t]+\hbar I_{wh}[t]$, where $I_{wh}[t]\in C^\bullet_{red}(\sL)$, so that $QI_{wh}[t]=\{I_{wh}[t],\cdot\}_t=\Delta_t I_{wh}[t]=0$. Moreover, by the commentary preceding the definition of the obstruction, $QI_{tr}[t]+d_{\sL}I_{tr}[t]+\frac{1}{2}\{I_{tr}[t],I_{tr}[t]\}_t=0$. The first statement of the lemma then follows by direct computation. To see that $\Obstr[t]\in C^\bullet_{red}(\sL)$, recall that $I_{tr}[t]$ has Sym-degree two with respect to $\sF$ inputs, so that $\Delta_tI_{tr}[t]$ has Sym-degree zero with respect to $\sF$ inputs, so that $\Obstr[t]\in C^\bullet(\sL)$. Moreover, all terms in $\Delta_tI_{tr}[t]$ and $d_\sL I_{wh}$ have at least linear dependence on $\sL$ inputs, so that $\Obstr[t]\in C^\bullet_{red}(\sL)$. 
\end{proof}

\begin{remark}
This form for the obstruction is particular to families of free theories.
\end{remark}

By Lemma \ref{lem: formofobstr}, $I[t]$ satisfies the wQME, since $\Obstr[t]\in C^\bullet(\sL)$ and all elements of $C^\bullet(\sL)$ are in the $\{\cdot, \cdot\}_t$-center of $\Obq{}[t]$.

We have also the following lemma, which describes the relationship between the obstruction at various length scales. It is a modification of Lemma 11.1.1 from Chapter 5 of \cite{cost} to the case where the obstruction lives only at order $\hbar$, which we know to be the case by the previous lemma.

\begin{lemma}[Lemma 11.1.1 of Chapter 5 of \cite{cost}]
\label{lem: obstrtangent}
Let 
\[
W_{t,t'}: \sO(\sL[1]\oplus \sF) \to \sO(\sL[1]\oplus \sF)
\]
be the map defined by
\[W_{t,t'}(J) = W(P(t,t'),J) \mod \hbar.\]
Then, if $\epsilon$ is a parameter of cohomological degree $-1$ and square 0, we have 
\[
I_{tr}[t'] + \epsilon\Obstr[t'] = W_{t,t'}\left(I_{tr}[t]+\epsilon\Obstr[t]\right).
\] 
\end{lemma}

This lemma gives precise meaning to the idea that the collection $\{\Obstr[t]\}$ lies in the tangent space to the classical theory described by $\{I_{tr}[t]\}$. Moreover, it is discussed in \cite{cost} that $\lim_{t\to0}W_{t,t'}(I_{tr}[t])$ exists and is local, so we have the following 

\begin{lemma}[Corollary 11.1.2 of Chapter 5 of \cite{cost}]
\label{lem: obstr}
The limit
\[
\Obstr : = \lim_{t\to 0} \Obstr[t]
\]
exists and is a closed, degree 1 element of $C^\bullet_{loc,red}(\sL)$.
\end{lemma}
\bibliography{AxialAnomalyintheBVFormalism}
\bibliographystyle{spmpsci}
\end{document}